\documentclass[sigplan,10pt]{acmart}
\renewcommand\footnotetextcopyrightpermission[1]{}
\makeatletter
\renewcommand\@ACM@checkaffil{}
\makeatother
\usepackage[inline]{enumitem}
\usepackage{algpseudocode}
\usepackage{adjustbox}
\usepackage[most]{tcolorbox}
\usepackage{colortbl}
\usepackage{placeins}
\usepackage{xspace}
\usepackage{booktabs}

\newtcolorbox[auto counter]{abox}[6][]{
	enhanced,
	colframe=black,
	colback=white,
	boxrule={#4},
	arc={#3},
	auto outer arc,
	pad at break*=0pt,
	vfill before first,
	before={\par\medskip\noindent},
	after={},
	top=12pt, left=4pt, enlarge top by=6pt,
	title={\rule[-.3\baselineskip]{0pt}{\baselineskip}\normalsize\sffamily\bfseries #2}, 
	varwidth boxed title*=-10pt, 
	attach boxed title to top left={yshift=-10pt,xshift=10pt}, 
	coltitle=black,
	boxed title style={colback=white,boxrule={#6},arc={#5},auto outer arc},
	#1
}

\newenvironment{algobox}[2][]
{\begin{abox}[#1]{Algorithm~\thetcbcounter: \normalfont #2}{0.5pt}{0.5pt}{1pt}{0.75pt}}
{\end{abox}}

\tcbset{
	idealfunc/.style={
		colback=gray!2,
		colframe=black,
		fonttitle=\bfseries,
		boxrule=0.2mm,
		rounded corners,
		title=Ideal Functionality $\mathcal{F}$,
	}
}

\newif\ifpearl
\pearlfalse 

\ifpearl
  \newcommand{\sys}{Pearl}
  \newcommand{\sysfootnote}{\footnote{Like its namesake, \textbf{P}rivate
  \textbf{E}nclave for \textbf{A}gentic \textbf{R}etrieval and \textbf{L}ongterm
  memory accumulates layers over time, and remains opaque to the eye.}}
  \newcommand{\SysEnc}{\ensuremath{\mathsf{Enclave}_{\mathsf{Pearl}}}}
  \newcommand{\sysAn}{A \sys}
  
  \newcommand{\prooflocation}{supplementary material (\S A)}
  \newcommand{\calibrationtablecontext}{Table~1 in the supplementary material (\S B)}
  \newcommand{\calibrationtableref}{Supp.\ \S B, Table~1}
  \newcommand{\hawkesdetailref}{Supp.\ \S B.1}
  \newcommand{\recencydetailref}{Supp.\ \S B.2}
  \newcommand{\systemarchitecturefigure}{figures/system_architecture.pdf}
  \newcommand{\evalaccuracyfigure}{figures/eval_accuracy_pearl}
  \newcommand{\evalbandwidthfigure}{figures/eval_bandwidth_pearl}
  \newcommand{\evallatencyfigure}{figures/eval_latency_pearl}
  
  \newcommand{\leadpara}[1]{\textbf{#1}}
\else
  \newcommand{\sys}{Opal}
  \newcommand{\sysfootnote}{\footnote{Like the gemstone, \textbf{O}blivious
  \textbf{P}ersonal \textbf{A}I \textbf{L}ongterm memory accumulates pieces over
  time, and remains opaque to the eye.}}
  \newcommand{\SysEnc}{\ensuremath{\mathsf{Enclave}_{\mathsf{Opal}}}}
  \newcommand{\sysAn}{An \sys}
  
  \newcommand{\prooflocation}{\S\ref{sec:security-proof}}
  \newcommand{\calibrationtablecontext}{Table~\ref{tab:calibration}}
  \newcommand{\calibrationtableref}{Table~\ref{tab:calibration}}
  \newcommand{\hawkesdetailref}{\S\ref{subsec:supp-hawkes}}
  \newcommand{\recencydetailref}{\S\ref{subsec:recency-sup}}
  \newcommand{\systemarchitecturefigure}{figures/system_architecture_opal.pdf}
  \newcommand{\evalaccuracyfigure}{figures/eval_accuracy_opal}
  \newcommand{\evalbandwidthfigure}{figures/eval_bandwidth_opal}
  \newcommand{\evallatencyfigure}{figures/eval_latency_opal}
  
  \newcommand{\leadpara}[1]{\paragraph{#1}}
\fi

\newcommand{\Query}{\ensuremath{\mathsf{Query}}\xspace}
\newcommand{\Ingest}{\ensuremath{\mathsf{Ingest}}\xspace}
\newcommand{\Dream}{\ensuremath{\mathsf{Dream}}}
\newcommand{\Embed}{\ensuremath{\mathsf{Embed}}}
\newcommand{\Synthesize}{\ensuremath{\mathsf{Synthesize}}}
\newcommand{\Traverse}{\ensuremath{\mathsf{Traverse}}}
\newcommand{\Score}{\ensuremath{\mathsf{Score}}}
\newcommand{\Rerank}{\ensuremath{\mathsf{Rerank}}}
\newcommand{\BatchSearch}{\ensuremath{\mathsf{BatchSearch}}}
\newcommand{\BatchInsert}{\ensuremath{\mathsf{BatchInsert}}}
\newcommand{\Update}{\ensuremath{\mathsf{Update}}}
\newcommand{\Insert}{\ensuremath{\mathsf{Insert}}}

\newcommand{\EmbEnc}{\ensuremath{\mathsf{Enclave}_{\mathsf{Emb}}}}
\newcommand{\OpalEnc}{\SysEnc}
\newcommand{\LLMEnc}{\ensuremath{\mathsf{Enclave}_{\mathsf{LLM}}}}

\newcommand{\KG}{\ensuremath{\mathsf{KG}}}
\newcommand{\ANN}{\ensuremath{\mathsf{ANN}}}
\newcommand{\ORAM}{\ensuremath{\mathsf{ORAM}}}

\definecolor{PendingResult}{HTML}{0F766E}

\providecommand{\rbtext}[1]{#1}

\renewcommand{\evalaccuracyfigure}{figures/v11/eval-accuracy}

\renewcommand{\evallatencyfigure}{figures/v11/eval-latency-v85-leg8}
\renewcommand{\evalbandwidthfigure}{figures/v11/eval-bandwidth-e3}

\begin{document}

\title{\sys: Private Memory for Personal AI}

\ifpearl
\author{Paper \#1199}
\else
\author{Darya Kaviani}
\affiliation{\institution{UC Berkeley}}
\author{Alp Eren Ozdarendeli}
\affiliation{\institution{UC Berkeley}}
\author{Jinhao Zhu}
\affiliation{\institution{UC Berkeley}}
\author{Yu Ding}
\affiliation{\institution{Google DeepMind}}
\author{Raluca Ada Popa}
\affiliation{\institution{UC Berkeley, Google DeepMind}}
\fi
\begin{abstract}
    Personal AI systems increasingly retain long-term memory of user activity,
including documents, emails, messages, meetings, and ambient recordings. Trusted
hardware can keep this data private, but struggles to scale with a growing datastore. This pushes the data to external storage, which exposes retrieval
access patterns that leak private information to the application provider.
Oblivious RAM (ORAM) is a cryptographic primitive that can hide these patterns,
but it requires a fixed access budget, precluding the query-dependent traversals
that agentic memory systems rely on for accuracy.
We present \sys{}, a private memory system for personal AI. Our key insight is
to decouple all data-dependent reasoning from the bulk of personal data,
confining it to the trusted enclave. \rbtext{The untrusted} disk then sees only fixed,
oblivious memory accesses. This enclave-resident component uses a lightweight
knowledge graph to capture personal context that semantic search alone misses
and handles continuous ingestion by piggybacking capacity management
on every ORAM access. Evaluated on a comprehensive synthetic personal-data
pipeline driven by stochastic communication models, \sys{} improves retrieval
accuracy by 13 percentage points over semantic search and achieves 33$\times$
higher throughput with 19$\times$ lower infrastructure cost than a secure
baseline. \sys{} is under consideration for deployment to millions of users at
a major AI provider.

\end{abstract}

\maketitle

\pagestyle{plain}

\begin{figure}[t]
    \centering
    \includegraphics[width=0.8\columnwidth]{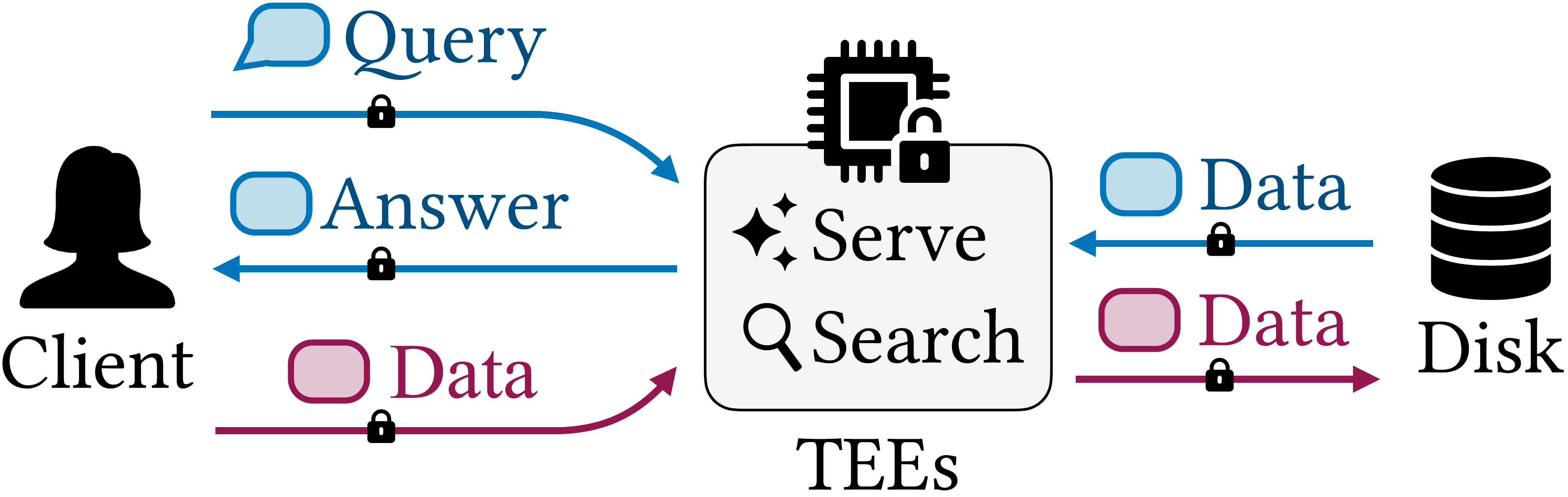}
    \caption{System model. The client sends {\color[HTML]{0075B9}queries} and
    {\color[HTML]{941651}ingestion} data to TEEs that run search and model
    serving while storing encrypted state on untrusted disk.}
    \label{fig:system-model}
\end{figure}

\section{Introduction}\label{sec:intro}

Users increasingly share sensitive, personal information with AI systems in
their daily interactions~\cite{chatbotselfdisclosure, voicepersonalassistant,
digitalconfessions, longitudinalstudyofselfdisclosure}. This trend is
accelerating: Microsoft Recall continuously screenshots user
activity~\cite{windowsrecall}, Google Gemini ingests emails and
photos~\cite{geminipersonalintelligence}, and AI memory tools capture notes,
meetings, and conversations~\cite{limitlessai, memai, personalai}. Looking
ahead, embodied AI agents in robotics and wearable devices are beginning to
capture even richer streams of ambient data (continuous audio, video, and sensor
recordings) from everyday life~\cite{1x, sunday, googleastra,
metaraybannametag}. 

The result is an unprecedented concentration of personal data at providers whose
profiles of user behavior can be shared with third parties, retained for AI
training, or compelled by governments~\cite{ftcsocialmedia2024, openaigov2025,
pewprivacy2023}. End-to-end encrypted messaging has established private
communication as a baseline user expectation~\cite{signal, whatsapp, messenger},
yet personal AI systems, which accumulate far richer portraits of their users
than any messaging app~\cite{chatbotselfdisclosure, voicepersonalassistant,
digitalconfessions, longitudinalstudyofselfdisclosure, pewprivacy2023}, offer no
comparable guarantee. Indeed, outsourced workers were recently found reviewing
sensitive videos captured by Meta's Ray-Ban smart
glasses~\cite{metaraybanreview}. Worse, compromising the application provider's
infrastructure could expose this data in its entirety.

In response, AI providers including Apple~\cite{privatecloudcompute} and
Google~\cite{googletrustedexecution} have turned to trusted
hardware, the most practical approach for keeping inference queries hidden from
the provider at scale. Yet LLM context windows and trusted hardware memory are
finite and expensive, driving personal AI systems toward persistent long-term
memory that retains information across sessions and tailors responses over
time~\cite{openaimemory}. Google has announced a private, server-side version that keeps memory in encrypted storage and decrypts it inside a secure enclave for each request~\cite{googleprivatememory2026}. Hiding access patterns this way requires loading a user's entire store into the enclave on every request, so every request pays for the
whole memory. The alternatives are no better:

\begin{itemize}[leftmargin=*]
    \item \textit{Selective fetches} from untrusted storage into the enclave expose which records each query
    accesses. Even when individual memories are encrypted, observing which
    memories a system retrieves is nearly as revealing as seeing the query text
    itself, since retrieved items are semantically similar to the
    query~\cite{informationleakageinembeddings,
    textembeddingsrevealalmostastext,
    sentenceembeddingleaksmoreinformationthanyouexpect}.
    \item \textit{Client-local memory} stores personal data on the user's device
    and re-uploads relevant context on every request~\cite{cecollm,
    privatecloudcompute}, yet personal corpora quickly outgrow a single
    device~\cite{stuffiveseen, mylifebits}, local storage lacks the durability
    required for a long-lived memory store~\cite{diskfailuresrealworld,
    datalossrecovery, fsyncfailures, mobileflashwear}, and multi-device access
    demands a unified backend~\cite{multidevicecomputing, multideviceusage}.
    Thus, assistant providers maintain memory cloud-natively~\cite{openaimemory,
    geminipersonalintelligence, googleprivatememory2026}, so \sys{} is designed for
    this cloud-hosted setting.
\end{itemize}

What is needed, then, is a way to keep search and model serving inside trusted
hardware while storing personal data on untrusted cloud storage
(Figure~\ref{fig:system-model}), hiding both content and access patterns from
the provider. Agentic memory systems retrieve relevant context via semantic
search over embeddings~\cite{memgpt, mem0, memos, raptor, generativeagents}. The
natural tool for hiding these access patterns is oblivious RAM
(ORAM)~\cite{oram96, oram90, binarytreeoram, pathoram, ringoram}, a
cryptographic memory abstraction that hides data access patterns at a per-access
cost logarithmic in database size. Compass~\cite{compass} co-designs approximate
nearest neighbor (ANN) search with ORAM, showing that semantic search can be
made access-pattern-hiding.

However, agentic memory systems achieve high accuracy precisely because they
augment semantic search with query-dependent traversals, using knowledge graphs (KGs), temporal
filters, and multi-hop lookups to strategically narrow results~\cite{mem0,
memgpt, memos, amem, zep, hipporag, beyonddialoguetime}. Oblivious ANN search
does not support this: it enforces a fixed number of ORAM reads per query, as
any variation leaks information about what is being retrieved. Padding ORAM accesses to a fixed budget appears incompatible with these
traversals, since supporting them naively requires either additional ORAM
round-trips or a larger fixed budget that increases overhead for every query. The problem is compounded by continuous ingestion: prior oblivious ANN
work supports insertions and deletions but not garbage collection, and
data-dependently deciding when to perform them leaks information about recent activity. In short, efficient obliviousness and
accuracy appear fundamentally at odds for personal AI memory.

To resolve this tension, we design \sys\sysfootnote, a private memory system for
personal AI that is efficient and accurate. \sys{} adds a privacy-preserving
storage and retrieval layer to trusted-hardware inference enclaves, enabling
users to persist and query their personal data (documents, emails, messages,
meetings, and ambient recordings) without revealing content or access patterns
to the application provider. \sys{} splits storage between trusted enclave
memory and an untrusted ORAM-backed disk (\S\ref{sec:system-overview}). Two ORAM
databases on disk hold embedding vectors and raw data chunks, while compact
index metadata and a personal knowledge graph remain inside the enclave. The key
insight is that all data-dependent reasoning can run entirely inside an enclave,
so that the storage provider only ever sees fixed-size, oblivious ORAM accesses.
\sys{} realizes this insight through three novel contributions:

First, \textbf{knowledge-graph-filtered
search}~(\S\ref{sec:kg-filtered-search}). Pure semantic search is often
insufficient for personal data: queries like ``what did Alice say in last
Monday's meeting?'' depend on temporal context and entity relationships that
embedding similarity alone cannot capture. Structured agent memory systems address this by augmenting vector search
with knowledge graphs, structured representations of entities and their
relationships~\cite{zep, hipporag, magma, graphrag, lightrag}. However, these systems build heavyweight, content-rich
graphs whose data-dependent traversals are incompatible with ORAM's fixed
access budget.
Our key observation is that personal data is inherently structured: it
involves particular people, times, and projects. This structure can be factored out of the content itself. \sys{} exploits this by building
a compact, metadata-only graph of normalized identifiers (people, sources, time
ranges) that fits entirely inside the enclave and stores no message or document
bodies. An LLM
extracts structured predicates from the query, graph traversal narrows the
candidate set in trusted memory, and the ANN scores only within that admissible
set before a single fixed-budget ORAM fetch. All query-dependent reasoning stays inside the enclave, so the storage
provider sees only a fixed-size, oblivious access pattern.

Second, \textbf{oblivious dreaming}~(\S\ref{sec:oblivious-dreaming}). Agentic
memory systems routinely perform heavy maintenance (compaction, conflict
resolution, summarization)~\cite{memgpt, raptor, mem0, memos, zep,
graphrag, generativeagents} whose access patterns would leak information about
the user's data. A fixed maintenance schedule is not viable either, since it
must either scan the entire store, incurring significant overhead, or leave
stale state lingering. By nature of the cryptographic protocol, every ORAM access already reads and
writes blocks, including dummy blocks and unrelated real chunks.
Like biological memory consolidating during sleep, oblivious dreaming folds
maintenance into these ordinary accesses without issuing any extra access or
breaking the security guarantee. This enables expiring old items and reusing
their slots, all at steady-state capacity behind a fixed access pattern.

Third, a \textbf{synthetic personal-data pipeline}
(\S\ref{sec:synthetic-personal-data}). Evaluating a personal memory system
requires realistic multi-modal corpora spanning months or years, but real
personal data is too sensitive to collect or share at scale. Existing synthetic
benchmarks~\cite{personabench} generate isolated personal artifacts but lack the
temporal dynamics essential for memory evaluation: messages arrive in bursty
threads, meetings spawn follow-up emails, and projects evolve over weeks. We
address this gap with a Hawkes-process~\cite{hawkes1971, hawkes1974} arrival
model that reproduces correlated, bursty communication patterns. Paired with an
LLM-driven scenario planner, the pipeline produces years of realistic personal
data. We expect this to be of independent interest beyond privacy.

\leadpara{Evaluation summary.} \sys{} achieves 33$\times$ higher throughput and
19$\times$ lower infrastructure cost than our In-Memory baseline, which follows this design~\cite{googleprivatememory2026,enclavedb,obliviate,hardidx}. KG-filtered search
improves accuracy by 13\,pp over plain ANN, matching Graphiti~\cite{zep}, a plaintext
agentic memory system. Oblivious dreaming keeps active entries within a fixed capacity
budget at the same accuracy as eager, non-oblivious ANN rebalancing. \sys{}'s search index also needs little training data: one trained on a small prefix of a user's data is as accurate as one trained on far more. \sys{} answers queries in 2.27\,s and ingests a new memory
in 0.48\,s, 2.01$\times$ and 13.6$\times$ faster than this baseline,
respectively. A major AI provider has evaluated \sys{} and is considering
deploying its techniques in their production memory system, serving millions of
users.

\begin{figure*}[t]
    \centering
    \includegraphics[width=0.68\textwidth]{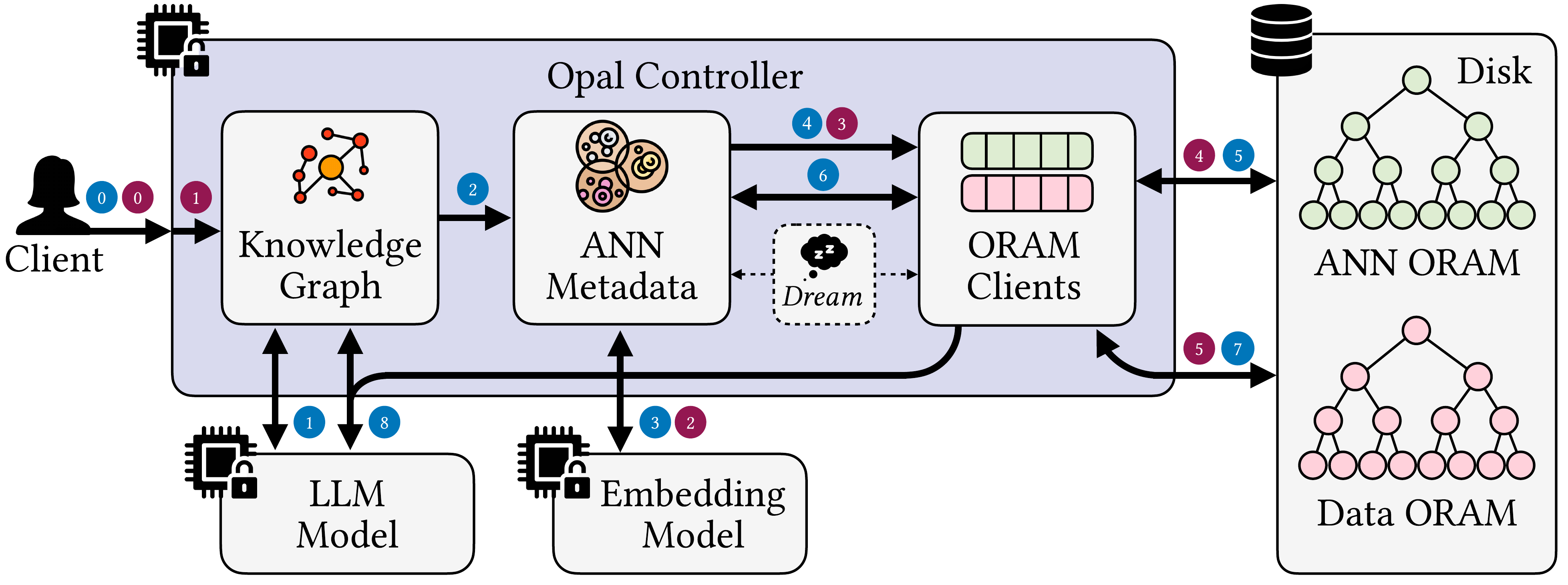}
        \caption{\sys{} system architecture. The \sys{} Controller orchestrates
        query and ingestion operations across three enclaves, routing data
        through oblivious storage (ORAM) on untrusted disk.
        {\color[HTML]{0075B9}Blue} numbers trace the query flow;
        {\color[HTML]{941651}red} numbers trace ingestion. }
        \label{fig:system-architecture}
\end{figure*}

\section{System Overview}
\label{sec:system-overview}

\subsection{System Model}
\label{subsec:system-model}

\sysAn{} instance consists of three components: the client, a set of trusted
enclaves, and an untrusted disk. All private computation takes place inside the
enclaves, which hold compact client-side state in memory, while the untrusted
disk stores encrypted ORAM state. There are three enclaves, each with a distinct
role:
\begin{itemize}[leftmargin=*]
    \item $\OpalEnc$ (controller): orchestrates \sys{}'s query and ingestion
    operations, and maintains the knowledge graph, ANN index metadata, and ORAM
    client state.
    \item $\EmbEnc$ (embedding): serves an embedding model.
    \item $\LLMEnc$ (LLM): serves an LLM model.\footnote{\rbtext{The three roles could share one enclave without weakening \sys{}'s
guarantees; we separate them so that \sys{} can reuse existing confidential
model services (\S\ref{subsec:experiment-setup}), which providers have begun
building for proprietary models~\cite{privatecloudcompute,
googletrustedexecution, anthropictrustedexecution}.}}
\end{itemize}

The untrusted disk stores two ORAM databases: $\ORAM_{\ANN}$ holds vector
embeddings used to find relevant matches, and $\ORAM_{\mathsf{Data}}$ holds the
raw data chunks passed to the LLM for answer generation. \sys{} exposes two
operations:
\begin{itemize}[leftmargin=*]
    \item $\OpalEnc.\Query(q)$: retrieve and synthesize an answer from personal
    memory for a natural-language query~$q$.
    \item $\OpalEnc.\Ingest(d)$: ingest a new fixed-size data chunk~$d$ into the
    memory store. 
\end{itemize}
Unlike client-local memory, where the device must store and search all personal
data itself~\cite{windowsrecall}, \sys{} keeps minimal state on the client (a
secret key~$k$ and a monotonic counter~$\mathsf{ctr}$; see
\S\ref{sssec:freshness-rollback}): the
logical client only submits pre-chunked fixed-size items to $\OpalEnc$, which
handles storage and retrieval on ORAM-backed disk.\footnote{An alternative deployment could place a normalized ingestion pipeline at the application provider, as is common in production systems~\cite{facebookrealtime, tfx, prochlo}; this would require a different freshness mechanism (e.g., a verifiable randomness beacon~\cite{nistir8213}) in place of the monotonic counter in \S\ref{sssec:freshness-rollback}.}
A single enclave serves multiple
users' $\OpalEnc$ controllers, each with their own separate ORAM databases.

\subsection{Threat Model \& Security Guarantees}
\label{subsec:threat-model-and-security-guarantees}
\subsubsection{Adversary Model}
\label{sssec:trust-boundary}

We adopt an established threat model of a malicious host with an enclave that is
assumed to be uncompromised, as in prior TEE-backed and confidential-VM systems
(e.g., Signal's SVR3~\cite{svr3}, Occlum~\cite{occlum}, and
Gramine-TDX~\cite{graminetdx}; see also~\cite{enclavedb, obliviate, hardidx,
weave, hecate, cvmexplained, trustedaiagents}), and industry confidential AI
systems from Apple~\cite{privatecloudcompute},
Google~\cite{googletrustedexecution}, and
Anthropic~\cite{anthropictrustedexecution}.
In particular, the attacker cannot view or tamper with the data and computation
inside the enclave, but it can observe, drop, replay, reorder, and tamper with
all communication, data and computation outside of the enclave. Since
network-level traffic already reveals which user is contacting the service, we
focus on hiding what the user is doing, not that the user is present.
\rbtext{Even when a source application already holds some inputs in plaintext
(e.g., email), the combined memory remains sensitive: the memories each query
retrieves reveal what the user is focused on and tie it to ambient recordings
that \sys{} stores only in encrypted form.}
\sys{} does not protect against DoS attacks from the server.

It is particularly important to protect access patterns to storage outside of
the enclave because, in standard confidential VM deployments, persistent state resides on
provider-managed remote storage~\cite{googlepersistentdisk} accessed through
host- or hypervisor-mediated I/O paths that lie outside the trust
boundary~\cite{inteltdxguestsecurity, occlum, cvmexplained,
confidentialcontainers, tinfoilconfidentiality}. The host can therefore log
every storage access without attacking the TEE~\cite{googlepersistentdisk,
inteltdxguestsecurity}; prior systems treat this channel as a primary
threat~\cite{cvmexplained, confidentialcontainers, tinfoilconfidentiality}.
\sys{} ensures that off-enclave storage accesses are independent of query and data
contents.

We model each enclave as an instance of $\mathcal{G}_{\mathsf{att}}$ (as
formally defined in~\cite{formalabstractionsforenclaveattestation}), which runs
a program with private mutable state and returns outputs with
manufacturer-signed attestation bound to the code identity. We further assume
the host cannot read or modify enclave memory during execution, interacting only
through explicit I/O. Client-enclave and inter-enclave communication runs over
authenticated encrypted channels, and the inter-enclave messages are padded to fixed size. TEEs do not natively prevent rollback of
externally stored state~\cite{memoir, rote, narrator}; following prior
work~\cite{svr3, nimble}, \sys{} adds freshness mechanisms
(\S\ref{sssec:freshness-rollback}).

\leadpara{Out-of-scope attacks.}
We focus on access-pattern leakage outside the enclave, assuming the enclave
itself is uncompromised. A complementary line of work studies side-channel
attacks \emph{inside} enclaves: page-table and
paging~\cite{controlledchannelattacks, tellingyoursecretswithoutpagefaults,
foundintranslation}, cache and branch-prediction~\cite{cachezoom,
branchshadowing, invalidatecompare}, interrupt, exception, stepping, and
instruction-counting~\cite{wesee, sevstep, heckler, tdxdown, tdxploit}, timing
and performance-counter~\cite{ttime, counterseveillance}, memory-bus and
power~\cite{offchipattackonhardwareenclaves, oops}, physical~\cite{teefail,
badram}, transient-execution and microarchitectural leakage~\cite{foreshadow,
cacheout, aepicleak, cipherleaks}, fault injection~\cite{plundervolt}, and
attestation or deployment failures~\cite{sgxfail, teepitfalls}.

Current confidential-computing hardware makes these attacks substantially
harder. In our deployment (\S\ref{sec:evaluation}), all enclaves run in
TDX-backed CVMs, and $\EmbEnc$ and $\LLMEnc$ use NVIDIA~B200 GPUs in
confidential-computing mode. In particular, TDX removes the host's direct
software visibility into TD-private guest page tables and memory, so classic
guest-PTE accessed-bit oracles do not directly apply to TD-private
memory~\cite{foundintranslation, inteltdx, linuxkerneltdx,
inteltdxguestsecurity}. On the GPU side, confidential-computing mode removes
management, debugging, and profiling observability~\cite{nvidiasecureai}.
Deployed platforms further mitigate residual channels
operationally~\cite{awsnitrosidechannels, tinfoilsidechannels,
tinfoilconfidentiality}.

Limitations remain: recent work studies host-induced channels on TDX including
interrupt and exception attacks~\cite{wesee, heckler}, stepping
attacks~\cite{sevstep, tdxdown, tdxploit}, timing and performance-counter
attacks~\cite{ttime, counterseveillance}, cache attacks~\cite{tdxploit}, and
ciphertext-side attacks against CVM memory encryption~\cite{relocatevote}.
GPU-internal side channels such as cache timing~\cite{invalidatecompare,
renderedinsecure} and memory-disturbance attacks~\cite{gpuhammer} also persist.
These remaining attacks require microarchitectural, cache, exception, stepping,
timing, or performance-counter channels rather than ordinary host software
observation. Proposals for strengthening enclave security include interrupt-aware
and controlled-channel hardening~\cite{aexnotify, tlblur}, defenses against
performance-counter attacks~\cite{teecorrelate}, TDX-specific hardening of trust
domains~\cite{tetd}, doubly-oblivious data structures~\cite{oblix}, formal
attestation models~\cite{formalabstractionsforenclaveattestation}, and
smaller-TCB enclave architectures~\cite{sanctum, komodo, keystone}.

\subsubsection{Security Guarantees \& Definition}
\label{sssec:guarantees}
\sys{} provides the following security guarantees:
\begin{itemize}[leftmargin=*]
    \item \textbf{Confidentiality.} Under the adversary model above
    (\S\ref{sssec:trust-boundary}), all user data,
    queries, and responses remain hidden from the adversary. The adversary may
    learn whether an operation is $\Query$ or $\Ingest$, but beyond the declared
    leakage it learns neither the content of any item, query, or response, nor
    which stored data is relevant to a query or which records were accessed on
    its behalf.
    \item \textbf{Integrity and Freshness.} Tampering with or rollback of
    ORAM-backed storage is detected on every request (\S\ref{sssec:freshness-rollback}).
\end{itemize}

\leadpara{Security Definition.} Let $\Pi$ denote the set of public system
parameters: ANN fetch count $n$, reranking budget $K$,
tree depth $L$, and bucket capacity $Z$.
We provide an indistinguishability-based security definition relative to a
leakage function~$\mathcal{L}$. For any request sequence, $\mathcal{L}$ reveals
the public parameters~$\Pi$, the operation type at each step, and the resulting
fixed trace shape: the number of fixed-size inter-enclave messages, the ORAM
batch sizes, and the lengths of the client's requests and responses. We also treat the provider checkpoint boundary~$\kappa_i$ as an
exogenous public schedule fixed across the two worlds, and we do not model the
timing of requests. The game is played in the
$\mathcal{G}_{\mathsf{att}}$-hybrid model. The adversary controls all
computation, storage, and message delivery outside enclave boundaries. It wins
if it learns information beyond $\mathcal{L}$ about the user's queries or data,
or causes an undetected incorrect execution of \sys{}.

\begin{figure}[!t]
    \centering
    \begin{tcolorbox}[colback=white, colframe=black, boxrule=0.4pt, arc=4pt,
        left=6pt, right=6pt, top=6pt, bottom=6pt]
        \fontsize{10pt}{12pt}\selectfont
    \begin{enumerate}[leftmargin=*]
        \item The challenger $\mathcal{C}$ chooses a uniformly random bit~$b$.
        \item The adversary $\mathcal{A}$ chooses equally large datasets $D_0$
        and $D_1$ and public parameters~$\Pi$. The challenger initializes the
        system with $D_b$ and~$\Pi$.
        \item For each step~$i$:
        \begin{enumerate}[leftmargin=*]
            \item The adaptive adversary $\mathcal{A}$ chooses public requests
            $q_{i,0}, q_{i,1} \in \{\Query,\Ingest\}$.
            \item If $\mathcal{L}(q_{i,0}) \neq \mathcal{L}(q_{i,1})$, the
            challenger aborts. Otherwise, it executes $q_{i,b}$ according to
            \sys{}'s protocol in the $\mathcal{G}_{\mathsf{att}}$-hybrid model,
            aborting on any ORAM-storage or checkpoint verification
            failure.
            \item The challenger gets the request result $r_i$.
        \end{enumerate}
        \item The adversary $\mathcal{A}$ outputs a guess $b'$.
    \end{enumerate}
    \end{tcolorbox}
    \caption{Security game for \sys{}}
    \label{fig:security-game}
\end{figure}

The adversary $\mathcal{A}$ wins the game if the challenger $\mathcal{C}$ does
not abort and one of the following two conditions is met:
\begin{enumerate*}[label=\textbf{\itshape\arabic*\upshape)}]
    \item $b' = b$;
    \item the sequence of requests $(q_{i,b},\, r_i)$ is an incorrect execution
    of \sys{}'s plaintext algorithm with $\Pi$ on $D_b$ under the same provider
    residency and checkpoint policy.
\end{enumerate*}

\begin{theorem}
\label{thm:opal-security}
As defined by the security game in Figure~\ref{fig:security-game}, for any PPT
stateful adversary~$\mathcal{A}$, the \sys{} protocol defined in
\S\ref{sec:opal-system-design} and Algorithm~\ref{alg:opal} satisfies condition~\textbf{\itshape 1\upshape)} with
probability at most $\frac{1}{2} +
\mathsf{negl}(\lambda)$ and
condition~\textbf{\itshape 2\upshape)} with
probability at most $\mathsf{negl}(\lambda)$ in the
$\mathcal{G}_{\mathsf{att}}$-hybrid model, when instantiated with a durable
monotonic client counter, a secure key-derivation function, a SUF-CMA secure
MAC, a collision-resistant hash function, and an authenticated encryption with
associated data scheme.
\end{theorem}

We give a proof sketch in \S\ref{sec:security}; the full proof is in
\prooflocation.

\section{\sys{}'s System Design}
\label{sec:opal-system-design}
\begin{adjustbox}{max width=\textwidth, center, minipage=[t]{1.1\textwidth},
float={figure*}}
    \begin{algobox}[label={alg:opal}]{\sys{}}
        \small
        \algtext*{EndIf}
        \begin{minipage}[t]{0.48\textwidth}
            \underline{$\OpalEnc.\Query(q)$}
            \begin{algorithmic}[1]
                \State $F \gets \LLMEnc.\Traverse(q,\, \KG)$ \Comment{filter set}
                \State $A_q \gets \KG.\Traverse(F)$ \Comment{KG traversal}
                \State $e_q \gets \EmbEnc.\Embed(q)$ \Comment{embed query}
                \State $C \gets \ANN.\Score(e_q, A_q)$ \Comment{ANN scoring}
                \State $P_{\ANN} \gets \BatchSearch(\ORAM_{\ANN},\, C,\, n)$ \Comment{top-$n$}
                \State $S_K \gets \ANN.\Rerank(P_{\ANN},\, e_q,\, K)$ \Comment{top-$K$}
                \State $P_{\mathsf{Data}} \gets \BatchSearch(\ORAM_{\mathsf{Data}},\, S_K,\, K)$;\ $\Dream()$
                \State $a \gets \LLMEnc.\Synthesize(P_{\mathsf{Data}}[S_K])$ \Comment{answer synthesis}
                \State \Return $a$
            \end{algorithmic}
        \end{minipage}
        \hfill
        \begin{minipage}[t]{0.48\textwidth}
            \underline{$\OpalEnc.\Ingest(d)$}
            \begin{algorithmic}[1]
                \State $\KG.\Update(d.\mathsf{metadata})$ \Comment{update KG}
                \State $e_d \gets \EmbEnc.\Embed(d)$ \Comment{embed data}
                \State $\ANN.\Insert(e_d)$ \Comment{assign cluster + update index}
                \State $P_{\ANN} \gets \BatchInsert(\ORAM_{\ANN},\, e_d)$
                \State $P_{\mathsf{Data}} \gets \BatchInsert(\ORAM_{\mathsf{Data}},\, d)$;\ $\Dream()$
            \end{algorithmic}
        \end{minipage}
    \end{algobox}
\end{adjustbox}

\sys{}'s central design principle is that all data-dependent reasoning executes
inside the enclave, so that every access crossing the trust boundary follows a
fixed, content-independent pattern.
A query first extracts structured predicates inside $\LLMEnc$, traverses the enclave-resident
KG to form an admissible set, embeds the query, and scores ANN candidates before
issuing a fixed-budget ORAM fetch for vector reranking. \sys{} then performs a
second fixed-budget ORAM fetch for the top raw chunks and passes them to
$\LLMEnc$ for final answer synthesis.
An ingestion updates KG metadata, embeds the
new chunk, and inserts it into the ANN metadata and both ORAM stores. In both cases,
$\Dream()$ executes only on blocks already brought into the enclave on the
ordinary ORAM path.
Algorithm~\ref{alg:opal} and
Figure~\ref{fig:system-architecture} describe the full protocol.

\subsection{Technical Background}
\label{subsec:technical-background}

\subsubsection{Oblivious RAM}

An ORAM allows a client to read and write blocks on an untrusted server such
that the server cannot distinguish accesses to different blocks~\cite{oram96,
pathoram}. \sys{} uses Ring ORAM~\cite{ringoram}, a tree-based construction
where the server stores $N$ blocks in a binary tree of depth~$L$ with buckets of
capacity~$Z$. Each access reads and writes a fixed number of encrypted blocks
regardless of which block is requested; a client-side position map and stash
track block locations between accesses. Each access incurs $O(\log N)$
bandwidth, so any additional access beyond the fixed budget multiplies cost
proportionally.\footnote{\rbtext{We use Ring ORAM rather than Path ORAM~\cite{pathoram} because each access reads only one block per bucket on its path rather than whole buckets; Compass's oblivious ANN search builds on the same construction~\cite{compass}.}}
Both ORAM stores are encrypted under a per-client key derived
from the client secret; integrity and rollback are described in
\S\ref{sssec:freshness-rollback}.

\subsubsection{IVF Indexes}

Inverted file (IVF) indexes partition a vector collection into clusters, each
represented by a centroid~\cite{faiss}. At query time, the index identifies the
nearest centroids and scores only vectors in those clusters, avoiding a full
scan. Product quantization (PQ) compresses each vector into a short code that
enables approximate scoring without accessing the full vector~\cite{faiss}.

Prior oblivious ANN work~\cite{compass} builds on graph-based indexes (HNSW),
but multi-hop traversals multiply ORAM round-trips and graphs are difficult to
update incrementally~\cite{freshdiskann, spfresh}. IVF indexes handle
insertion-heavy workloads more gracefully~\cite{freshdiskann, spfresh} and fit
our setting naturally: scoring uses compact PQ codes inside the enclave, and only a
single bulk ORAM fetch retrieves full vectors for reranking.

\subsection{KG-filtered Search}
\label{sec:kg-filtered-search}

Structured retrieval in agentic memory systems achieves accuracy by building
heavyweight, content-rich knowledge graphs that tightly couple retrieval with
data-dependent traversals~\cite{zep, hipporag, graphrag, lightrag}. These graphs
store extracted entities, relationships, and even community summaries alongside
the corpus, and every query triggers multi-hop walks over this disk-resident
state. In an ORAM-backed store, each such traversal either inflates the fixed
access budget or leaks information about what the system retrieves. Scoring the
full corpus first and filtering afterward would waste the fixed ORAM budget on
semantically similar but contextually irrelevant items, while filtering outside
the enclave would leak query predicates.

\sys{} departs from this paradigm entirely: instead of a content-rich graph, it
maintains a compact, metadata-only knowledge graph of normalized identifiers
(people, sources, time ranges) that fits inside the enclave and stores no message
or document bodies. In Algorithm~\ref{alg:opal}, $\LLMEnc.\Traverse(q,\KG)$ extracts a
filter set~$F$ from the query, and $\KG.\Traverse(F)$ resolves those predicates
entirely inside the enclave to produce an admissible set~$A_q$. The ANN then
scores only within~$A_q$, after which \sys{} performs the same fixed-budget ORAM
fetch and reranking pipeline as in the unfiltered case. Because filtration
changes \emph{which} items are eligible rather than \emph{how many} storage
accesses occur, \sys{} achieves the accuracy benefits of graph-augmented retrieval
(\S\ref{subsec:accuracy}). Despite storing no message or document bodies in the graph,
this lightweight approach matches a plaintext agentic memory baseline with
query-dependent retrieval.

\subsubsection{Enclave-Resident Graph}

The knowledge graph serves as a compact in-enclave control plane: raw text and
embeddings remain in ORAM on untrusted storage, and the graph itself stores only
identifiers, sparse metadata (timestamps, participants, titles and subject lines),
and structural relationships, all
deterministically extracted from application metadata (e.g., message
participants, meeting attendees, folder hierarchy) \rbtext{inside $\OpalEnc$} at ingest time without an
LLM call.
Concretely, the graph contains four node types:
\textbf{Artifact nodes} represent a logical memory unit (email,
    meeting, document), each carrying a timestamp and modality label for
    temporal and source-type filtering.
\textbf{Chunk nodes} are identifiers for fixed-size pieces of an
    artifact's content, each linked to its parent artifact; the chunk content
    itself resides in ORAM.
\textbf{Person nodes} represent named individuals (e.g., email
    senders, meeting attendees). Each artifact holds edges to its participants.
\textbf{Project nodes} represent workstreams (folder hierarchies,
    project spaces, recurring meetings); each artifact holds edges to its
    project node.
\begin{figure}[tbp]
    \centering
    \includegraphics[width=0.85\linewidth]{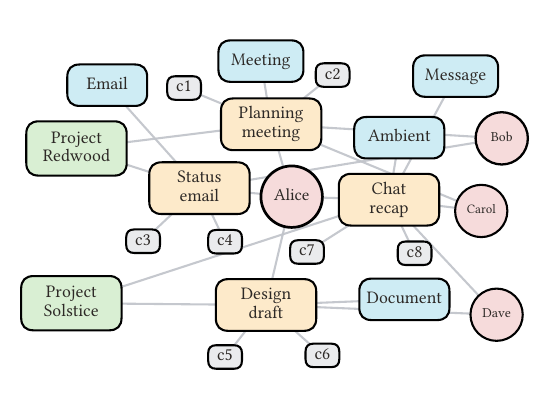}
    \caption{Sample \sys{} knowledge graph for user Alice. {%
    \setlength{\fboxsep}{1pt}%
    \colorbox[HTML]{FDEACA}{\strut artifact}%
    } nodes link to {%
    \setlength{\fboxsep}{1pt}%
    \colorbox[HTML]{CEECF4}{\strut modality}%
    }, {%
    \setlength{\fboxsep}{1pt}%
    \colorbox[HTML]{F6DBDB}{\strut person}%
    }, and {%
    \setlength{\fboxsep}{1pt}%
    \colorbox[HTML]{D9EFD3}{\strut project}%
    } nodes; {%
    \setlength{\fboxsep}{1pt}%
    \colorbox[HTML]{EAEBED}{\strut chunk}%
    } nodes are labeled $c_1$-$c_8$.} \Description{Sample \sys{} knowledge
    graph for user Alice. Artifact nodes link to modality, person, and project
    nodes, and chunk nodes are labeled c1 through c8. Ambient appears as an
    extra modality. Timestamps are omitted.}
    \label{fig:kg-filtered-search-graph}
\end{figure}
Multiple chunks from the same artifact share a single artifact record, so artifact metadata is stored once.

\subsubsection{Filtration Types}

Semantic similarity is necessary but insufficient for personal memory retrieval.
Personal information management (PIM) studies confirm that users naturally scope queries by time, people, source
type, and project context~\cite{omniquery, foldersortags, emailrefinding}. We
identify four classes of queries where embedding distance alone is unreliable,
each resolved by a cheap metadata predicate.

\textbf{Temporal.} Users naturally anchor queries to dates, intervals,
    or relative time expressions (e.g., ``last week''). \sys{} resolves these
    into time windows and retains only artifacts within range.
\textbf{Modality.} Users often specify the source type
    (e.g., ``in my emails''); modality filtration narrows~$A_q$ accordingly.
\textbf{Person.} Queries like ``What did Bob say about the budget?''
    require scoping memory to a named individual. \sys{} filters on normalized
    participant names attached at ingest (e.g., meeting attendees, email
    senders).
\textbf{Project or workstream.} \sys{} filters on workstream tags
    derived from organizational structure (folder hierarchies, project spaces,
    recurring meetings).

For example, the query ``What did Alice say about the budget last Monday?''
yields predicates \texttt{person\,=\,'Alice' AND timestamp BETWEEN
'2025-03-24' AND '2025-03-25'}.
\rbtext{If the query yields no confident predicate, $A_q$ is the whole
corpus.}

\subsubsection{Query-Time Traversal}\label{sec:query-time-traversal}

Because each retrieval round requires a full ORAM access, \sys{} gets exactly
one shot at filtering; a variable number of rounds would leak query information,
and fixing the budget at multiple rounds would multiply ORAM latency and
bandwidth for every query, even when a single pass suffices.
The traversal stage therefore \rbtext{has the LLM assign a confidence} to each
predicate in~$F$: only
high-confidence predicates become hard constraints, and low-confidence ones are
dropped. If the resulting conjunction yields too few candidates, a
deterministic relaxation cascade drops predicates and widens temporal windows
until enough candidates remain.

Modality filtering restricts~$A_q$ to a single source class, but answers often
reside in a linked artifact of a different type (e.g., a meeting decision
captured in a follow-up email). \sys{} \emph{percolates}~$A_q$ along
source-link edges in the KG, adding linked artifacts of other modalities.
Without percolation, a correct modality predicate would still miss the chunk
that contains the answer.

\subsection{Oblivious Dreaming}
\label{sec:oblivious-dreaming}

Insecure agentic memory systems routinely perform heavy maintenance:
compaction and recursive summarization~\cite{memgpt,
raptor}, conflict resolution and deduplication~\cite{mem0, memos,
zep}, periodic graph refresh~\cite{zep, graphrag}, and neighborhood-scanning
reflection~\cite{generativeagents}. Under ORAM, each of these produces
distinctive access patterns that let the adversary observe \emph{when}
maintenance occurs and correlate it with recent activity. Any fixed maintenance
schedule must either scan the entire store to catch up, or leave stale state
lingering.

\sys{} addresses this through oblivious dreaming: the $\Dream()$ callback
(Algorithm~\ref{alg:opal}) executes inside the enclave after every request's ORAM
batches, expiring entries in enclave-resident metadata. No maintenance-specific storage access is
ever issued: the external trace is identical whether maintenance fires or
not~(\S\ref{sec:security}).

\leadpara{Adaptive retention.}\label{subsec:adaptive-retention} An ORAM database has a fixed capacity. One could double the store at fixed
intervals, but different users fill capacity at different rates, so any uniform
schedule either wastes space or falls behind; worse, the store grows without
bound and ORAM overhead ($O(\log N)$ per access) grows with it. The store must
therefore expire old items, and the retention policy must be user-agnostic: a
single configuration that adapts to any workload without per-user tuning. We
size the ORAM capacity to hold roughly one year of data, informed by PIM research showing that most personal searches target items from the past few weeks, with a long tail extending over months~\cite{oncefound, emaillifetime, refindfilesemailsweb}.

\sys{} assigns each ingested item a TTL measured in logical operations. Each query can also refresh the TTL of the top-$K$
results it retrieves, producing LRU-like behavior: recent and
recurrently accessed items survive, while cold items age out automatically.
This matches empirical personal-data access patterns, where re-access concentrates
heavily in the near past but retains a long tail~\cite{emailsearch,
emaillifetime, oncefound}.
Crucially, expiry is evaluated just-in-time on blocks
already in the stash, requiring no additional ORAM accesses, extending the
inline-expiry idea from oblivious messaging~\cite{myco}.

To target a steady-state store size~$N$, \sys{} sets
\[
\text{TTL}
=
\left\lceil
\frac{N}{w \cdot c \cdot \eta}
\right\rceil,
\qquad
\eta = 1 + (1-w)\ln n,
\]
where $w$ is the write ratio, $c$ is the average number of chunks generated per
logical item, and $n$ is the ANN fetch budget. All parameters are public or
derivable from public workload statistics; recall that the adversary already
knows operation types and counts under our threat
model~(\S\ref{subsec:threat-model-and-security-guarantees}), so the TTL itself
leaks nothing. The correction factor~$\eta$ accounts for lifetime refresh: as
the read fraction grows, items survive longer in expectation, so the base TTL
must shrink to maintain the same expected occupancy. Once the store reaches
capacity, expired items are deleted opportunistically as their slots are
reused, and the active
database converges toward~$N$ in steady state (\S\ref{subsec:oblivious-dreaming}
confirms convergence and that neither the ORAM nor its stash overflows in a
multi-year replay).

\subsection{Freshness and Rollback}
\label{sssec:freshness-rollback}

Each $\Query$ and $\Ingest$ carries a client secret~$k$ and the current
monotonic counter~$\mathsf{ctr}$. The enclave derives
$(k_{\mathsf{mac}}, k_{\mathsf{enc}}) \gets \mathsf{KDF}(k)$. The key
$k_{\mathsf{enc}}$ is the per-client encryption key used for both
$\ORAM_{\ANN}$ and $\ORAM_{\mathsf{Data}}$, as well as for sealing client
checkpoints.

Both ORAM trees maintain Merkle-style integrity over buckets~\cite{pathoram,
compass}; after each request the enclave forms a rollback record:
\[
    \begin{aligned}
    R = (&\mathsf{ctr}, \mathsf{root}_{\ANN}, \mathsf{root}_{\mathsf{Data}},\\
    &\mathsf{MAC}_{k_{\mathsf{mac}}}(\mathsf{ctr}, \mathsf{root}_{\ANN},
    \mathsf{root}_{\mathsf{Data}})).
    \end{aligned}
\]
Here $\mathsf{root}_{\ANN}$ and $\mathsf{root}_{\mathsf{Data}}$ are the current
ORAM roots, and the MAC is computed over a canonically encoded tuple. Any
accepted storage state for the request must match $R$ under the current counter,
giving per-request freshness.

Let $\kappa_i \le i$ denote the latest request index whose client state has been
checkpointed for that logical client by step~$i$ under the provider's residency
and checkpoint policy rather than the inter-enclave call pattern; $\kappa_i$
need not advance on every request and advances only when that client is evicted
from enclave RAM. When $\kappa_i$ advances, the enclave pads its resident client
state to a fixed size and writes
$C = \mathsf{AEAD.Enc}_{k_{\mathsf{enc}}}(\mathsf{ctr}, \mathsf{client\_state})$.
\rbtext{Because $C$ has a fixed size, every checkpoint writes the same amount of
data regardless of the interval between checkpoints.}
If a checkpoint is written under the current counter, the same~$\mathsf{ctr}$ is
bound into it, preventing independent rollback of the rollback record and the
sealed client state. Thus after request~$i$, sealed client state is guaranteed
fresh only up to~$\kappa_i$: a crash may recover the authenticated
checkpoint at~$\kappa_i$ rather than the current state.

Before any ORAM access, the enclave verifies the MAC in $R$, the checkpoint tag,
and both ORAM roots, and aborts on any mismatch. It also rejects any client
request whose $\mathsf{ctr}$ is not exactly one more than the counter in~$R$. A
crash that loses updates made after the latest checkpoint therefore halts that
client's store (fail-stop).
All client-to-enclave and inter-enclave calls use authenticated encrypted channels.
\rbtext{We model a user's devices as one logical client.}

\rbtext{A user's devices share $k$ and $\mathsf{ctr}$ over a channel the provider
does not control, such as direct device-to-device synchronization. Keeping $\mathsf{ctr}$ in provider storage,
even encrypted, would not suffice, because the provider could replay an older
counter together with its rollback record. A device with a stale counter is
rejected before any ORAM access, so it can synchronize and retry.}

\section{Synthetic Personal-Data Pipeline\label{sec:synthetic-personal-data}}

Evaluating \sys{} requires a corpus with realistic temporal dynamics and 
cross-modal causal structure, since both the maintenance mechanisms and
structured search filters in \sys{} depend on when artifacts arrive and how they
relate to one another. No existing dataset provides this: real personal data is
too sensitive to collect or share at scale, and prior
benchmarks~\cite{longmemeval, locomo, perltqa, timelineqa, openlifelogqa,
atmbench, amemgym, amabench, realmem, esmemeval, halumem, memoryagentbench,
lifebench, astrabench, memsim} and digital-footprint
pipelines~\cite{personatrace, privasis, converse, personabench} generate
personal artifacts but lack empirically grounded temporal arrival models or cross-modal dependencies.

Synthesizing a realistic substitute poses two challenges. First, modalities are
not independent: a meeting triggers follow-up emails, and generating each stream
in isolation yields a corpus that never exercises cross-modal retrieval filters.
Second, no single source covers the large parameter space (arrival rates, noise fractions, thread lengths, social-graph sizes).
\calibrationtablecontext{} map each empirical finding to the parameter it motivates. Building on PersonaBench's~\cite{personabench, personahub,
finepersonas} social-graph methodology, we address both challenges with a
four-stage pipeline that produces years of multimodal activity for any given
persona and social graph.

\leadpara{(1)~Life-state schedule.} A deterministic daily schedule
partitions the anchor person's time into behavioral states (e.g., asleep,
commuting, focus time, free time, weekend). Each state modulates the
baseline arrival intensity of every modality, hard-gating some modalities in
certain states (e.g., no meetings while asleep) while downweighting
others (e.g., half the message rate during focus time). Multiplier values
are motivated by empirical activity rates (\calibrationtableref).

\leadpara{(2)~Multivariate Hawkes process.} Hawkes processes are the
standard statistical framework for modeling self-exciting event streams~\cite{hawkes1971, hawkes1974}. We use a multivariate
Hawkes process~\cite{mulch} over six modalities: email, meeting,
document, query, message, and ambient. Each event can trigger
follow-on events in the same or different modalities (e.g., a meeting triggers
emails and documents), with
excitation parameters motivated by empirical communication studies
(\calibrationtableref). The life-state multipliers from Stage~1
modulate the baseline intensity without changing its weekday average. Further parameterization details are in
\hawkesdetailref.

\leadpara{(3)~Event planning.} A scenario planner constructs narrative arcs
from the social graph and assigns
each Stage~2 event its participants, topic, and,
when applicable, scenario membership. Contacts are drawn from
modality-eligible graph edges weighted by social circle.

\leadpara{(4)~Content generation.} Events are realized by modality-specific
LLM generators into concrete artifacts. The corpus mixes scenario-backed
events with standalone activity not tied to any scenario, a fraction of
which is routine noise (newsletters, gossip, spam) set per modality
(\calibrationtableref). Scenario-backed artifacts are written from shared Stage~3 scenario metadata.

For the timestamped retrieval events, we combine question categories from established agentic memory
benchmarks~\cite{longmemeval, locomo} with the natural scoping dimensions
identified by PIM research (time, people, source type,
project)~\cite{omniquery, foldersortags, emailrefinding}, yielding seven
categories: \emph{single-fact}, \emph{knowledge-update}, \emph{multi-hop},
\emph{temporal}, \emph{modality}, \emph{person}, and \emph{project}.
Artifact selection for question generation follows an empirically
motivated recency distribution~\cite{stuffiveseen, igarashi2022}
(\recencydetailref).
We focus on text-derived and structured artifacts; extending to native image
and video memory and validating against real user populations are important future
work~\cite{geminipersonalintelligence, googleastra, metaraybannametag}.
Figure~\ref{fig:synthetic-modality-timeline} shows the weekly cadence.

\begin{figure}[ht]
    \centering
    \includegraphics[width=0.96\linewidth]{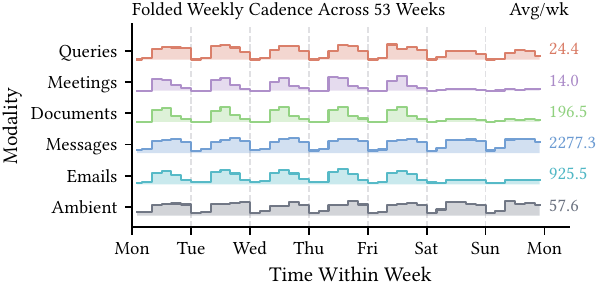}
    \caption{Average weekly cadence of each modality across the synthetic
    corpus. Messages and emails cluster into dense bursts, and meetings
    concentrate during work hours.} \Description{Average
    weekly cadence of each modality across the full synthetic corpus, showing
    bursty messages and emails and meetings concentrated in work hours.}
    \label{fig:synthetic-modality-timeline}
\end{figure}

\section{Security Proof Sketch}
\label{sec:security}

We sketch the proof here; the full proof is in \prooflocation. We rely on
Compass's~\cite{compass} batched-access ORAM proof for two
guarantees under Theorem~\ref{thm:opal-security}: ORAM executions with the same
public trace shape are indistinguishable, and tampering with
Merkle-authenticated ORAM storage is detected except with negligible
probability~\citep[Thm.~1; Supp.~Lemmas~4-5]{compass}.

For Condition~\textbf{\itshape 1\upshape)}, the only extra observable events are
fixed-size inter-enclave calls and the client's request and response, whose
lengths $\mathcal{L}$ reveals. For each request type, the number of such calls
and the ORAM batch sizes are fixed by the public parameters~$\Pi$; all other
computation is enclave-internal, and $\Dream()$ (\S\ref{sec:oblivious-dreaming})
performs only in-place maintenance on already materialized state, so it adds no
new event. Thus executions with the same leakage have the same public trace
shape, and indistinguishability follows from the Compass privacy guarantee and
authenticated encrypted channels for inter-enclave calls.

For Condition~\textbf{\itshape 2\upshape)}, the freshness and rollback mechanism
(\S\ref{sssec:freshness-rollback}) gives per-request freshness for the
ORAM-backed storage path and sealed client state. Compass's authenticated-storage
guarantee protects the ORAM trees, and requests that replay a committed counter are
rejected by the monotonic counter~$\mathsf{ctr}$. Inter-enclave calls run over the authenticated
encrypted channels assumed in the threat model.

\section{Evaluation}
\label{sec:evaluation}

\begin{figure*}[t!]
    \centering
    \includegraphics[width=0.74\linewidth]{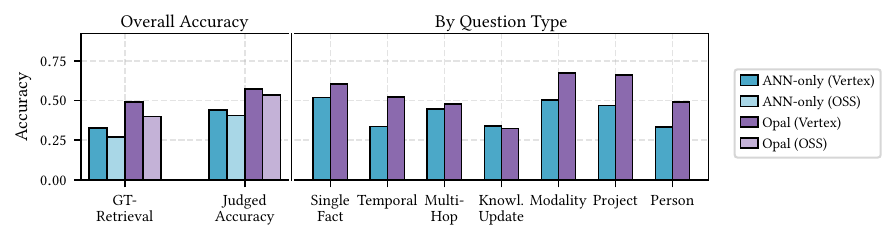}
        \caption{Accuracy comparison for ANN-only and \sys{}. Left:
        aggregate GT-retrieval rate and judged accuracy across Vertex and
        open-weight (OSS) stacks. Right: judged accuracy by question type (Vertex).}
        \label{fig:eval-accuracy}
    \Description{Accuracy comparison showing aggregate and per-category results
    for ANN-only and \sys{}.}
\end{figure*}

We evaluate whether \sys{} achieves access-pattern privacy without sacrificing accuracy or efficiency:
whether KG-filtered search recovers the accuracy that oblivious search sacrifices (\S\ref{subsec:accuracy}),
what the privacy guarantees cost in system overhead (\S\ref{subsec:bandwidth}--\S\ref{subsec:query-latency}),
and whether oblivious dreaming maintains steady-state capacity and tracks non-oblivious accuracy (\S\ref{subsec:oblivious-dreaming}).

\subsection{Implementation\label{subsec:implementation}}

The \sys{} protocol is implemented in ${\sim}$16K lines of Rust\footnote{We omit Ring ORAM's XOR optimization~\cite{ringoram} (our ORAM is
disk-backed). Ring ORAM metadata is stored in the sealed client state.} and the synthetic
personal-data pipeline (\S\ref{sec:synthetic-personal-data}) in ${\sim}$25K lines
of Python. ANN search uses IVF-PQ indexing~\cite{faiss}.
We use \texttt{ring}~\cite{ringcrate} for AES-256-GCM authenticated encryption,
the RustCrypto crates~\cite{rustcrypto} for SHA-256 ORAM integrity trees, HKDF-SHA256 key
derivation, and HMAC-SHA256 message authentication, and
\texttt{rand\_chacha}~\cite{randchacha} for cryptographic pseudorandom number
generation. Authenticated channels from clients to enclaves and between enclaves
use TLS via the Tinfoil SDK~\cite{tinfoilsdk}.
Inter-enclave messages for query and
ingestion are each padded to a fixed size;
sealed client state is padded to 49\,MiB. 

\subsection{Experiment Setup\label{subsec:experiment-setup}}

\leadpara{Baselines.} We evaluate \sys{} against five baselines.
\textbf{ANN-only} performs ANN retrieval without KG filtering, isolating the
accuracy gain of structured search. \rbtext{\textbf{Graphiti}~\cite{zep} is a graph-memory system with query-dependent retrieval, run in
plaintext as an accuracy baseline. We strengthen Graphiti with the same question time windows \sys{} uses.} \textbf{LIRE}~\cite{spfresh} is the incremental rebalancing
policy from SPFresh~\cite{spfresh}, used as the non-oblivious accuracy baseline for oblivious
dreaming (\S\ref{subsec:oblivious-dreaming}). \textbf{Plaintext \sys{}} runs the
same pipeline over unencrypted storage with selective fetches and the same IVF-PQ candidate selection, giving a lower bound on overhead. \textbf{In-Memory \sys{}} loads the entire encrypted database
into the enclave on every query and read-modify-writes it on ingestion, matching
the privacy guarantee of \sys{} without ORAM; this reflects prior enclave-backed
designs that keep the full database inside the trusted boundary~\cite{enclavedb,obliviate,hardidx} and Google's Private AI Compute memory~\cite{googleprivatememory2026}, whose team confirmed it is a reasonable implementation of their architecture.\footnote{We do not evaluate a secure Graphiti baseline: its retrieval, provenance, and deduplication paths would each require oblivious protection, and keeping the full graph in-enclave instead is what our In-Memory baseline captures.}

\leadpara{Deployment.} All enclaves are co-located on a single EU host with
Tinfoil~\cite{tinfoil} Intel TDX confidential VMs~\cite{tinfoilcvm}, attested
via Tinfoil's attestation proxy~\cite{tinfoilshim} and verification
stack~\cite{tinfoilverification}.
\sys{} runs in a Tinfoil
Container~\cite{tinfoilcontainers} (4~vCPUs, 4~GiB RAM) with encrypted drives for untrusted persistent storage.
The LLM
(\texttt{gpt-oss-20b}) and embedding (\texttt{nomic-embed-text}) enclaves are
model-serving confidential VMs that we deploy on Tinfoil from published
releases~\cite{tinfoilmodelidentity} (one NVIDIA~B200 GPU each, in confidential-computing
mode). For accuracy we additionally evaluate a \rbtext{Google Cloud
Vertex AI} backend (\texttt{gemini-3.1-flash-\allowbreak{}lite} with
\texttt{text-embedding-005}) as a representative frontier stack. 
To capture realistic WAN latency, the client runs on a
GCP \texttt{n2-standard-8} VM in \texttt{europe-west2} (London). Volume bandwidth on the \sys{} controller is
1.54~GB/s read and 0.42~GB/s write.

\leadpara{Parameters.}
We use standard retrieval parameters. Raw data is chunked
into 50-word segments with 10-word overlap~\cite{chunksize,ragbestpractices}. We
retrieve $n{=}200$ ANN candidates, rerank~\cite{dpr}, and pass the top $K{=}10$
to the LLM~\cite{lostinthemiddle}.
We use standard ORAM
configurations~\cite{compass} with an $L{=}14$ binary tree, the smallest that holds a
one-year retention target~\cite{oncefound, emaillifetime, refindfilesemailsweb}.
One year of synthetic data yields 1.3K
evaluation queries and 184K artifacts (49K emails, 121K messages, 10K documents,
3K ambient notes, 742 meetings), producing 215K chunks.

\subsection{Accuracy\label{subsec:accuracy}}

Figure~\ref{fig:eval-accuracy} reports \emph{GT-retrieval rate} (whether the
top-$K$ includes ground-truth chunks recorded at question generation) and
\emph{judged accuracy} (LLM judges whether retrieved chunks contain the correct
answer). KG-filtered search improves judged accuracy over ANN-only by 13\,pp
(Figure~\ref{fig:eval-accuracy}). We evaluate a
chronological replay of one year of synthetic data
(\S\ref{sec:synthetic-personal-data}). 

On Vertex, \sys{} reaches 57.3\% judged accuracy. OSS results are lower, but
the relative gains are consistent. \sys{}'s gains over
ANN-only are largest on project (19.1\,pp), temporal (18.7\,pp), person (15.9--19.5\,pp), and modality
(17.1--19.2\,pp) questions. \rbtext{\sys{}'s KG-filtered retrieval matches or slightly exceeds Graphiti (57.3\% vs.\ 52.2\%), mostly on questions that name a time period.}

Prior memory benchmarks~\cite{longmemeval,locomo,memorybank} place frontier
accuracy at 50--70\%; \sys{} matches this range despite our harder multi-modal
setting.

\subsection{Bandwidth\label{subsec:bandwidth}}

The key efficiency advantage of \sys{} is asymptotic: ORAM enables selective
fetches with $O(\log N)$ bandwidth, whereas matching \sys{}'s storage-side
privacy without ORAM requires reading the full encrypted database on every
access, giving $O(N)$ bandwidth.
Figure~\ref{fig:eval-bandwidth} shows this gap
widening as database size grows.

At 524K entries, \sys{} uses 11.6\,MiB of online traffic per query and
0.10\,MiB per ingestion, while the In-Memory baseline uses 4.51\,GiB and
9.03\,GiB respectively. Across the 1K-to-524K sweep, \sys{} uses 1.9--399$\times$
less online query bandwidth than the In-Memory baseline.

\begin{figure}[tbp]
    \centering
    \includegraphics[width=\linewidth]{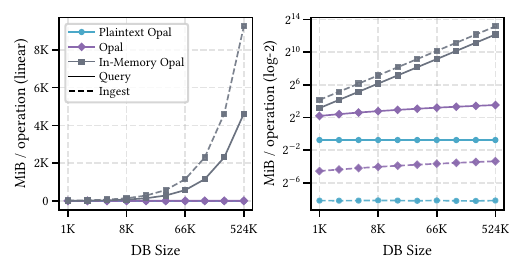}
    \caption{Online bandwidth per query and ingestion operation for \sys{},
    In-Memory \sys{}, and Plaintext \sys{} across database sizes. \sys{} is
    measured; the baselines are computed from their store sizes. Online traffic
    excludes Ring ORAM eviction and reshuffles. Left: linear
    scale; right: log-2 scale.} \Description{Bandwidth comparison across systems
    as corpus size increases.}
    \label{fig:eval-bandwidth}
\end{figure}

\begin{figure*}[t!]
    \centering
    \includegraphics[width=0.67\textwidth]{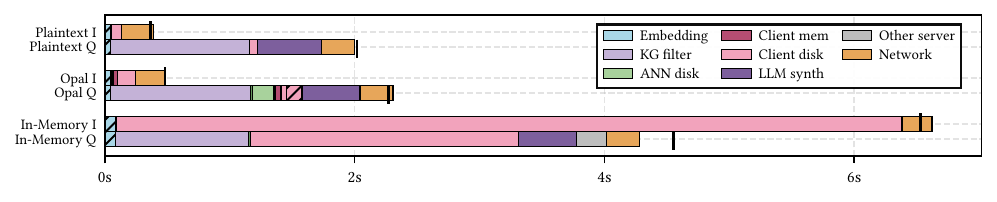}
        \caption{Single-client latency for query (Q) and ingestion (I).
        ANN denotes embeddings, Data raw encrypted chunks, and Client
        sealed per-user state; each split shows disk I/O versus
        in-memory compute.}
        \Description{Latency breakdown for query and ingest paths across systems.}
        \label{fig:eval-latency}
\end{figure*}

\subsection{Throughput and Cost\label{subsec:throughput}}

\S\ref{subsec:accuracy} shows that semantic search alone does not meet our
accuracy objectives. We therefore apply the same KG-filter traversal to the
baselines in all performance experiments and vary only the storage path, so the
comparisons reflect secure retrieval overhead rather than a lower-accuracy
pipeline.
Table~\ref{tab:throughput} summarizes serving performance and cost.

\leadpara{Throughput.}
We measure sustained multi-user throughput by packing concurrent users into
a single enclave until its memory limit or a 15\,s 95th-percentile latency cap is reached; as each user's request completes, its state is evicted and a new user is admitted, maintaining
maximal occupancy throughout.
This setup inherently captures memory pressure
from co-resident users and reflects steady-state serving capacity. Each
user's sealed client state (KG, ANN metadata, and ORAM clients)
is sealed and evicted to disk.
\sys{}
sustains 25.26~ingests/s and 2.55~queries/s, while In-Memory
\sys{} sustains 0.732~ingests/s and 0.314~queries/s
(34.5$\times$ and 8.2$\times$ higher, respectively). The
difference arises because In-Memory \sys{} must read its entire encrypted database on every request (and rewrite it on ingestion) to match \sys{}'s privacy guarantee, whereas \sys{}'s ORAM accesses touch only
$O(\log N)$ blocks.
Dividing the weighted throughput by the per-user request rate under the
synthetic workload (257K ingests and 1.3K queries per user-year), a single
enclave can serve 2{,}952 users for \sys{} versus 88.7 for In-Memory
\sys{}.
Since serving scales horizontally, this directly determines how many enclaves an
application provider must run.
One might ask whether In-Memory \sys{} could keep a single user resident to
avoid repeated loads. However, the full encrypted database is ${\sim}$44$\times$
larger than \sys{}'s per-user state, so a resident user takes
half of the engine's 4\,GiB yet at the per-user request rate from our synthetic
workload sits idle over 98\% of the time. Production serving systems evict idle
per-user state for exactly this reason~\cite{orleans, sessionstate}, which is
why both \sys{} and In-Memory \sys{} use this model in our evaluation.

\leadpara{Cost.}
The primary metric is annual application infrastructure cost: enclave compute plus
persistent storage.
We price compute at Tinfoil's posted rates~\cite{tinfoilcontainerspricing} and storage at \$0.01/GiB-month based on
published object-storage tiers~\cite{gcpcloudstoragepricing}.
At one million
users, \sys{}'s infrastructure cost is 18.7$\times$ lower than
In-Memory \sys{} (\$1.72M vs \$32.07M annually).
The gap is smaller than the throughput ratio because ORAM's tree structure increases
\sys{}'s per-user storage footprint by ${\sim}$1.6$\times$, but
compute dominates total infrastructure cost, so \sys{} retains a large advantage.

For a rough end-to-end estimate, we add embedding and LLM inference costs using
Tinfoil's token rates~\cite{tinfoilgptosspricing,tinfoilnomicpricing}. This
model/API cost is identical across systems since both use the same retrieval
pipeline, so it compresses the relative gap without changing the ranking.
Including it yields \$3.68M for \sys{} versus \$34.03M for
In-Memory \sys{} annually (9.25$\times$ lower overall). We focus on infrastructure cost because capability-adjusted inference cost
is falling roughly tenfold or more per year~\cite{epochinferencecost,
a16zllmflation, haiindex2025}, while cloud compute and storage pricing
improves by up to ${\sim}$15\% per VM generation~\cite{awsec2m6i, awsec2m7i}.

\begin{table}[t]
\centering
\caption{Serving metrics at one million users. Tput.\ is throughput at the
workload's mix of ingestions and queries; Users/Encl.\ divides it by one
user's request rate. Cost columns show infrastructure cost (compute + storage) and total cost (including model/API).}
\label{tab:throughput}
\small
\setlength{\tabcolsep}{4.5pt}
\begin{tabular}{@{}l c c c@{}}
\toprule
\textbf{System} & \textbf{Tput. (ops/s)} & \textbf{Users/Encl.} & \textbf{Cost (\$/yr)} \\
\midrule
\sys{} & 24.18 & 2{,}952 & \$1.72M / \$3.68M \\
In-Memory & 0.727 & 88.7 & \$32.07M / \$34.03M \\
\midrule
\textbf{Ratio} & 33.3$\times$ & 33.3$\times$ & 18.7$\times$ / 9.25$\times$ \\
\bottomrule
\end{tabular}
\end{table}

\subsection{Latency\label{subsec:query-latency}}

Figure~\ref{fig:eval-latency} breaks down single-client latency by stage, including client-side overhead (medians over 30 distinct requests).
On
queries, \sys{} adds 1.12$\times$ overhead over Plaintext \sys{}
(2.27\,s vs 2.02\,s) and is 2.01$\times$ faster than In-Memory \sys{}
(4.56\,s).
LLM calls still dominate the online path: KG filtering and answer
synthesis account for a median of 1.66\,s of the 2.27\,s query, while ORAM fetches and
client-state handling contribute 0.30\,s.
Ingestion shows the
sharper systems gap: \sys{} is 1.32$\times$ slower than Plaintext \sys{} but
13.6$\times$ faster than In-Memory \sys{}, at 0.48\,s versus 6.53\,s.
In-Memory \sys{}'s ingestion is dominated by rewriting its full encrypted store, which its queries skip.

\subsection{Oblivious Dreaming\label{subsec:oblivious-dreaming}}

We evaluate whether oblivious dreaming's deferred maintenance (\S\ref{sec:oblivious-dreaming}) degrades retrieval quality by replaying three years of synthetic data with a
one-year TTL, so the store reaches steady state after year~1 and remains at
capacity thereafter. Eager LIRE~\cite{spfresh} uses the same one-year TTL, and neither system refreshes TTLs at query time. Adaptive retention (\S\ref{subsec:adaptive-retention}) achieves tight
lifetime control: expired items have mean lifetime 0.97 years (std 0.004).
Over three years, later ingestions reuse 522K expired slots, and \sys{} issues no
maintenance request, whereas eager LIRE issues 6.1M full-vector requests. 
\sys{} achieves 51.7\% judged accuracy versus 51.7\% for eager LIRE on the Vertex stack,
matching the non-oblivious baseline without any index maintenance.
The ORAM stash stays bounded throughout.
\sys{}'s search index needs very little training data: an IVF-PQ index trained once on
just the first 0.26\% of a user's data serves all three years at the same
accuracy (51.7\%) as one trained on 32$\times$ more.

We also run a no-expiration control that keeps all three years (56.2\%;
Figure~\ref{fig:eval-dreaming}).  The 4.6\,pp gap is from deliberate expiry.

\begin{figure}[tbp]
    \centering
    \includegraphics[width=0.88\linewidth]{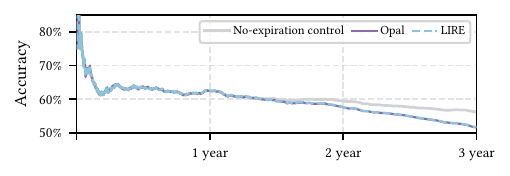}
    \caption{Three-year replay of oblivious dreaming: cumulative judged accuracy of \sys{}, LIRE, and the no-expiration control.}
    \Description{Cumulative judged accuracy over the three-year replay for \sys{} and
    a no-expiration control.}
    \label{fig:eval-dreaming}
\end{figure}

\section{Related work\label{sec:related-work}}

\leadpara{Confidential AI systems.} \cite{memtrust, trustedaiagents,
privatecloudcompute, googletrustedexecution, anthropictrustedexecution,
openanonymity} protect model inference. Protecting persistent memory additionally
requires hiding storage access patterns. Client-local
memory~\cite{windowsrecall, privacypreservingpersonalassistant, privgemo} stores personal data on the user's device but cannot match cloud capacity, durability, or multi-device access
(\S\ref{sec:intro})~\cite{stuffiveseen, mylifebits, diskfailuresrealworld}.
Deployed assistants therefore store memory cloud-natively~\cite{openaimemory,
geminipersonalintelligence}, and
\sys{} targets this setting. Our In-Memory baseline (\S\ref{subsec:experiment-setup}) is modeled after Google's
Private AI Compute memory~\cite{googleprivatememory2026}. Concurrent work~\cite{memtrust} offloads cold data
to untrusted storage and uses $k$-anonymity, which lacks provable guarantees and enables intersection
attacks~\cite{hangwithyourbuddies, quantifyingwebsearchprivacy,
leakageabuseattacksagainstsearchableencryption}.

\leadpara{Long-term agent memory.} A rich area spanning
research~\cite{generativeagents, mem0, memorybank, memgpt, memos, amem,
beyonddialoguetime, helloagain, secom, rmm, thinkinmemory, raptor, reflexion,
mem1, nemori, agentworkflowmemory, voyager, zep, hipporag, magma, arigraph,
editablememorygraph, sgmem, graphrag, lightrag, multimetarag} and deployed
products~\cite{googleastra, geminipersonalintelligence, notebooklm, memai,
chatgpt, openaimemory, awsagentcorememory, limitlessai, personalai}. \sys{} brings security to this space.

\leadpara{Synthetic personal data and memory benchmarks.} \cite{personatrace, privasis, converse,
personabench, mmlifelong, longmemeval, locomo, perltqa,
timelineqa, openlifelogqa, atmbench, realmem, amemgym, amabench, esmemeval,
halumem, memoryagentbench, lifebench, astrabench, memsim, beyondgoldfish,
odysseybench, appworld} provide valuable evaluation methodology but target
dialogues, agent traces, or video streams. Our pipeline extends this work to multi-modal
personal corpora with empirically motivated temporal dynamics
(\S\ref{sec:synthetic-personal-data}).

\leadpara{Private retrieval.} Embeddings leak nearly as much as query
text~\cite{informationleakageinembeddings, textembeddingsrevealalmostastext,
sentenceembeddingleaksmoreinformationthanyouexpect} (\S\ref{sec:intro}), making
access-pattern hiding essential. Prior work addresses private semantic
search~\cite{compass, pacmann, panther, tiptoe, remoterag} and oblivious
queries~\cite{oblidb, enigmap, snoopy}, but pure semantic search is insufficient
for agentic memory (\S\ref{sec:intro}). \sys{}'s KG-filtered retrieval addresses
this gap.

\leadpara{Dynamic ANN.} \cite{freshdiskann, spfresh, quake, lsmvec, crackivf}
supports update-heavy workloads as in \sys{}, but does not hide access patterns.
\sys{}'s oblivious dreaming (\S\ref{sec:oblivious-dreaming}) adds oblivious
retention to this space. 

\leadpara{Filtered ANN.} \cite{diskann, filtereddiskann, acorn, windowfilters,
caps, hqi, analyticdbv, mansw, irangegraph, pase, vbase, airship, serf,
surveyfann, pinecone, weaviate, spann, milvus, compass, ung, nhq, hqann}
restricts search to subsets matching metadata predicates, typically for speed.
KG-filtered search (\S\ref{sec:kg-filtered-search}) is a special case of ANN
pre-filtering; \sys{} uses it for accuracy.

\leadpara{Trusted storage.} TEE-I/O~\cite{reallysafefastconfidentialio,
dmtfspdm, pcietdisp, pcieide} hides storage commands by trusting the device.
However, this significantly enlarges the TCB~\cite{inteltdxconnectteeio}, and
recent attacks on confidential-I/O stacks~\cite{selfencryptingdeception,
formalanalysisofspdm, certidevu404544} demonstrate that extending trust to
storage hardware is premature. In \sys{}, storage remains untrusted, and access
patterns are hidden cryptographically via ORAM~\cite{oram96, oram90,
binarytreeoram, pathoram, ringoram}.

\section{Conclusion}
\label{sec:conclusion}

\sys{} is a private memory system for personal AI that keeps data-dependent
reasoning inside trusted enclaves while exposing only oblivious access patterns
to storage. KG-filtered search and oblivious dreaming enable accurate, efficient
retrieval, achieving an order of magnitude
higher throughput and lower cost than the state of the art. \sys{}
takes a significant step
towards private long-term memory at scale.

\ifpearl\else \ifdefined\mlsyssubmission
\else
\leadpara{Acknowledgments.} We thank Wei Huang and Marci Meingast
for early discussions and product-side insights at Google DeepMind and Google
AI. We thank Tanya Verma and Jules Drean at Tinfoil for their responsive and
ongoing confidential VM support. We also thank the students in the Sky security
group, especially Corban Villa, Diogo Antunes, and Yiping Ma, for their
feedback. This work is supported by the Amazon AI PhD Fellowship and by gifts
from Accenture, Amazon, AMD, Anyscale, Broadcom, Google, IBM, Intel, Intesa
Sanpaolo, Lambda, Lightspeed, Mibura, NVIDIA, Samsung SDS, and SAP.
\fi
\fi

\bibliographystyle{ACM-Reference-Format}
\bibliography{references}


\begin{thebibliography}{270}


\ifx \showCODEN    \undefined \def \showCODEN     #1{\unskip}     \fi
\ifx \showISBNx    \undefined \def \showISBNx     #1{\unskip}     \fi
\ifx \showISBNxiii \undefined \def \showISBNxiii  #1{\unskip}     \fi
\ifx \showISSN     \undefined \def \showISSN      #1{\unskip}     \fi
\ifx \showLCCN     \undefined \def \showLCCN      #1{\unskip}     \fi
\ifx \shownote     \undefined \def \shownote      #1{#1}          \fi
\ifx \showarticletitle \undefined \def \showarticletitle #1{#1}   \fi
\ifx \showURL      \undefined \def \showURL       {\relax}        \fi
\providecommand\bibfield[2]{#2}
\providecommand\bibinfo[2]{#2}
\providecommand\natexlab[1]{#1}
\providecommand\showeprint[2][]{arXiv:#2}

\bibitem[{1X Technologies}(2026)]%
        {1x}
\bibfield{author}{\bibinfo{person}{{1X Technologies}}.}
  \bibinfo{year}{2026}\natexlab{}.
\newblock \bibinfo{title}{1x}.
\newblock \bibinfo{howpublished}{\url{https://www.1x.tech/}}.
\newblock


\bibitem[Ahmad et~al\mbox{.}(2018)]%
        {obliviate}
\bibfield{author}{\bibinfo{person}{Adil Ahmad}, \bibinfo{person}{Kyungtae Kim},
  \bibinfo{person}{Muhammad~Ihsanulhaq Sarfaraz}, {and}
  \bibinfo{person}{Byoungyoung Lee}.} \bibinfo{year}{2018}\natexlab{}.
\newblock \showarticletitle{{OBLIVIATE:} {A} Data Oblivious Filesystem for
  Intel {SGX}}. In \bibinfo{booktitle}{\emph{{NDSS}}}. \bibinfo{publisher}{The
  Internet Society}.
\newblock


\bibitem[Ai et~al\mbox{.}(2017)]%
        {emailsearch}
\bibfield{author}{\bibinfo{person}{Qingyao Ai}, \bibinfo{person}{Susan~T.
  Dumais}, \bibinfo{person}{Nick Craswell}, {and} \bibinfo{person}{Daniel~J.
  Liebling}.} \bibinfo{year}{2017}\natexlab{}.
\newblock \showarticletitle{Characterizing Email Search using Large-scale
  Behavioral Logs and Surveys}. In \bibinfo{booktitle}{\emph{{WWW}}}.
  \bibinfo{publisher}{{ACM}}, \bibinfo{pages}{1511--1520}.
\newblock


\bibitem[Alrashed et~al\mbox{.}(2018)]%
        {emaillifetime}
\bibfield{author}{\bibinfo{person}{Tarfah Alrashed},
  \bibinfo{person}{Ahmed~Hassan Awadallah}, {and} \bibinfo{person}{Susan~T.
  Dumais}.} \bibinfo{year}{2018}\natexlab{}.
\newblock \showarticletitle{The Lifetime of Email Messages: {A} Large-Scale
  Analysis of Email Revisitation}. In \bibinfo{booktitle}{\emph{{CHIIR}}}.
  \bibinfo{publisher}{{ACM}}, \bibinfo{pages}{120--129}.
\newblock


\bibitem[{Amazon Web Services}(2024)]%
        {awsnitrosidechannels}
\bibfield{author}{\bibinfo{person}{{Amazon Web Services}}.}
  \bibinfo{year}{2024}\natexlab{}.
\newblock \bibinfo{title}{The Security Design of the {AWS} Nitro System}.
\newblock
  \bibinfo{howpublished}{\url{https://docs.aws.amazon.com/whitepapers/latest/security-design-of-aws-nitro-system/security-design-of-aws-nitro-system.html}}.
\newblock


\bibitem[{Amazon Web Services}(2026a)]%
        {awsec2m6i}
\bibfield{author}{\bibinfo{person}{{Amazon Web Services}}.}
  \bibinfo{year}{2026}\natexlab{a}.
\newblock \bibinfo{title}{Amazon {EC2} {M6i} Instances}.
\newblock
  \bibinfo{howpublished}{\url{https://aws.amazon.com/ec2/instance-types/m6i/}}.
\newblock


\bibitem[{Amazon Web Services}(2026b)]%
        {awsec2m7i}
\bibfield{author}{\bibinfo{person}{{Amazon Web Services}}.}
  \bibinfo{year}{2026}\natexlab{b}.
\newblock \bibinfo{title}{Amazon {EC2} {M7i} Instances}.
\newblock
  \bibinfo{howpublished}{\url{https://aws.amazon.com/ec2/instance-types/m7i/}}.
\newblock


\bibitem[Anderson and Schooler(1991)]%
        {anderson1991}
\bibfield{author}{\bibinfo{person}{John~R. Anderson} {and}
  \bibinfo{person}{Lael~J. Schooler}.} \bibinfo{year}{1991}\natexlab{}.
\newblock \showarticletitle{Reflections of the Environment in Memory}.
\newblock \bibinfo{journal}{\emph{Psychological Science}} \bibinfo{volume}{2},
  \bibinfo{number}{6} (\bibinfo{year}{1991}), \bibinfo{pages}{396--408}.
\newblock
\href{https://doi.org/10.1111/j.1467-9280.1991.tb00174.x}{doi:\nolinkurl{10.1111/j.1467-9280.1991.tb00174.x}}


\bibitem[Angel et~al\mbox{.}(2023)]%
        {nimble}
\bibfield{author}{\bibinfo{person}{Sebastian Angel}, \bibinfo{person}{Aditya
  Basu}, \bibinfo{person}{Weidong Cui}, \bibinfo{person}{Trent Jaeger},
  \bibinfo{person}{Stella Lau}, \bibinfo{person}{Srinath T.~V. Setty}, {and}
  \bibinfo{person}{Sudheesh Singanamalla}.} \bibinfo{year}{2023}\natexlab{}.
\newblock \showarticletitle{Nimble: Rollback Protection for Confidential Cloud
  Services}. In \bibinfo{booktitle}{\emph{{OSDI}}}.
  \bibinfo{publisher}{{USENIX} Association}, \bibinfo{pages}{193--208}.
\newblock


\bibitem[Anokhin et~al\mbox{.}(2025)]%
        {arigraph}
\bibfield{author}{\bibinfo{person}{Petr Anokhin}, \bibinfo{person}{Nikita
  Semenov}, \bibinfo{person}{Artyom~Y. Sorokin}, \bibinfo{person}{Dmitry
  Evseev}, \bibinfo{person}{Andrey Kravchenko}, \bibinfo{person}{Mikhail
  Burtsev}, {and} \bibinfo{person}{Evgeny Burnaev}.}
  \bibinfo{year}{2025}\natexlab{}.
\newblock \showarticletitle{AriGraph: Learning Knowledge Graph World Models
  with Episodic Memory for {LLM} Agents}. In
  \bibinfo{booktitle}{\emph{{IJCAI}}}. \bibinfo{publisher}{ijcai.org},
  \bibinfo{pages}{12--20}.
\newblock


\bibitem[Anthropic(2025)]%
        {anthropictrustedexecution}
\bibfield{author}{\bibinfo{person}{Anthropic}.}
  \bibinfo{year}{2025}\natexlab{}.
\newblock \bibinfo{title}{Confidential Inference via Trusted Virtual Machines}.
\newblock
  \bibinfo{howpublished}{\url{https://www.anthropic.com/research/confidential-inference-trusted-vms}}.
\newblock


\bibitem[Appenzeller(2024)]%
        {a16zllmflation}
\bibfield{author}{\bibinfo{person}{Guido Appenzeller}.}
  \bibinfo{year}{2024}\natexlab{}.
\newblock \bibinfo{title}{Welcome to {LLMflation}: {LLM} Inference Cost Is
  Going Down Fast}.
\newblock
  \bibinfo{howpublished}{\url{https://a16z.com/llmflation-llm-inference-cost/}}.
\newblock


\bibitem[Apple(2024)]%
        {privatecloudcompute}
\bibfield{author}{\bibinfo{person}{Apple}.} \bibinfo{year}{2024}\natexlab{}.
\newblock \showarticletitle{Private Cloud Compute}.
\newblock
  \bibinfo{howpublished}{\url{https://security.apple.com/blog/private-cloud-compute/}}.
\newblock \bibinfo{journal}{\emph{Apple Security Blog}} (\bibinfo{year}{2024}).
\newblock


\bibitem[Argilla(2024)]%
        {finepersonas}
\bibfield{author}{\bibinfo{person}{Argilla}.} \bibinfo{year}{2024}\natexlab{}.
\newblock \bibinfo{title}{{FinePersonas}: Synthetic Email Conversations}.
\newblock
  \bibinfo{howpublished}{\url{https://huggingface.co/datasets/argilla/FinePersonas-Synthetic-Email-Conversations}}.
\newblock


\bibitem[Baylor et~al\mbox{.}(2017)]%
        {tfx}
\bibfield{author}{\bibinfo{person}{Denis Baylor}, \bibinfo{person}{Eric Breck},
  \bibinfo{person}{Heng{-}Tze Cheng}, \bibinfo{person}{Noah Fiedel},
  \bibinfo{person}{Chuan~Yu Foo}, \bibinfo{person}{Zakaria Haque},
  \bibinfo{person}{Salem Haykal}, \bibinfo{person}{Mustafa Ispir},
  \bibinfo{person}{Vihan Jain}, \bibinfo{person}{Levent Koc},
  \bibinfo{person}{Chiu~Yuen Koo}, \bibinfo{person}{Lukasz Lew},
  \bibinfo{person}{Clemens Mewald}, \bibinfo{person}{Akshay~Naresh Modi},
  \bibinfo{person}{Neoklis Polyzotis}, \bibinfo{person}{Sukriti Ramesh},
  \bibinfo{person}{Sudip Roy}, \bibinfo{person}{Steven~Euijong Whang},
  \bibinfo{person}{Martin Wicke}, \bibinfo{person}{Jarek Wilkiewicz},
  \bibinfo{person}{Xin Zhang}, {and} \bibinfo{person}{Martin Zinkevich}.}
  \bibinfo{year}{2017}\natexlab{}.
\newblock \showarticletitle{{TFX:} {A} TensorFlow-Based Production-Scale
  Machine Learning Platform}. In \bibinfo{booktitle}{\emph{{KDD}}}.
  \bibinfo{publisher}{{ACM}}, \bibinfo{pages}{1387--1395}.
\newblock


\bibitem[{BBC News}(2026)]%
        {metaraybanreview}
\bibfield{author}{\bibinfo{person}{{BBC News}}.}
  \bibinfo{year}{2026}\natexlab{}.
\newblock \bibinfo{title}{Meta smart glasses: Workers reviewing intimate
  content}.
\newblock
  \bibinfo{howpublished}{\url{https://www.bbc.com/news/articles/c0q33nvj0qpo}}.
\newblock


\bibitem[Bernstein et~al\mbox{.}(2014)]%
        {orleans}
\bibfield{author}{\bibinfo{person}{Philip~A. Bernstein},
  \bibinfo{person}{Sergey Bykov}, \bibinfo{person}{Alan Geller},
  \bibinfo{person}{Gabriel Kliot}, {and} \bibinfo{person}{Jorgen Thelin}.}
  \bibinfo{year}{2014}\natexlab{}.
\newblock \bibinfo{booktitle}{\emph{Orleans: Distributed Virtual Actors for
  Programmability and Scalability}}.
\newblock \bibinfo{type}{{T}echnical {R}eport} MSR-TR-2014-41.
  \bibinfo{institution}{Microsoft Research}.
\newblock


\bibitem[Bhat et~al\mbox{.}(2025)]%
        {chunksize}
\bibfield{author}{\bibinfo{person}{Sinchana~Ramakanth Bhat},
  \bibinfo{person}{Max Rudat}, \bibinfo{person}{Jannis Spiekermann}, {and}
  \bibinfo{person}{Nicolas Flores-Herr}.} \bibinfo{year}{2025}\natexlab{}.
\newblock \showarticletitle{Rethinking Chunk Size For Long-Document Retrieval:
  A Multi-Dataset Analysis}.
\newblock \bibinfo{journal}{\emph{arXiv preprint arXiv:2505.21700}}
  (\bibinfo{year}{2025}).
\newblock


\bibitem[Bian et~al\mbox{.}(2026)]%
        {realmem}
\bibfield{author}{\bibinfo{person}{Haonan Bian}, \bibinfo{person}{Zhiyuan Yao},
  \bibinfo{person}{Sen Hu}, \bibinfo{person}{Zishan Xu},
  \bibinfo{person}{Shaolei Zhang}, \bibinfo{person}{Yifu Guo},
  \bibinfo{person}{Ziliang Yang}, \bibinfo{person}{Xueran Han},
  \bibinfo{person}{Huacan Wang}, {and} \bibinfo{person}{Ronghao Chen}.}
  \bibinfo{year}{2026}\natexlab{}.
\newblock \showarticletitle{RealMem: Benchmarking LLMs in Real-World
  Memory-Driven Interaction}.
\newblock \bibinfo{journal}{\emph{arXiv preprint arXiv:2601.06966}}
  (\bibinfo{year}{2026}).
\newblock


\bibitem[Bittau et~al\mbox{.}(2017)]%
        {prochlo}
\bibfield{author}{\bibinfo{person}{Andrea Bittau}, \bibinfo{person}{{\'{U}}lfar
  Erlingsson}, \bibinfo{person}{Petros Maniatis}, \bibinfo{person}{Ilya
  Mironov}, \bibinfo{person}{Ananth Raghunathan}, \bibinfo{person}{David Lie},
  \bibinfo{person}{Mitch Rudominer}, \bibinfo{person}{Ushasree Kode},
  \bibinfo{person}{Julien Tinn{\'{e}}s}, {and} \bibinfo{person}{Bernhard
  Seefeld}.} \bibinfo{year}{2017}\natexlab{}.
\newblock \showarticletitle{Prochlo: Strong Privacy for Analytics in the
  Crowd}. In \bibinfo{booktitle}{\emph{{SOSP}}}. \bibinfo{publisher}{{ACM}},
  \bibinfo{pages}{441--459}.
\newblock


\bibitem[Bodea et~al\mbox{.}(2025)]%
        {trustedaiagents}
\bibfield{author}{\bibinfo{person}{Teofil Bodea}, \bibinfo{person}{Masanori
  Misono}, \bibinfo{person}{Julian Pritzi}, \bibinfo{person}{Patrick Sabanic},
  \bibinfo{person}{Thore Sommer}, \bibinfo{person}{Harshavardhan Unnibhavi},
  \bibinfo{person}{David Schall}, \bibinfo{person}{Nuno Santos},
  \bibinfo{person}{Dimitrios Stavrakakis}, {and} \bibinfo{person}{Pramod
  Bhatotia}.} \bibinfo{year}{2025}\natexlab{}.
\newblock \showarticletitle{Trusted AI Agents in the Cloud}.
\newblock \bibinfo{journal}{\emph{arXiv preprint arXiv:2512.05951}}
  (\bibinfo{year}{2025}).
\newblock


\bibitem[Borrello et~al\mbox{.}(2022)]%
        {aepicleak}
\bibfield{author}{\bibinfo{person}{Pietro Borrello}, \bibinfo{person}{Andreas
  Kogler}, \bibinfo{person}{Martin Schwarzl}, \bibinfo{person}{Moritz Lipp},
  \bibinfo{person}{Daniel Gruss}, {and} \bibinfo{person}{Michael Schwarz}.}
  \bibinfo{year}{2022}\natexlab{}.
\newblock \showarticletitle{{\AE}PIC Leak: Architecturally Leaking
  Uninitialized Data from the Microarchitecture}. In
  \bibinfo{booktitle}{\emph{{USENIX} Security Symposium}}.
  \bibinfo{publisher}{{USENIX} Association}, \bibinfo{pages}{3917--3934}.
\newblock


\bibitem[Bulck et~al\mbox{.}(2018)]%
        {foreshadow}
\bibfield{author}{\bibinfo{person}{Jo~Van Bulck}, \bibinfo{person}{Marina
  Minkin}, \bibinfo{person}{Ofir Weisse}, \bibinfo{person}{Daniel Genkin},
  \bibinfo{person}{Baris Kasikci}, \bibinfo{person}{Frank Piessens},
  \bibinfo{person}{Mark Silberstein}, \bibinfo{person}{Thomas~F. Wenisch},
  \bibinfo{person}{Yuval Yarom}, {and} \bibinfo{person}{Raoul Strackx}.}
  \bibinfo{year}{2018}\natexlab{}.
\newblock \showarticletitle{Foreshadow: Extracting the Keys to the Intel {SGX}
  Kingdom with Transient Out-of-Order Execution}. In
  \bibinfo{booktitle}{\emph{{USENIX} Security Symposium}}.
  \bibinfo{publisher}{{USENIX} Association}, \bibinfo{pages}{991--1008}.
\newblock


\bibitem[Bulck et~al\mbox{.}(2017)]%
        {tellingyoursecretswithoutpagefaults}
\bibfield{author}{\bibinfo{person}{Jo~Van Bulck}, \bibinfo{person}{Nico
  Weichbrodt}, \bibinfo{person}{R{\"{u}}diger Kapitza}, \bibinfo{person}{Frank
  Piessens}, {and} \bibinfo{person}{Raoul Strackx}.}
  \bibinfo{year}{2017}\natexlab{}.
\newblock \showarticletitle{Telling Your Secrets without Page Faults: Stealthy
  Page Table-Based Attacks on Enclaved Execution}. In
  \bibinfo{booktitle}{\emph{{USENIX} Security Symposium}}.
  \bibinfo{publisher}{{USENIX} Association}, \bibinfo{pages}{1041--1056}.
\newblock


\bibitem[Cai et~al\mbox{.}(2024)]%
        {ung}
\bibfield{author}{\bibinfo{person}{Yuzheng Cai}, \bibinfo{person}{Jiayang Shi},
  \bibinfo{person}{Yizhuo Chen}, {and} \bibinfo{person}{Weiguo Zheng}.}
  \bibinfo{year}{2024}\natexlab{}.
\newblock \showarticletitle{Navigating Labels and Vectors: {A} Unified Approach
  to Filtered Approximate Nearest Neighbor Search}.
\newblock \bibinfo{journal}{\emph{Proc. {ACM} Manag. Data}}
  \bibinfo{volume}{2}, \bibinfo{number}{6} (\bibinfo{year}{2024}),
  \bibinfo{pages}{246:1--246:27}.
\newblock


\bibitem[Cao et~al\mbox{.}(2021)]%
        {multitasking}
\bibfield{author}{\bibinfo{person}{Hancheng Cao}, \bibinfo{person}{Chia{-}Jung
  Lee}, \bibinfo{person}{Shamsi~T. Iqbal}, \bibinfo{person}{Mary Czerwinski},
  \bibinfo{person}{Priscilla N.~Y. Wong}, \bibinfo{person}{Sean Rintel},
  \bibinfo{person}{Brent~J. Hecht}, \bibinfo{person}{Jaime Teevan}, {and}
  \bibinfo{person}{Longqi Yang}.} \bibinfo{year}{2021}\natexlab{}.
\newblock \showarticletitle{Large Scale Analysis of Multitasking Behavior
  During Remote Meetings}. In \bibinfo{booktitle}{\emph{{CHI}}}.
  \bibinfo{publisher}{{ACM}}, \bibinfo{pages}{448:1--448:13}.
\newblock


\bibitem[Carron et~al\mbox{.}(2016)]%
        {dunbarnumber}
\bibfield{author}{\bibinfo{person}{P{\'{a}}draig~Mac Carron},
  \bibinfo{person}{Kimmo Kaski}, {and} \bibinfo{person}{Robin Dunbar}.}
  \bibinfo{year}{2016}\natexlab{}.
\newblock \showarticletitle{Calling Dunbar's numbers}.
\newblock \bibinfo{journal}{\emph{Soc. Networks}}  \bibinfo{volume}{47}
  (\bibinfo{year}{2016}), \bibinfo{pages}{151--155}.
\newblock


\bibitem[Cash et~al\mbox{.}(2016)]%
        {leakageabuseattacksagainstsearchableencryption}
\bibfield{author}{\bibinfo{person}{David Cash}, \bibinfo{person}{Paul Grubbs},
  \bibinfo{person}{Jason Perry}, {and} \bibinfo{person}{Thomas Ristenpart}.}
  \bibinfo{year}{2016}\natexlab{}.
\newblock \showarticletitle{Leakage-Abuse Attacks Against Searchable
  Encryption}.
\newblock \bibinfo{journal}{\emph{{IACR} Cryptol. ePrint Arch.}}
  \bibinfo{volume}{2016} (\bibinfo{year}{2016}), \bibinfo{pages}{718}.
\newblock


\bibitem[{CERT Coordination Center}(2025)]%
        {certidevu404544}
\bibfield{author}{\bibinfo{person}{{CERT Coordination Center}}.}
  \bibinfo{year}{2025}\natexlab{}.
\newblock \bibinfo{title}{Vulnerability Note {VU\#404544}: Vulnerabilities
  identified in {PCIe} Integrity and Data Encryption ({IDE}) protocol
  specification}.
\newblock \bibinfo{howpublished}{\url{https://kb.cert.org/vuls/id/404544}}.
\newblock


\bibitem[Chen et~al\mbox{.}(2025)]%
        {halumem}
\bibfield{author}{\bibinfo{person}{Ding Chen}, \bibinfo{person}{Simin Niu},
  \bibinfo{person}{Kehang Li}, \bibinfo{person}{Peng Liu},
  \bibinfo{person}{Xiangping Zheng}, \bibinfo{person}{Bo Tang},
  \bibinfo{person}{Xinchi Li}, \bibinfo{person}{Feiyu Xiong}, {and}
  \bibinfo{person}{Zhiyu Li}.} \bibinfo{year}{2025}\natexlab{}.
\newblock \showarticletitle{HaluMem: Evaluating Hallucinations in Memory
  Systems of Agents}.
\newblock \bibinfo{journal}{\emph{arXiv preprint arXiv:2511.03506}}
  (\bibinfo{year}{2025}).
\newblock


\bibitem[Chen et~al\mbox{.}(2026a)]%
        {mmlifelong}
\bibfield{author}{\bibinfo{person}{Guo Chen}, \bibinfo{person}{Lidong Lu},
  \bibinfo{person}{Yicheng Liu}, \bibinfo{person}{Liangrui Dong},
  \bibinfo{person}{Lidong Zou}, \bibinfo{person}{Jixin Lv},
  \bibinfo{person}{Zhenquan Li}, \bibinfo{person}{Xinyi Mao},
  \bibinfo{person}{Baoqi Pei}, \bibinfo{person}{Shihao Wang},
  \bibinfo{person}{Zhiqi Li}, \bibinfo{person}{Karan Sapra},
  \bibinfo{person}{Fuxiao Liu}, \bibinfo{person}{Yin-Dong Zheng},
  \bibinfo{person}{Yifei Huang}, \bibinfo{person}{Limin Wang},
  \bibinfo{person}{Zhiding Yu}, \bibinfo{person}{Andrew Tao},
  \bibinfo{person}{Guilin Liu}, {and} \bibinfo{person}{Tong Lu}.}
  \bibinfo{year}{2026}\natexlab{a}.
\newblock \showarticletitle{Towards Multimodal Lifelong Understanding: A
  Dataset and Agentic Baseline}.
\newblock \bibinfo{journal}{\emph{arXiv preprint arXiv:2603.05484}}
  (\bibinfo{year}{2026}).
\newblock


\bibitem[Chen et~al\mbox{.}(2016)]%
        {facebookrealtime}
\bibfield{author}{\bibinfo{person}{Guoqiang~Jerry Chen},
  \bibinfo{person}{Janet~L. Wiener}, \bibinfo{person}{Shridhar Iyer},
  \bibinfo{person}{Anshul Jaiswal}, \bibinfo{person}{Ran Lei},
  \bibinfo{person}{Nikhil Simha}, \bibinfo{person}{Wei Wang},
  \bibinfo{person}{Kevin Wilfong}, \bibinfo{person}{Tim Williamson}, {and}
  \bibinfo{person}{Serhat Yilmaz}.} \bibinfo{year}{2016}\natexlab{}.
\newblock \showarticletitle{Realtime Data Processing at Facebook}. In
  \bibinfo{booktitle}{\emph{{SIGMOD} Conference}}. \bibinfo{publisher}{{ACM}},
  \bibinfo{pages}{1087--1098}.
\newblock


\bibitem[Chen et~al\mbox{.}(2021)]%
        {spann}
\bibfield{author}{\bibinfo{person}{Qi Chen}, \bibinfo{person}{Bing Zhao},
  \bibinfo{person}{Haidong Wang}, \bibinfo{person}{Mingqin Li},
  \bibinfo{person}{Chuanjie Liu}, \bibinfo{person}{Zengzhong Li},
  \bibinfo{person}{Mao Yang}, {and} \bibinfo{person}{Jingdong Wang}.}
  \bibinfo{year}{2021}\natexlab{}.
\newblock \showarticletitle{{SPANN:} Highly-efficient Billion-scale Approximate
  Nearest Neighborhood Search}. In \bibinfo{booktitle}{\emph{NeurIPS}}.
  \bibinfo{pages}{5199--5212}.
\newblock


\bibitem[Chen et~al\mbox{.}(2026b)]%
        {esmemeval}
\bibfield{author}{\bibinfo{person}{Tiantian Chen}, \bibinfo{person}{Jiaqi Lu},
  \bibinfo{person}{Ying Shen}, {and} \bibinfo{person}{Lin Zhang}.}
  \bibinfo{year}{2026}\natexlab{b}.
\newblock \showarticletitle{ES-MemEval: Benchmarking Conversational Agents on
  Personalized Long-Term Emotional Support}.
\newblock \bibinfo{journal}{\emph{arXiv preprint arXiv:2602.01885}}
  (\bibinfo{year}{2026}).
\newblock


\bibitem[Cheng et~al\mbox{.}(2025)]%
        {remoterag}
\bibfield{author}{\bibinfo{person}{Yihang Cheng}, \bibinfo{person}{Lan Zhang},
  \bibinfo{person}{Junyang Wang}, \bibinfo{person}{Mu Yuan}, {and}
  \bibinfo{person}{Yunhao Yao}.} \bibinfo{year}{2025}\natexlab{}.
\newblock \showarticletitle{RemoteRAG: {A} Privacy-Preserving {LLM} Cloud {RAG}
  Service}. In \bibinfo{booktitle}{\emph{{ACL} (Findings)}}
  \emph{(\bibinfo{series}{Findings of {ACL}})}. \bibinfo{publisher}{Association
  for Computational Linguistics}, \bibinfo{pages}{3820--3837}.
\newblock


\bibitem[Cheng et~al\mbox{.}(2026)]%
        {lifebench}
\bibfield{author}{\bibinfo{person}{Zihao Cheng}, \bibinfo{person}{Weixin Wang},
  \bibinfo{person}{Yu Zhao}, \bibinfo{person}{Ziyang Ren},
  \bibinfo{person}{Jiaxuan Chen}, \bibinfo{person}{Ruiyang Xu},
  \bibinfo{person}{Shuai Huang}, \bibinfo{person}{Yang Chen},
  \bibinfo{person}{Guowei Li}, \bibinfo{person}{Mengshi Wang},
  \bibinfo{person}{Yi Xie}, \bibinfo{person}{Ren Zhu}, \bibinfo{person}{Zeren
  Jiang}, \bibinfo{person}{Keda Lu}, \bibinfo{person}{Yihong Li},
  \bibinfo{person}{Xiaoliang Wang}, \bibinfo{person}{Liwei Liu}, {and}
  \bibinfo{person}{Cam-Tu Nguyen}.} \bibinfo{year}{2026}\natexlab{}.
\newblock \showarticletitle{LifeBench: A Benchmark for Long-Horizon
  Multi-Source Memory}.
\newblock \bibinfo{journal}{\emph{arXiv preprint arXiv:2603.03781}}
  (\bibinfo{year}{2026}).
\newblock


\bibitem[Chhikara et~al\mbox{.}(2025)]%
        {mem0}
\bibfield{author}{\bibinfo{person}{Prateek Chhikara}, \bibinfo{person}{Dev
  Khant}, \bibinfo{person}{Saket Aryan}, \bibinfo{person}{Taranjeet Singh},
  {and} \bibinfo{person}{Deshraj Yadav}.} \bibinfo{year}{2025}\natexlab{}.
\newblock \showarticletitle{Mem0: Building Production-Ready {AI} Agents with
  Scalable Long-Term Memory}. In \bibinfo{booktitle}{\emph{{ECAI}}}
  \emph{(\bibinfo{series}{Frontiers in Artificial Intelligence and
  Applications})}. \bibinfo{publisher}{{IOS} Press},
  \bibinfo{pages}{2993--3000}.
\newblock


\bibitem[Chollet et~al\mbox{.}(2024)]%
        {privacypreservingpersonalassistant}
\bibfield{author}{\bibinfo{person}{G{\'e}rard Chollet}, \bibinfo{person}{Hugues
  Sansen}, \bibinfo{person}{Yannis Tevissen}, \bibinfo{person}{J{\'e}r{\^o}me
  Boudy}, \bibinfo{person}{Mossaab Hariz}, \bibinfo{person}{Christophe Lohr},
  {and} \bibinfo{person}{Fathy Yassa}.} \bibinfo{year}{2024}\natexlab{}.
\newblock \showarticletitle{Privacy preserving personal assistant with
  on-device diarization and spoken dialogue system for home and beyond}.
\newblock \bibinfo{journal}{\emph{arXiv preprint arXiv:2401.01146}}
  (\bibinfo{year}{2024}).
\newblock


\bibitem[Chuang et~al\mbox{.}(2026)]%
        {teefail}
\bibfield{author}{\bibinfo{person}{Jalen Chuang}, \bibinfo{person}{Alex Seto},
  \bibinfo{person}{Nicolas Berrios}, \bibinfo{person}{Stephan van Schaik},
  \bibinfo{person}{Christina Garman}, {and} \bibinfo{person}{Daniel Genkin}.}
  \bibinfo{year}{2026}\natexlab{}.
\newblock \bibinfo{title}{{TEE.fail}: Breaking Trusted Execution Environments
  via {DDR5} Memory Bus Interposition}.
\newblock \bibinfo{howpublished}{\url{https://tee.fail/}}.
\newblock


\bibitem[Civan et~al\mbox{.}(2008)]%
        {foldersortags}
\bibfield{author}{\bibinfo{person}{Andrea Civan}, \bibinfo{person}{William
  Jones}, \bibinfo{person}{Predrag~V. Klasnja}, {and} \bibinfo{person}{Harry
  Bruce}.} \bibinfo{year}{2008}\natexlab{}.
\newblock \showarticletitle{Better to organize personal information by folders
  or by tags?: The devil is in the details}. In
  \bibinfo{booktitle}{\emph{{ASIST}}} \emph{(\bibinfo{series}{Proc. Assoc. Inf.
  Sci. Technol.}, \bibinfo{number}{1})}. \bibinfo{publisher}{Wiley},
  \bibinfo{pages}{1--13}.
\newblock


\bibitem[{Confidential Containers}(2024)]%
        {confidentialcontainers}
\bibfield{author}{\bibinfo{person}{{Confidential Containers}}.}
  \bibinfo{year}{2024}\natexlab{}.
\newblock \bibinfo{title}{Trust Model for Confidential Containers}.
\newblock
  \bibinfo{howpublished}{\url{https://confidentialcontainers.org/docs/architecture/trust-model/trust-model/}}.
\newblock


\bibitem[Connell et~al\mbox{.}(2024)]%
        {svr3}
\bibfield{author}{\bibinfo{person}{Graeme Connell}, \bibinfo{person}{Vivian
  Fang}, \bibinfo{person}{Rolfe Schmidt}, \bibinfo{person}{Emma Dauterman},
  {and} \bibinfo{person}{Raluca~Ada Popa}.} \bibinfo{year}{2024}\natexlab{}.
\newblock \showarticletitle{Secret Key Recovery in a Global-Scale End-to-End
  Encryption System}. In \bibinfo{booktitle}{\emph{{OSDI}}}.
  \bibinfo{publisher}{{USENIX} Association}, \bibinfo{pages}{703--719}.
\newblock


\bibitem[Constable et~al\mbox{.}(2023)]%
        {aexnotify}
\bibfield{author}{\bibinfo{person}{Scott Constable}, \bibinfo{person}{Jo~Van
  Bulck}, \bibinfo{person}{Xiang Cheng}, \bibinfo{person}{Yuan Xiao},
  \bibinfo{person}{Cedric Xing}, \bibinfo{person}{Ilya Alexandrovich},
  \bibinfo{person}{Taesoo Kim}, \bibinfo{person}{Frank Piessens},
  \bibinfo{person}{Mona Vij}, {and} \bibinfo{person}{Mark Silberstein}.}
  \bibinfo{year}{2023}\natexlab{}.
\newblock \showarticletitle{AEX-Notify: Thwarting Precise Single-Stepping
  Attacks through Interrupt Awareness for Intel {SGX} Enclaves}. In
  \bibinfo{booktitle}{\emph{{USENIX} Security Symposium}}.
  \bibinfo{publisher}{{USENIX} Association}, \bibinfo{pages}{4051--4068}.
\newblock


\bibitem[Costan et~al\mbox{.}(2016)]%
        {sanctum}
\bibfield{author}{\bibinfo{person}{Victor Costan}, \bibinfo{person}{Ilia~A.
  Lebedev}, {and} \bibinfo{person}{Srinivas Devadas}.}
  \bibinfo{year}{2016}\natexlab{}.
\newblock \showarticletitle{Sanctum: Minimal Hardware Extensions for Strong
  Software Isolation}. In \bibinfo{booktitle}{\emph{{USENIX} Security
  Symposium}}. \bibinfo{publisher}{{USENIX} Association},
  \bibinfo{pages}{857--874}.
\newblock


\bibitem[Cremers et~al\mbox{.}(2023)]%
        {formalanalysisofspdm}
\bibfield{author}{\bibinfo{person}{Cas Cremers}, \bibinfo{person}{Alexander
  Dax}, {and} \bibinfo{person}{Aurora Naska}.} \bibinfo{year}{2023}\natexlab{}.
\newblock \showarticletitle{Formal Analysis of {SPDM:} Security Protocol and
  Data Model version 1.2}. In \bibinfo{booktitle}{\emph{{USENIX} Security
  Symposium}}. \bibinfo{publisher}{{USENIX} Association},
  \bibinfo{pages}{6611--6628}.
\newblock


\bibitem[Croes et~al\mbox{.}(2024)]%
        {digitalconfessions}
\bibfield{author}{\bibinfo{person}{Emmelyn A.~J. Croes},
  \bibinfo{person}{Marjolijn~L. Antheunis}, \bibinfo{person}{Chris van~der
  Lee}, {and} \bibinfo{person}{Jan M.~S. de Wit}.}
  \bibinfo{year}{2024}\natexlab{}.
\newblock \showarticletitle{Digital Confessions: The Willingness to Disclose
  Intimate Information to a Chatbot and its Impact on Emotional Well-Being}.
\newblock \bibinfo{journal}{\emph{Interact. Comput.}} \bibinfo{volume}{36},
  \bibinfo{number}{5} (\bibinfo{year}{2024}), \bibinfo{pages}{279--292}.
\newblock


\bibitem[Dauterman et~al\mbox{.}(2021)]%
        {snoopy}
\bibfield{author}{\bibinfo{person}{Emma Dauterman}, \bibinfo{person}{Vivian
  Fang}, \bibinfo{person}{Ioannis Demertzis}, \bibinfo{person}{Natacha Crooks},
  {and} \bibinfo{person}{Raluca~Ada Popa}.} \bibinfo{year}{2021}\natexlab{}.
\newblock \showarticletitle{Snoopy: Surpassing the Scalability Bottleneck of
  Oblivious Storage}. In \bibinfo{booktitle}{\emph{{SOSP}}}.
  \bibinfo{publisher}{{ACM}}, \bibinfo{pages}{655--671}.
\newblock


\bibitem[de~Castro(2025)]%
        {bulkemails}
\bibfield{author}{\bibinfo{person}{Vicente~Bicudo de Castro}.}
  \bibinfo{year}{2025}\natexlab{}.
\newblock \showarticletitle{Quantifying the burden of organisational bulk
  emails in a business school}.
\newblock \bibinfo{journal}{\emph{Higher Education Research \& Development}}
  \bibinfo{volume}{44}, \bibinfo{number}{8} (\bibinfo{year}{2025}),
  \bibinfo{pages}{1934--1948}.
\newblock
\href{https://doi.org/10.1080/07294360.2025.2510662}{doi:\nolinkurl{10.1080/07294360.2025.2510662}}


\bibitem[Dearman and Pierce(2008)]%
        {multidevicecomputing}
\bibfield{author}{\bibinfo{person}{David Dearman} {and}
  \bibinfo{person}{Jeffrey~S. Pierce}.} \bibinfo{year}{2008}\natexlab{}.
\newblock \showarticletitle{It's on my other computer!: computing with multiple
  devices}. In \bibinfo{booktitle}{\emph{{CHI}}}. \bibinfo{publisher}{{ACM}},
  \bibinfo{pages}{767--776}.
\newblock


\bibitem[{DMTF}(2025)]%
        {dmtfspdm}
\bibfield{author}{\bibinfo{person}{{DMTF}}.} \bibinfo{year}{2025}\natexlab{}.
\newblock \bibinfo{title}{Security Protocol and Data Model ({SPDM})
  Specification}.
\newblock \bibinfo{howpublished}{\url{https://www.dmtf.org/dsp/DSP0274}}.
\newblock


\bibitem[Du et~al\mbox{.}(2024)]%
        {perltqa}
\bibfield{author}{\bibinfo{person}{Yiming Du}, \bibinfo{person}{Hongru Wang},
  \bibinfo{person}{Zhengyi Zhao}, \bibinfo{person}{Bin Liang},
  \bibinfo{person}{Baojun Wang}, \bibinfo{person}{Wanjun Zhong},
  \bibinfo{person}{Zezhong Wang}, {and} \bibinfo{person}{Kam{-}Fai Wong}.}
  \bibinfo{year}{2024}\natexlab{}.
\newblock \showarticletitle{{PerLTQA}: A Personal Long-Term Memory Dataset for
  Memory Classification, Retrieval, and Fusion in Question Answering}. In
  \bibinfo{booktitle}{\emph{{SIGHAN}}}. \bibinfo{publisher}{Association for
  Computational Linguistics}, \bibinfo{pages}{152--164}.
\newblock


\bibitem[Dumais et~al\mbox{.}(2015)]%
        {stuffiveseen}
\bibfield{author}{\bibinfo{person}{Susan~T. Dumais}, \bibinfo{person}{Edward
  Cutrell}, \bibinfo{person}{Jonathan~J. Cadiz}, \bibinfo{person}{Gavin
  Jancke}, \bibinfo{person}{Raman Sarin}, {and} \bibinfo{person}{Daniel~C.
  Robbins}.} \bibinfo{year}{2015}\natexlab{}.
\newblock \showarticletitle{Stuff I've Seen: {A} System for Personal
  Information Retrieval and Re-Use}.
\newblock \bibinfo{journal}{\emph{{SIGIR} Forum}} \bibinfo{volume}{49},
  \bibinfo{number}{2} (\bibinfo{year}{2015}), \bibinfo{pages}{28--35}.
\newblock
\href{https://doi.org/10.1145/2888422.2888425}{doi:\nolinkurl{10.1145/2888422.2888425}}


\bibitem[Edge et~al\mbox{.}(2024)]%
        {graphrag}
\bibfield{author}{\bibinfo{person}{Darren Edge}, \bibinfo{person}{Ha Trinh},
  \bibinfo{person}{Newman Cheng}, \bibinfo{person}{Joshua Bradley},
  \bibinfo{person}{Alex Chao}, \bibinfo{person}{Apurva Mody},
  \bibinfo{person}{Steven Truitt}, \bibinfo{person}{Dasha Metropolitansky},
  \bibinfo{person}{Robert~Osazuwa Ness}, {and} \bibinfo{person}{Jonathan
  Larson}.} \bibinfo{year}{2024}\natexlab{}.
\newblock \showarticletitle{From Local to Global: A Graph RAG Approach to
  Query-Focused Summarization}.
\newblock \bibinfo{journal}{\emph{arXiv preprint arXiv:2404.16130}}
  (\bibinfo{year}{2024}).
\newblock


\bibitem[Engels et~al\mbox{.}(2024)]%
        {windowfilters}
\bibfield{author}{\bibinfo{person}{Joshua Engels}, \bibinfo{person}{Benjamin
  Landrum}, \bibinfo{person}{Shangdi Yu}, \bibinfo{person}{Laxman Dhulipala},
  {and} \bibinfo{person}{Julian Shun}.} \bibinfo{year}{2024}\natexlab{}.
\newblock \showarticletitle{Approximate Nearest Neighbor Search with Window
  Filters}. In \bibinfo{booktitle}{\emph{{ICML}}}
  \emph{(\bibinfo{series}{Proceedings of Machine Learning Research})}.
  \bibinfo{publisher}{{PMLR} / OpenReview.net}, \bibinfo{pages}{12469--12490}.
\newblock


\bibitem[{Epoch AI}(2025)]%
        {epochinferencecost}
\bibfield{author}{\bibinfo{person}{{Epoch AI}}.}
  \bibinfo{year}{2025}\natexlab{}.
\newblock \bibinfo{title}{{LLM} Inference Price Trends}.
\newblock
  \bibinfo{howpublished}{\url{https://epoch.ai/data-insights/llm-inference-price-trends}}.
\newblock
\newblock
\shownote{Accessed 2026-10-27}.


\bibitem[Eskandarian and Zaharia(2019)]%
        {oblidb}
\bibfield{author}{\bibinfo{person}{Saba Eskandarian} {and}
  \bibinfo{person}{Matei Zaharia}.} \bibinfo{year}{2019}\natexlab{}.
\newblock \showarticletitle{ObliDB: Oblivious Query Processing for Secure
  Databases}.
\newblock \bibinfo{journal}{\emph{Proc. {VLDB} Endow.}} \bibinfo{volume}{13},
  \bibinfo{number}{2} (\bibinfo{year}{2019}), \bibinfo{pages}{169--183}.
\newblock


\bibitem[{Federal Trade Commission}(2024)]%
        {ftcsocialmedia2024}
\bibfield{author}{\bibinfo{person}{{Federal Trade Commission}}.}
  \bibinfo{year}{2024}\natexlab{}.
\newblock \bibinfo{booktitle}{\emph{A Look Behind the Screens: Examining the
  Data Practices of Social Media and Video Streaming Services}}.
\newblock \bibinfo{type}{{T}echnical {R}eport}.
\newblock
\urldef\tempurl%
\url{https://www.ftc.gov/system/files/ftc_gov/pdf/Social-Media-6b-Report-9-11-2024.pdf}
\showURL{%
\tempurl}


\bibitem[Ferraiuolo et~al\mbox{.}(2017)]%
        {komodo}
\bibfield{author}{\bibinfo{person}{Andrew Ferraiuolo}, \bibinfo{person}{Andrew
  Baumann}, \bibinfo{person}{Chris Hawblitzel}, {and} \bibinfo{person}{Bryan
  Parno}.} \bibinfo{year}{2017}\natexlab{}.
\newblock \showarticletitle{Komodo: Using verification to disentangle
  secure-enclave hardware from software}. In
  \bibinfo{booktitle}{\emph{{SOSP}}}. \bibinfo{publisher}{{ACM}},
  \bibinfo{pages}{287--305}.
\newblock


\bibitem[Fox et~al\mbox{.}(2016)]%
        {Fox02042016}
\bibfield{author}{\bibinfo{person}{Eric~W. Fox}, \bibinfo{person}{Martin~B.
  Short}, \bibinfo{person}{Frederic~P. Schoenberg}, \bibinfo{person}{Kathryn~D.
  Coronges}, {and} \bibinfo{person}{Andrea~L. Bertozzi}.}
  \bibinfo{year}{2016}\natexlab{}.
\newblock \showarticletitle{Modeling E-mail Networks and Inferring Leadership
  Using Self-Exciting Point Processes}.
\newblock \bibinfo{journal}{\emph{J. Amer. Statist. Assoc.}}
  \bibinfo{volume}{111}, \bibinfo{number}{514} (\bibinfo{year}{2016}),
  \bibinfo{pages}{564--584}.
\newblock
\href{https://doi.org/10.1080/01621459.2015.1135802}{doi:\nolinkurl{10.1080/01621459.2015.1135802}}


\bibitem[Fuhry et~al\mbox{.}(2017)]%
        {hardidx}
\bibfield{author}{\bibinfo{person}{Benny Fuhry}, \bibinfo{person}{Raad
  Bahmani}, \bibinfo{person}{Ferdinand Brasser}, \bibinfo{person}{Florian
  Hahn}, \bibinfo{person}{Florian Kerschbaum}, {and}
  \bibinfo{person}{Ahmad{-}Reza Sadeghi}.} \bibinfo{year}{2017}\natexlab{}.
\newblock \showarticletitle{HardIDX: Practical and Secure Index with {SGX}}. In
  \bibinfo{booktitle}{\emph{DBSec}} \emph{(\bibinfo{series}{Lecture Notes in
  Computer Science})}. \bibinfo{publisher}{Springer},
  \bibinfo{pages}{386--408}.
\newblock


\bibitem[Gast et~al\mbox{.}(2025)]%
        {counterseveillance}
\bibfield{author}{\bibinfo{person}{Stefan Gast}, \bibinfo{person}{Hannes
  Weissteiner}, \bibinfo{person}{Robin~Leander Schr{\"{o}}der}, {and}
  \bibinfo{person}{Daniel Gruss}.} \bibinfo{year}{2025}\natexlab{}.
\newblock \showarticletitle{CounterSEVeillance: Performance-Counter Attacks on
  {AMD} {SEV-SNP}}. In \bibinfo{booktitle}{\emph{{NDSS}}}.
  \bibinfo{publisher}{The Internet Society}.
\newblock


\bibitem[Ge et~al\mbox{.}(2024)]%
        {personahub}
\bibfield{author}{\bibinfo{person}{Tao Ge}, \bibinfo{person}{Xin Chan},
  \bibinfo{person}{Xiaoyang Wang}, \bibinfo{person}{Dian Yu},
  \bibinfo{person}{Haitao Mi}, {and} \bibinfo{person}{Dong Yu}.}
  \bibinfo{year}{2024}\natexlab{}.
\newblock \showarticletitle{Scaling synthetic data creation with 1,000,000,000
  personas}.
\newblock \bibinfo{journal}{\emph{arXiv preprint arXiv:2406.20094}}
  (\bibinfo{year}{2024}).
\newblock


\bibitem[Ge et~al\mbox{.}(2022)]%
        {hecate}
\bibfield{author}{\bibinfo{person}{Xinyang Ge}, \bibinfo{person}{Hsuan{-}Chi
  Kuo}, {and} \bibinfo{person}{Weidong Cui}.} \bibinfo{year}{2022}\natexlab{}.
\newblock \showarticletitle{Hecate: Lifting and Shifting On-Premises Workloads
  to an Untrusted Cloud}. In \bibinfo{booktitle}{\emph{{CCS}}}.
  \bibinfo{publisher}{{ACM}}, \bibinfo{pages}{1231--1242}.
\newblock


\bibitem[Gemmell et~al\mbox{.}(2006)]%
        {mylifebits}
\bibfield{author}{\bibinfo{person}{Jim Gemmell}, \bibinfo{person}{Gordon Bell},
  {and} \bibinfo{person}{Roger Lueder}.} \bibinfo{year}{2006}\natexlab{}.
\newblock \showarticletitle{MyLifeBits: a personal database for everything}.
\newblock \bibinfo{journal}{\emph{Commun. {ACM}}} \bibinfo{volume}{49},
  \bibinfo{number}{1} (\bibinfo{year}{2006}), \bibinfo{pages}{88--95}.
\newblock
\href{https://doi.org/10.1145/1107458.1107460}{doi:\nolinkurl{10.1145/1107458.1107460}}


\bibitem[Gervais et~al\mbox{.}(2014)]%
        {quantifyingwebsearchprivacy}
\bibfield{author}{\bibinfo{person}{Arthur Gervais}, \bibinfo{person}{Reza
  Shokri}, \bibinfo{person}{Adish Singla}, \bibinfo{person}{Srdjan Capkun},
  {and} \bibinfo{person}{Vincent Lenders}.} \bibinfo{year}{2014}\natexlab{}.
\newblock \showarticletitle{Quantifying Web-Search Privacy}. In
  \bibinfo{booktitle}{\emph{{CCS}}}. \bibinfo{publisher}{{ACM}},
  \bibinfo{pages}{966--977}.
\newblock


\bibitem[Gliklich et~al\mbox{.}(2016)]%
        {drivetext}
\bibfield{author}{\bibinfo{person}{Emily Gliklich}, \bibinfo{person}{Rong Guo},
  {and} \bibinfo{person}{Regan~W. Bergmark}.} \bibinfo{year}{2016}\natexlab{}.
\newblock \showarticletitle{Texting while driving: A study of 1211 U.S. adults
  with the Distracted Driving Survey}.
\newblock \bibinfo{journal}{\emph{Preventive Medicine Reports}}
  \bibinfo{volume}{4} (\bibinfo{year}{2016}), \bibinfo{pages}{486--489}.
\newblock
\href{https://doi.org/10.1016/j.pmedr.2016.09.003}{doi:\nolinkurl{10.1016/j.pmedr.2016.09.003}}


\bibitem[Gollapudi et~al\mbox{.}(2023)]%
        {filtereddiskann}
\bibfield{author}{\bibinfo{person}{Siddharth Gollapudi}, \bibinfo{person}{Neel
  Karia}, \bibinfo{person}{Varun Sivashankar}, \bibinfo{person}{Ravishankar
  Krishnaswamy}, \bibinfo{person}{Nikit Begwani}, \bibinfo{person}{Swapnil
  Raz}, \bibinfo{person}{Yiyong Lin}, \bibinfo{person}{Yin Zhang},
  \bibinfo{person}{Neelam Mahapatro}, \bibinfo{person}{Premkumar Srinivasan},
  \bibinfo{person}{Amit Singh}, {and} \bibinfo{person}{Harsha~Vardhan
  Simhadri}.} \bibinfo{year}{2023}\natexlab{}.
\newblock \showarticletitle{Filtered-DiskANN: Graph Algorithms for Approximate
  Nearest Neighbor Search with Filters}. In \bibinfo{booktitle}{\emph{{WWW}}}.
  \bibinfo{publisher}{{ACM}}, \bibinfo{pages}{3406--3416}.
\newblock


\bibitem[Gomaa et~al\mbox{.}(2025)]%
        {converse}
\bibfield{author}{\bibinfo{person}{Amr Gomaa}, \bibinfo{person}{Ahmed Salem},
  {and} \bibinfo{person}{Sahar Abdelnabi}.} \bibinfo{year}{2025}\natexlab{}.
\newblock \showarticletitle{ConVerse: Benchmarking Contextual Safety in
  Agent-to-Agent Conversations}.
\newblock \bibinfo{journal}{\emph{arXiv preprint arXiv:2511.05359}}
  (\bibinfo{year}{2025}).
\newblock


\bibitem[{Google}(2026)]%
        {notebooklm}
\bibfield{author}{\bibinfo{person}{{Google}}.} \bibinfo{year}{2026}\natexlab{}.
\newblock \bibinfo{title}{NotebookLM}.
\newblock \bibinfo{howpublished}{\url{https://notebooklm.google/}}.
\newblock


\bibitem[{Google Cloud}(2026a)]%
        {gcpcloudstoragepricing}
\bibfield{author}{\bibinfo{person}{{Google Cloud}}.}
  \bibinfo{year}{2026}\natexlab{a}.
\newblock \bibinfo{title}{Google Cloud Storage Pricing}.
\newblock
  \bibinfo{howpublished}{\url{https://cloud.google.com/storage/pricing}}.
\newblock
\newblock
\shownote{Accessed 2026-10-27: Nearline storage, \$0.010 per GiB-month}.


\bibitem[{Google Cloud}(2026b)]%
        {googlepersistentdisk}
\bibfield{author}{\bibinfo{person}{{Google Cloud}}.}
  \bibinfo{year}{2026}\natexlab{b}.
\newblock \bibinfo{title}{Persistent Disk}.
\newblock
  \bibinfo{howpublished}{\url{https://cloud.google.com/compute/docs/disks/persistent-disks}}.
\newblock


\bibitem[{Google DeepMind}(2026)]%
        {googleastra}
\bibfield{author}{\bibinfo{person}{{Google DeepMind}}.}
  \bibinfo{year}{2026}\natexlab{}.
\newblock \bibinfo{title}{Google Project Astra}.
\newblock
  \bibinfo{howpublished}{\url{https://deepmind.google/models/project-astra/}}.
\newblock


\bibitem[{Google Private AI Compute Team}(2026)]%
        {googleprivatememory2026}
\bibfield{author}{\bibinfo{person}{{Google Private AI Compute Team}}.}
  \bibinfo{year}{2026}\natexlab{}.
\newblock \bibinfo{title}{Advancing {Private AI Compute} with secure,
  server-side memory}.
\newblock
  \bibinfo{howpublished}{\url{https://deepmind.google/blog/advancing-private-ai-compute-with-secure-server-side-memory/}}.
\newblock


\bibitem[Guo et~al\mbox{.}(2025)]%
        {lightrag}
\bibfield{author}{\bibinfo{person}{Zirui Guo}, \bibinfo{person}{Lianghao Xia},
  \bibinfo{person}{Yanhua Yu}, \bibinfo{person}{Tu Ao}, {and}
  \bibinfo{person}{Chao Huang}.} \bibinfo{year}{2025}\natexlab{}.
\newblock \showarticletitle{LightRAG: Simple and Fast Retrieval-Augmented
  Generation}. In \bibinfo{booktitle}{\emph{{EMNLP} (Findings)}}.
  \bibinfo{publisher}{Association for Computational Linguistics},
  \bibinfo{pages}{10746--10761}.
\newblock


\bibitem[Gupta et~al\mbox{.}(2023)]%
        {caps}
\bibfield{author}{\bibinfo{person}{Gaurav Gupta}, \bibinfo{person}{Jonah Yi},
  \bibinfo{person}{Benjamin Coleman}, \bibinfo{person}{Chen Luo},
  \bibinfo{person}{Vihan Lakshman}, {and} \bibinfo{person}{Anshumali
  Shrivastava}.} \bibinfo{year}{2023}\natexlab{}.
\newblock \showarticletitle{Caps: A practical partition index for filtered
  similarity search}.
\newblock \bibinfo{journal}{\emph{arXiv preprint arXiv:2308.15014}}
  (\bibinfo{year}{2023}).
\newblock


\bibitem[Gutierrez et~al\mbox{.}(2024)]%
        {hipporag}
\bibfield{author}{\bibinfo{person}{Bernal~Jimenez Gutierrez},
  \bibinfo{person}{Yiheng Shu}, \bibinfo{person}{Yu Gu},
  \bibinfo{person}{Michihiro Yasunaga}, {and} \bibinfo{person}{Yu Su}.}
  \bibinfo{year}{2024}\natexlab{}.
\newblock \showarticletitle{HippoRAG: Neurobiologically Inspired Long-Term
  Memory for Large Language Models}. In \bibinfo{booktitle}{\emph{NeurIPS}}.
\newblock


\bibitem[Hawkes(1971)]%
        {hawkes1971}
\bibfield{author}{\bibinfo{person}{Alan~G. Hawkes}.}
  \bibinfo{year}{1971}\natexlab{}.
\newblock \showarticletitle{Point Spectra of Some Mutually Exciting Point
  Processes}.
\newblock \bibinfo{journal}{\emph{Journal of the Royal Statistical Society
  Series B: Statistical Methodology}} \bibinfo{volume}{33}, \bibinfo{number}{3}
  (\bibinfo{year}{1971}), \bibinfo{pages}{438--443}.
\newblock
\href{https://doi.org/10.1111/j.2517-6161.1971.tb01530.x}{doi:\nolinkurl{10.1111/j.2517-6161.1971.tb01530.x}}


\bibitem[Hawkes and Oakes(1974)]%
        {hawkes1974}
\bibfield{author}{\bibinfo{person}{Alan~G. Hawkes} {and} \bibinfo{person}{David
  Oakes}.} \bibinfo{year}{1974}\natexlab{}.
\newblock \showarticletitle{A cluster process representation of a self-exciting
  process}.
\newblock \bibinfo{journal}{\emph{Journal of Applied Probability}}
  \bibinfo{volume}{11}, \bibinfo{number}{3} (\bibinfo{year}{1974}),
  \bibinfo{pages}{493--503}.
\newblock
\href{https://doi.org/10.2307/3212693}{doi:\nolinkurl{10.2307/3212693}}


\bibitem[Henzinger et~al\mbox{.}(2023)]%
        {tiptoe}
\bibfield{author}{\bibinfo{person}{Alexandra Henzinger}, \bibinfo{person}{Emma
  Dauterman}, \bibinfo{person}{Henry Corrigan{-}Gibbs}, {and}
  \bibinfo{person}{Nickolai Zeldovich}.} \bibinfo{year}{2023}\natexlab{}.
\newblock \showarticletitle{Private Web Search with Tiptoe}. In
  \bibinfo{booktitle}{\emph{{SOSP}}}. \bibinfo{publisher}{{ACM}},
  \bibinfo{pages}{396--416}.
\newblock


\bibitem[Hu et~al\mbox{.}(2025)]%
        {memoryagentbench}
\bibfield{author}{\bibinfo{person}{Yuanzhe Hu}, \bibinfo{person}{Yu Wang},
  {and} \bibinfo{person}{Julian McAuley}.} \bibinfo{year}{2025}\natexlab{}.
\newblock \showarticletitle{Evaluating Memory in LLM Agents via Incremental
  Multi-Turn Interactions}.
\newblock \bibinfo{journal}{\emph{arXiv preprint arXiv:2507.05257}}
  (\bibinfo{year}{2025}).
\newblock


\bibitem[Igarashi et~al\mbox{.}(2022)]%
        {igarashi2022}
\bibfield{author}{\bibinfo{person}{Naoki Igarashi}, \bibinfo{person}{Yukihiko
  Okada}, \bibinfo{person}{Hiroki Sayama}, {and} \bibinfo{person}{Yukie Sano}.}
  \bibinfo{year}{2022}\natexlab{}.
\newblock \showarticletitle{A two-phase model of collective memory decay with a
  dynamical switching point}.
\newblock \bibinfo{journal}{\emph{Scientific Reports}} \bibinfo{volume}{12},
  \bibinfo{number}{1} (\bibinfo{year}{2022}).
\newblock
\href{https://doi.org/10.1038/s41598-022-25840-9}{doi:\nolinkurl{10.1038/s41598-022-25840-9}}


\bibitem[{Intel}(2022)]%
        {inteltdx}
\bibfield{author}{\bibinfo{person}{{Intel}}.} \bibinfo{year}{2022}\natexlab{}.
\newblock \bibinfo{title}{Intel Trust Domain Extensions ({TDX})}.
\newblock
  \bibinfo{howpublished}{\url{https://www.intel.com/content/dam/develop/external/us/en/documents/tdx-whitepaper-final9-17.pdf}}.
\newblock


\bibitem[{Intel}(2026)]%
        {inteltdxguestsecurity}
\bibfield{author}{\bibinfo{person}{{Intel}}.} \bibinfo{year}{2026}\natexlab{}.
\newblock \bibinfo{title}{Intel Trust Domain Extension Guest Linux Kernel
  Hardening Strategy}.
\newblock
  \bibinfo{howpublished}{\url{https://intel.github.io/ccc-linux-guest-hardening-docs/tdx-guest-hardening.html}}.
\newblock


\bibitem[{Intel Corporation}(2025)]%
        {inteltdxconnectteeio}
\bibfield{author}{\bibinfo{person}{{Intel Corporation}}.}
  \bibinfo{year}{2025}\natexlab{}.
\newblock \bibinfo{title}{Intel {TDX} Connect {TEE-IO} Device Guide}.
\newblock
  \bibinfo{howpublished}{\url{https://www.intel.com/content/www/us/en/content-details/861654/intel-tdx-connect-tee-io-device-guide.html}}.
\newblock


\bibitem[Jia et~al\mbox{.}(2025)]%
        {foundintranslation}
\bibfield{author}{\bibinfo{person}{Grace Jia}, \bibinfo{person}{Alex Wong},
  {and} \bibinfo{person}{Anurag Khandelwal}.} \bibinfo{year}{2025}\natexlab{}.
\newblock \showarticletitle{Found in Translation: {A} Generative Language
  Modeling Approach to Memory Access Pattern Attacks}. In
  \bibinfo{booktitle}{\emph{{USENIX} Security Symposium}}.
  \bibinfo{publisher}{{USENIX} Association}, \bibinfo{pages}{7957--7975}.
\newblock


\bibitem[Jiang et~al\mbox{.}(2026)]%
        {magma}
\bibfield{author}{\bibinfo{person}{Dongming Jiang}, \bibinfo{person}{Yi Li},
  \bibinfo{person}{Guanpeng Li}, {and} \bibinfo{person}{Bingzhe Li}.}
  \bibinfo{year}{2026}\natexlab{}.
\newblock \showarticletitle{MAGMA: A Multi-Graph based Agentic Memory
  Architecture for AI Agents}.
\newblock \bibinfo{journal}{\emph{arXiv preprint arXiv:2601.03236}}
  (\bibinfo{year}{2026}).
\newblock


\bibitem[Jiayang et~al\mbox{.}(2026)]%
        {amemgym}
\bibfield{author}{\bibinfo{person}{Cheng Jiayang}, \bibinfo{person}{Dongyu Ru},
  \bibinfo{person}{Lin Qiu}, \bibinfo{person}{Yiyang Li},
  \bibinfo{person}{Xuezhi Cao}, \bibinfo{person}{Yangqiu Song}, {and}
  \bibinfo{person}{Xunliang Cai}.} \bibinfo{year}{2026}\natexlab{}.
\newblock \showarticletitle{AMemGym: Interactive Memory Benchmarking for
  Assistants in Long-Horizon Conversations}.
\newblock \bibinfo{journal}{\emph{arXiv preprint arXiv:2603.01966}}
  (\bibinfo{year}{2026}).
\newblock


\bibitem[Jin and Wu(2025)]%
        {cecollm}
\bibfield{author}{\bibinfo{person}{Hongpeng Jin} {and} \bibinfo{person}{Yanzhao
  Wu}.} \bibinfo{year}{2025}\natexlab{}.
\newblock \showarticletitle{CE-CoLLM: Efficient and Adaptive Large Language
  Models Through Cloud-Edge Collaboration}. In
  \bibinfo{booktitle}{\emph{{ICWS}}}. \bibinfo{publisher}{{IEEE}},
  \bibinfo{pages}{316--323}.
\newblock


\bibitem[Jo et~al\mbox{.}(2024)]%
        {chatbotselfdisclosure}
\bibfield{author}{\bibinfo{person}{Eunkyung Jo}, \bibinfo{person}{Yuin Jeong},
  \bibinfo{person}{SoHyun Park}, \bibinfo{person}{Daniel~A. Epstein}, {and}
  \bibinfo{person}{Young{-}Ho Kim}.} \bibinfo{year}{2024}\natexlab{}.
\newblock \showarticletitle{Understanding the Impact of Long-Term Memory on
  Self-Disclosure with Large Language Model-Driven Chatbots for Public Health
  Intervention}. In \bibinfo{booktitle}{\emph{{CHI}}}.
  \bibinfo{publisher}{{ACM}}, \bibinfo{pages}{440:1--440:21}.
\newblock


\bibitem[Johnson et~al\mbox{.}(2021)]%
        {faiss}
\bibfield{author}{\bibinfo{person}{Jeff Johnson}, \bibinfo{person}{Matthijs
  Douze}, {and} \bibinfo{person}{Herv{\'{e}} J{\'{e}}gou}.}
  \bibinfo{year}{2021}\natexlab{}.
\newblock \showarticletitle{Billion-Scale Similarity Search with GPUs}.
\newblock \bibinfo{journal}{\emph{{IEEE} Trans. Big Data}} \bibinfo{volume}{7},
  \bibinfo{number}{3} (\bibinfo{year}{2021}), \bibinfo{pages}{535--547}.
\newblock
\href{https://doi.org/10.1109/TBDATA.2019.2921572}{doi:\nolinkurl{10.1109/TBDATA.2019.2921572}}


\bibitem[Jones et~al\mbox{.}(2002)]%
        {oncefound}
\bibfield{author}{\bibinfo{person}{William Jones}, \bibinfo{person}{Susan~T.
  Dumais}, {and} \bibinfo{person}{Harry Bruce}.}
  \bibinfo{year}{2002}\natexlab{}.
\newblock \showarticletitle{Once found, what then? {A} study of "keeping"
  behaviors in the personal use of Web information}. In
  \bibinfo{booktitle}{\emph{{ASIST}}} \emph{(\bibinfo{series}{Proc. Assoc. Inf.
  Sci. Technol.}, \bibinfo{number}{1})}. \bibinfo{publisher}{Wiley},
  \bibinfo{pages}{391--402}.
\newblock


\bibitem[Jones et~al\mbox{.}(2014)]%
        {refindfilesemailsweb}
\bibfield{author}{\bibinfo{person}{William Jones}, \bibinfo{person}{Abe
  Wenning}, {and} \bibinfo{person}{Harry Bruce}.}
  \bibinfo{year}{2014}\natexlab{}.
\newblock \showarticletitle{How Do People Re-find Files, Emails and Web Pages?}
\newblock \bibinfo{journal}{\emph{iConference 2014 Proceedings}}
  (\bibinfo{year}{2014}).
\newblock
\href{https://doi.org/10.9776/14136}{doi:\nolinkurl{10.9776/14136}}


\bibitem[Karpukhin et~al\mbox{.}(2020)]%
        {dpr}
\bibfield{author}{\bibinfo{person}{Vladimir Karpukhin}, \bibinfo{person}{Barlas
  Oguz}, \bibinfo{person}{Sewon Min}, \bibinfo{person}{Patrick Lewis},
  \bibinfo{person}{Ledell Wu}, \bibinfo{person}{Sergey Edunov},
  \bibinfo{person}{Danqi Chen}, {and} \bibinfo{person}{Wen{-}tau Yih}.}
  \bibinfo{year}{2020}\natexlab{}.
\newblock \showarticletitle{Dense Passage Retrieval for Open-Domain Question
  Answering}. In \bibinfo{booktitle}{\emph{{EMNLP} {(1)}}}.
  \bibinfo{publisher}{Association for Computational Linguistics},
  \bibinfo{pages}{6769--6781}.
\newblock


\bibitem[Kaviani et~al\mbox{.}(2025)]%
        {myco}
\bibfield{author}{\bibinfo{person}{Darya Kaviani}, \bibinfo{person}{Deevashwer
  Rathee}, \bibinfo{person}{Bhargav Annem}, {and} \bibinfo{person}{Raluca~Ada
  Popa}.} \bibinfo{year}{2025}\natexlab{}.
\newblock \showarticletitle{Myco: Unlocking Polylogarithmic Accesses in
  Metadata-Private Messaging}. In \bibinfo{booktitle}{\emph{{SP}}}.
  \bibinfo{publisher}{{IEEE}}, \bibinfo{pages}{4419--4437}.
\newblock


\bibitem[Kelsey et~al\mbox{.}(2019)]%
        {nistir8213}
\bibfield{author}{\bibinfo{person}{John Kelsey}, \bibinfo{person}{Lu{\'i}s T.
  A.~N. Brand{\~a}o}, \bibinfo{person}{Rene Peralta}, {and}
  \bibinfo{person}{Harold Booth}.} \bibinfo{year}{2019}\natexlab{}.
\newblock \bibinfo{booktitle}{\emph{A Reference for Randomness Beacons: Format
  and Protocol Version 2}}.
\newblock \bibinfo{type}{{T}echnical {R}eport} NIST Interagency or Internal
  Report 8213 (Draft). \bibinfo{institution}{National Institute of Standards
  and Technology}.
\newblock
\href{https://doi.org/10.6028/NIST.IR.8213-draft}{doi:\nolinkurl{10.6028/NIST.IR.8213-draft}}


\bibitem[Kim et~al\mbox{.}(2026)]%
        {privasis}
\bibfield{author}{\bibinfo{person}{Hyunwoo Kim}, \bibinfo{person}{Niloofar
  Mireshghallah}, \bibinfo{person}{Michael Duan}, \bibinfo{person}{Rui Xin},
  \bibinfo{person}{Shuyue~Stella Li}, \bibinfo{person}{Jaehun Jung},
  \bibinfo{person}{David Acuna}, \bibinfo{person}{Qi Pang},
  \bibinfo{person}{Hanshen Xiao}, \bibinfo{person}{G.~Edward Suh},
  \bibinfo{person}{Sewoong Oh}, \bibinfo{person}{Yulia Tsvetkov},
  \bibinfo{person}{Pang~Wei Koh}, {and} \bibinfo{person}{Yejin Choi}.}
  \bibinfo{year}{2026}\natexlab{}.
\newblock \showarticletitle{Privasis: Synthesizing the Largest "Public" Private
  Dataset from Scratch}.
\newblock \bibinfo{journal}{\emph{arXiv preprint arXiv:2602.03183}}
  (\bibinfo{year}{2026}).
\newblock


\bibitem[Kooti et~al\mbox{.}(2015)]%
        {ageofemail}
\bibfield{author}{\bibinfo{person}{Farshad Kooti}, \bibinfo{person}{Luca~Maria
  Aiello}, \bibinfo{person}{Mihajlo Grbovic}, \bibinfo{person}{Kristina
  Lerman}, {and} \bibinfo{person}{Amin Mantrach}.}
  \bibinfo{year}{2015}\natexlab{}.
\newblock \showarticletitle{Evolution of Conversations in the Age of Email
  Overload}. In \bibinfo{booktitle}{\emph{{WWW}}}. \bibinfo{publisher}{{ACM}},
  \bibinfo{pages}{603--613}.
\newblock


\bibitem[Kuvaiskii et~al\mbox{.}(2024)]%
        {graminetdx}
\bibfield{author}{\bibinfo{person}{Dmitrii Kuvaiskii},
  \bibinfo{person}{Dimitrios Stavrakakis}, \bibinfo{person}{Kailun Qin},
  \bibinfo{person}{Cedric Xing}, \bibinfo{person}{Pramod Bhatotia}, {and}
  \bibinfo{person}{Mona Vij}.} \bibinfo{year}{2024}\natexlab{}.
\newblock \showarticletitle{Gramine-TDX: {A} Lightweight {OS} Kernel for
  Confidential VMs}. In \bibinfo{booktitle}{\emph{{CCS}}}.
  \bibinfo{publisher}{{ACM}}, \bibinfo{pages}{4598--4612}.
\newblock


\bibitem[Lee et~al\mbox{.}(2020a)]%
        {offchipattackonhardwareenclaves}
\bibfield{author}{\bibinfo{person}{Dayeol Lee}, \bibinfo{person}{Dongha Jung},
  \bibinfo{person}{Ian~T. Fang}, \bibinfo{person}{Chia{-}Che Tsai}, {and}
  \bibinfo{person}{Raluca~Ada Popa}.} \bibinfo{year}{2020}\natexlab{a}.
\newblock \showarticletitle{An Off-Chip Attack on Hardware Enclaves via the
  Memory Bus}. In \bibinfo{booktitle}{\emph{{USENIX} Security Symposium}}.
  \bibinfo{publisher}{{USENIX} Association}, \bibinfo{pages}{487--504}.
\newblock


\bibitem[Lee et~al\mbox{.}(2020b)]%
        {keystone}
\bibfield{author}{\bibinfo{person}{Dayeol Lee}, \bibinfo{person}{David
  Kohlbrenner}, \bibinfo{person}{Shweta Shinde}, \bibinfo{person}{Krste
  Asanovic}, {and} \bibinfo{person}{Dawn Song}.}
  \bibinfo{year}{2020}\natexlab{b}.
\newblock \showarticletitle{Keystone: an open framework for architecting
  trusted execution environments}. In \bibinfo{booktitle}{\emph{EuroSys}}.
  \bibinfo{publisher}{{ACM}}, \bibinfo{pages}{38:1--38:16}.
\newblock


\bibitem[Lee(2025)]%
        {lee2025slack}
\bibfield{author}{\bibinfo{person}{Robert~A. Lee}.}
  \bibinfo{year}{2025}\natexlab{}.
\newblock \bibinfo{title}{Slack Statistics 2026: Daily Active Users, Enterprise
  Trends \& Market Share Analysis}.
\newblock
  \bibinfo{howpublished}{\url{https://sqmagazine.co.uk/slack-statistics}}.
\newblock
\newblock
\shownote{SQ Magazine.}.


\bibitem[Lee et~al\mbox{.}(2017)]%
        {branchshadowing}
\bibfield{author}{\bibinfo{person}{Sangho Lee}, \bibinfo{person}{Ming{-}Wei
  Shih}, \bibinfo{person}{Prasun Gera}, \bibinfo{person}{Taesoo Kim},
  \bibinfo{person}{Hyesoon Kim}, {and} \bibinfo{person}{Marcus Peinado}.}
  \bibinfo{year}{2017}\natexlab{}.
\newblock \showarticletitle{Inferring Fine-grained Control Flow Inside {SGX}
  Enclaves with Branch Shadowing}. In \bibinfo{booktitle}{\emph{{USENIX}
  Security Symposium}}. \bibinfo{publisher}{{USENIX} Association},
  \bibinfo{pages}{557--574}.
\newblock


\bibitem[Lee et~al\mbox{.}(2025)]%
        {ttime}
\bibfield{author}{\bibinfo{person}{Woomin Lee}, \bibinfo{person}{Taehun Kim},
  \bibinfo{person}{Seunghee Shin}, \bibinfo{person}{Junbeom Hur}, {and}
  \bibinfo{person}{Youngjoo Shin}.} \bibinfo{year}{2025}\natexlab{}.
\newblock \showarticletitle{T-Time: {A} Fine-Grained Timing-Based
  Controlled-Channel Attack Against Intel {TDX}}. In
  \bibinfo{booktitle}{\emph{{ESORICS} {(3)}}} \emph{(\bibinfo{series}{Lecture
  Notes in Computer Science})}. \bibinfo{publisher}{Springer},
  \bibinfo{pages}{323--341}.
\newblock


\bibitem[Lefeuvre et~al\mbox{.}(2023)]%
        {reallysafefastconfidentialio}
\bibfield{author}{\bibinfo{person}{Hugo Lefeuvre}, \bibinfo{person}{David
  Chisnall}, \bibinfo{person}{Marios Kogias}, {and} \bibinfo{person}{Pierre
  Olivier}.} \bibinfo{year}{2023}\natexlab{}.
\newblock \showarticletitle{Towards (Really) Safe and Fast Confidential {I/O}}.
  In \bibinfo{booktitle}{\emph{HotOS}}. \bibinfo{publisher}{{ACM}},
  \bibinfo{pages}{214--222}.
\newblock


\bibitem[Li et~al\mbox{.}(2023)]%
        {sentenceembeddingleaksmoreinformationthanyouexpect}
\bibfield{author}{\bibinfo{person}{Haoran Li}, \bibinfo{person}{Mingshi Xu},
  {and} \bibinfo{person}{Yangqiu Song}.} \bibinfo{year}{2023}\natexlab{}.
\newblock \showarticletitle{Sentence Embedding Leaks More Information than You
  Expect: Generative Embedding Inversion Attack to Recover the Whole Sentence}.
  In \bibinfo{booktitle}{\emph{{ACL} (Findings)}}
  \emph{(\bibinfo{series}{Findings of {ACL}})}. \bibinfo{publisher}{Association
  for Computational Linguistics}, \bibinfo{pages}{14022--14040}.
\newblock


\bibitem[Li et~al\mbox{.}(2025d)]%
        {helloagain}
\bibfield{author}{\bibinfo{person}{Hao Li}, \bibinfo{person}{Chenghao Yang},
  \bibinfo{person}{An Zhang}, \bibinfo{person}{Yang Deng},
  \bibinfo{person}{Xiang Wang}, {and} \bibinfo{person}{Tat{-}Seng Chua}.}
  \bibinfo{year}{2025}\natexlab{d}.
\newblock \showarticletitle{Hello Again! LLM-powered Personalized Agent for
  Long-term Dialogue}. In \bibinfo{booktitle}{\emph{{NAACL} (Long Papers)}}.
  \bibinfo{publisher}{Association for Computational Linguistics},
  \bibinfo{pages}{5259--5276}.
\newblock


\bibitem[Li et~al\mbox{.}(2025a)]%
        {panther}
\bibfield{author}{\bibinfo{person}{Jingyu Li}, \bibinfo{person}{Zhicong Huang},
  \bibinfo{person}{Min Zhang}, \bibinfo{person}{Cheng Hong},
  \bibinfo{person}{Jian Liu}, \bibinfo{person}{Tao Wei}, {and}
  \bibinfo{person}{Wenguang Chen}.} \bibinfo{year}{2025}\natexlab{a}.
\newblock \showarticletitle{Panther: Private Approximate Nearest Neighbor
  Search in the Single Server Setting}. In \bibinfo{booktitle}{\emph{{CCS}}}.
  \bibinfo{publisher}{{ACM}}, \bibinfo{pages}{365--379}.
\newblock


\bibitem[Li et~al\mbox{.}(2024b)]%
        {omniquery}
\bibfield{author}{\bibinfo{person}{Jiahao~Nick Li},
  \bibinfo{person}{Zhuohao~Jerry Zhang}, {and} \bibinfo{person}{Jiaju Ma}.}
  \bibinfo{year}{2024}\natexlab{b}.
\newblock \showarticletitle{OmniQuery: Contextually Augmenting Captured
  Multimodal Memory to Enable Personal Question Answering}.
\newblock \bibinfo{journal}{\emph{arXiv preprint arXiv:2409.08250}}
  (\bibinfo{year}{2024}).
\newblock


\bibitem[Li et~al\mbox{.}(2024a)]%
        {teepitfalls}
\bibfield{author}{\bibinfo{person}{Mengyuan Li}, \bibinfo{person}{Yuheng Yang},
  \bibinfo{person}{Guoxing Chen}, \bibinfo{person}{Mengjia Yan}, {and}
  \bibinfo{person}{Yinqian Zhang}.} \bibinfo{year}{2024}\natexlab{a}.
\newblock \showarticletitle{SoK: Understanding Design Choices and Pitfalls of
  Trusted Execution Environments}. In \bibinfo{booktitle}{\emph{AsiaCCS}}.
  \bibinfo{publisher}{{ACM}}.
\newblock


\bibitem[Li et~al\mbox{.}(2021)]%
        {cipherleaks}
\bibfield{author}{\bibinfo{person}{Mengyuan Li}, \bibinfo{person}{Yinqian
  Zhang}, \bibinfo{person}{Huibo Wang}, \bibinfo{person}{Kang Li}, {and}
  \bibinfo{person}{Yueqiang Cheng}.} \bibinfo{year}{2021}\natexlab{}.
\newblock \showarticletitle{{CIPHERLEAKS:} Breaking Constant-time Cryptography
  on {AMD} {SEV} via the Ciphertext Side Channel}. In
  \bibinfo{booktitle}{\emph{{USENIX} Security Symposium}}.
  \bibinfo{publisher}{{USENIX} Association}, \bibinfo{pages}{717--732}.
\newblock


\bibitem[Li et~al\mbox{.}(2025b)]%
        {plosone}
\bibfield{author}{\bibinfo{person}{Wenbo Li}, \bibinfo{person}{David~S. Lee},
  \bibinfo{person}{Jonathan~L. Stahl}, {and} \bibinfo{person}{Joseph Bayer}.}
  \bibinfo{year}{2025}\natexlab{b}.
\newblock \showarticletitle{Reflecting on Dunbar’s numbers: Individual
  differences in energy allocation to personal relationships}.
\newblock \bibinfo{journal}{\emph{PLOS ONE}} \bibinfo{volume}{20},
  \bibinfo{number}{3} (\bibinfo{year}{2025}), \bibinfo{pages}{e0319604}.
\newblock
\href{https://doi.org/10.1371/journal.pone.0319604}{doi:\nolinkurl{10.1371/journal.pone.0319604}}


\bibitem[Li et~al\mbox{.}(2025c)]%
        {memos}
\bibfield{author}{\bibinfo{person}{Zhiyu Li}, \bibinfo{person}{Shichao Song},
  \bibinfo{person}{Hanyu Wang}, \bibinfo{person}{Simin Niu},
  \bibinfo{person}{Ding Chen}, \bibinfo{person}{Jiawei Yang},
  \bibinfo{person}{Chenyang Xi}, \bibinfo{person}{Huayi Lai},
  \bibinfo{person}{Jihao Zhao}, \bibinfo{person}{Yezhaohui Wang},
  \bibinfo{person}{Junpeng Ren}, \bibinfo{person}{Zehao Lin},
  \bibinfo{person}{Jiahao Huo}, \bibinfo{person}{Tianyi Chen},
  \bibinfo{person}{Kai Chen}, \bibinfo{person}{Kehang Li},
  \bibinfo{person}{Zhiqiang Yin}, \bibinfo{person}{Qingchen Yu},
  \bibinfo{person}{Bo Tang}, \bibinfo{person}{Hongkang Yang},
  \bibinfo{person}{Zhi-Qin~John Xu}, {and} \bibinfo{person}{Feiyu Xiong}.}
  \bibinfo{year}{2025}\natexlab{c}.
\newblock \showarticletitle{MemOS: An Operating System for Memory-Augmented
  Generation (MAG) in Large Language Models}.
\newblock \bibinfo{journal}{\emph{arXiv preprint arXiv:2505.22101}}
  (\bibinfo{year}{2025}).
\newblock


\bibitem[Limitless.ai(2026)]%
        {limitlessai}
\bibfield{author}{\bibinfo{person}{Limitless.ai}.}
  \bibinfo{year}{2026}\natexlab{}.
\newblock \bibinfo{title}{Limitless.ai}.
\newblock \bibinfo{howpublished}{\url{https://www.limitless.ai/}}.
\newblock


\bibitem[Lin et~al\mbox{.}(2025a)]%
        {gpuhammer}
\bibfield{author}{\bibinfo{person}{Chris~S. Lin}, \bibinfo{person}{Joyce Qu},
  {and} \bibinfo{person}{Gururaj Saileshwar}.}
  \bibinfo{year}{2025}\natexlab{a}.
\newblock \showarticletitle{GPUHammer: Rowhammer Attacks on {GPU} Memories are
  Practical}. In \bibinfo{booktitle}{\emph{{USENIX} Security Symposium}}.
  \bibinfo{publisher}{{USENIX} Association}, \bibinfo{pages}{5719--5738}.
\newblock


\bibitem[Lin et~al\mbox{.}(2025b)]%
        {surveyfann}
\bibfield{author}{\bibinfo{person}{Yanjun Lin}, \bibinfo{person}{Kai Zhang},
  \bibinfo{person}{Zhenying He}, \bibinfo{person}{Yinan Jing}, {and}
  \bibinfo{person}{X~Sean Wang}.} \bibinfo{year}{2025}\natexlab{b}.
\newblock \showarticletitle{Survey of Filtered Approximate Nearest Neighbor
  Search over the Vector-Scalar Hybrid Data}.
\newblock \bibinfo{journal}{\emph{arXiv preprint arXiv:2505.06501}}
  (\bibinfo{year}{2025}).
\newblock


\bibitem[Ling et~al\mbox{.}(2004)]%
        {sessionstate}
\bibfield{author}{\bibinfo{person}{Benjamin~C. Ling}, \bibinfo{person}{Emre
  K{\i}c{\i}man}, {and} \bibinfo{person}{Armando Fox}.}
  \bibinfo{year}{2004}\natexlab{}.
\newblock \showarticletitle{Session State: Beyond Soft State}. In
  \bibinfo{booktitle}{\emph{Proceedings of the 1st Symposium on Networked
  Systems Design and Implementation (NSDI)}}.
\newblock


\bibitem[Liu et~al\mbox{.}(2026)]%
        {openanonymity}
\bibfield{author}{\bibinfo{person}{Ken~Ziyu Liu}, \bibinfo{person}{Erik Chi},
  \bibinfo{person}{Sanmi Koyejo}, \bibinfo{person}{J.~Alex Halderman}, {and}
  \bibinfo{person}{Percy Liang}.} \bibinfo{year}{2026}\natexlab{}.
\newblock \bibinfo{title}{Unlinkable Inference as a User Privacy Architecture}.
\newblock \bibinfo{howpublished}{The Open Anonymity Project}.
\newblock
\urldef\tempurl%
\url{https://openanonymity.ai/blog/unlinkable-inference/}
\showURL{%
\tempurl}


\bibitem[Liu et~al\mbox{.}(2023)]%
        {thinkinmemory}
\bibfield{author}{\bibinfo{person}{Lei Liu}, \bibinfo{person}{Xiaoyan Yang},
  \bibinfo{person}{Yue Shen}, \bibinfo{person}{Binbin Hu},
  \bibinfo{person}{Zhiqiang Zhang}, \bibinfo{person}{Jinjie Gu}, {and}
  \bibinfo{person}{Guannan Zhang}.} \bibinfo{year}{2023}\natexlab{}.
\newblock \showarticletitle{Think-in-Memory: Recalling and Post-thinking Enable
  LLMs with Long-Term Memory}.
\newblock \bibinfo{journal}{\emph{arXiv preprint arXiv:2311.08719}}
  (\bibinfo{year}{2023}).
\newblock


\bibitem[Liu et~al\mbox{.}(2024)]%
        {lostinthemiddle}
\bibfield{author}{\bibinfo{person}{Nelson~F. Liu}, \bibinfo{person}{Kevin Lin},
  \bibinfo{person}{John Hewitt}, \bibinfo{person}{Ashwin Paranjape},
  \bibinfo{person}{Michele Bevilacqua}, \bibinfo{person}{Fabio Petroni}, {and}
  \bibinfo{person}{Percy Liang}.} \bibinfo{year}{2024}\natexlab{}.
\newblock \showarticletitle{Lost in the Middle: How Language Models Use Long
  Contexts}.
\newblock \bibinfo{journal}{\emph{Trans. Assoc. Comput. Linguistics}}
  \bibinfo{volume}{12} (\bibinfo{year}{2024}), \bibinfo{pages}{157--173}.
\newblock
\href{https://doi.org/10.1162/TACL\_A\_00638}{doi:\nolinkurl{10.1162/TACL\_A\_00638}}


\bibitem[Mageirakos et~al\mbox{.}(2025)]%
        {crackivf}
\bibfield{author}{\bibinfo{person}{Vasilis Mageirakos}, \bibinfo{person}{Bowen
  Wu}, {and} \bibinfo{person}{Gustavo Alonso}.}
  \bibinfo{year}{2025}\natexlab{}.
\newblock \showarticletitle{Cracking Vector Search Indexes}.
\newblock \bibinfo{journal}{\emph{Proc. {VLDB} Endow.}} \bibinfo{volume}{18},
  \bibinfo{number}{11} (\bibinfo{year}{2025}), \bibinfo{pages}{3951--3964}.
\newblock


\bibitem[Maharana et~al\mbox{.}(2024)]%
        {locomo}
\bibfield{author}{\bibinfo{person}{Adyasha Maharana},
  \bibinfo{person}{Dong{-}Ho Lee}, \bibinfo{person}{Sergey Tulyakov},
  \bibinfo{person}{Mohit Bansal}, \bibinfo{person}{Francesco Barbieri}, {and}
  \bibinfo{person}{Yuwei Fang}.} \bibinfo{year}{2024}\natexlab{}.
\newblock \showarticletitle{Evaluating Very Long-Term Conversational Memory of
  {LLM} Agents}. In \bibinfo{booktitle}{\emph{{ACL} {(1)}}}.
  \bibinfo{publisher}{Association for Computational Linguistics},
  \bibinfo{pages}{13851--13870}.
\newblock


\bibitem[Matetic et~al\mbox{.}(2017)]%
        {rote}
\bibfield{author}{\bibinfo{person}{Sinisa Matetic}, \bibinfo{person}{Mansoor
  Ahmed}, \bibinfo{person}{Kari Kostiainen}, \bibinfo{person}{Aritra Dhar},
  \bibinfo{person}{David~M. Sommer}, \bibinfo{person}{Arthur Gervais},
  \bibinfo{person}{Ari Juels}, {and} \bibinfo{person}{Srdjan Capkun}.}
  \bibinfo{year}{2017}\natexlab{}.
\newblock \showarticletitle{{ROTE:} Rollback Protection for Trusted Execution}.
  In \bibinfo{booktitle}{\emph{{USENIX} Security Symposium}}.
  \bibinfo{publisher}{{USENIX} Association}, \bibinfo{pages}{1289--1306}.
\newblock


\bibitem[McClain et~al\mbox{.}(2023)]%
        {pewprivacy2023}
\bibfield{author}{\bibinfo{person}{Colleen McClain}, \bibinfo{person}{Michelle
  Faverio}, \bibinfo{person}{Monica Anderson}, {and} \bibinfo{person}{Eugenie
  Park}.} \bibinfo{year}{2023}\natexlab{}.
\newblock \bibinfo{title}{Key Findings About {Americans} and Data Privacy}.
\newblock \bibinfo{howpublished}{Pew Research Center}.
\newblock
\urldef\tempurl%
\url{https://www.pewresearch.org/internet/2023/10/18/how-americans-view-data-privacy/}
\showURL{%
\tempurl}


\bibitem[Mei et~al\mbox{.}(2026)]%
        {atmbench}
\bibfield{author}{\bibinfo{person}{Jingbiao Mei}, \bibinfo{person}{Jinghong
  Chen}, \bibinfo{person}{Guangyu Yang}, \bibinfo{person}{Xinyu Hou},
  \bibinfo{person}{Margaret Li}, {and} \bibinfo{person}{Bill Byrne}.}
  \bibinfo{year}{2026}\natexlab{}.
\newblock \showarticletitle{According to Me: Long-Term Personalized Referential
  Memory QA}.
\newblock \bibinfo{journal}{\emph{arXiv preprint arXiv:2603.01990}}
  (\bibinfo{year}{2026}).
\newblock


\bibitem[Meijer and van Gastel(2019)]%
        {selfencryptingdeception}
\bibfield{author}{\bibinfo{person}{Carlo Meijer} {and} \bibinfo{person}{Bernard
  van Gastel}.} \bibinfo{year}{2019}\natexlab{}.
\newblock \showarticletitle{Self-Encrypting Deception: Weaknesses in the
  Encryption of Solid State Drives}. In \bibinfo{booktitle}{\emph{{IEEE}
  Symposium on Security and Privacy}}. \bibinfo{publisher}{{IEEE}},
  \bibinfo{pages}{72--87}.
\newblock


\bibitem[{Mem Labs}(2026)]%
        {memai}
\bibfield{author}{\bibinfo{person}{{Mem Labs}}.}
  \bibinfo{year}{2026}\natexlab{}.
\newblock \bibinfo{title}{Mem.ai}.
\newblock \bibinfo{howpublished}{\url{https://get.mem.ai/}}.
\newblock


\bibitem[{Meta}(2024)]%
        {messenger}
\bibfield{author}{\bibinfo{person}{{Meta}}.} \bibinfo{year}{2024}\natexlab{}.
\newblock \bibinfo{title}{{End-to-end encryption on Messenger explained}}.
\newblock
  \bibinfo{howpublished}{\url{https://about.fb.com/news/2024/03/end-to-end-encryption-on-messenger-explained/}}.
\newblock


\bibitem[Meulemeester et~al\mbox{.}(2025)]%
        {badram}
\bibfield{author}{\bibinfo{person}{Jesse~De Meulemeester},
  \bibinfo{person}{Luca Wilke}, \bibinfo{person}{David~F. Oswald},
  \bibinfo{person}{Thomas Eisenbarth}, \bibinfo{person}{Ingrid Verbauwhede},
  {and} \bibinfo{person}{Jo~Van Bulck}.} \bibinfo{year}{2025}\natexlab{}.
\newblock \showarticletitle{BadRAM: Practical Memory Aliasing Attacks on
  Trusted Execution Environments}. In \bibinfo{booktitle}{\emph{{SP}}}.
  \bibinfo{publisher}{{IEEE}}, \bibinfo{pages}{4117--4135}.
\newblock


\bibitem[{Microsoft}(2025)]%
        {microsoft2025infiniteworkday}
\bibfield{author}{\bibinfo{person}{{Microsoft}}.}
  \bibinfo{year}{2025}\natexlab{}.
\newblock \bibinfo{title}{Breaking Down the Infinite Workday}.
\newblock
  \bibinfo{howpublished}{\url{https://www.microsoft.com/en-us/worklab/work-trend-index/breaking-down-infinite-workday}}.
\newblock
\newblock
\shownote{2025 Work Trend Index Annual Report.}.


\bibitem[{Microsoft}(2026)]%
        {windowsrecall}
\bibfield{author}{\bibinfo{person}{{Microsoft}}.}
  \bibinfo{year}{2026}\natexlab{}.
\newblock \bibinfo{title}{Windows Copilot "Recall"}.
\newblock
  \bibinfo{howpublished}{\url{https://learn.microsoft.com/en-us/windows/ai/recall/}}.
\newblock


\bibitem[Mishra et~al\mbox{.}(2025)]%
        {oops}
\bibfield{author}{\bibinfo{person}{Nimish Mishra}, \bibinfo{person}{Kislay
  Arya}, \bibinfo{person}{Sarani Bhattacharya}, \bibinfo{person}{Paritosh
  Saxena}, {and} \bibinfo{person}{Debdeep Mukhopadhyay}.}
  \bibinfo{year}{2025}\natexlab{}.
\newblock \showarticletitle{"OOPS!": Out-Of-Band Remote Power Side-Channel
  Attacks on Intel {SGX} and {TDX}}. In \bibinfo{booktitle}{\emph{{DAC}}}.
  \bibinfo{publisher}{{IEEE}}, \bibinfo{pages}{1--7}.
\newblock


\bibitem[Mishra et~al\mbox{.}(2018)]%
        {oblix}
\bibfield{author}{\bibinfo{person}{Pratyush Mishra}, \bibinfo{person}{Rishabh
  Poddar}, \bibinfo{person}{Jerry Chen}, \bibinfo{person}{Alessandro Chiesa},
  {and} \bibinfo{person}{Raluca~Ada Popa}.} \bibinfo{year}{2018}\natexlab{}.
\newblock \showarticletitle{Oblix: An Efficient Oblivious Search Index}. In
  \bibinfo{booktitle}{\emph{{IEEE} Symposium on Security and Privacy}}.
  \bibinfo{publisher}{{IEEE} Computer Society}, \bibinfo{pages}{279--296}.
\newblock


\bibitem[Misono et~al\mbox{.}(2025)]%
        {cvmexplained}
\bibfield{author}{\bibinfo{person}{Masanori Misono}, \bibinfo{person}{Dimitrios
  Stavrakakis}, \bibinfo{person}{Nuno Santos}, {and} \bibinfo{person}{Pramod
  Bhatotia}.} \bibinfo{year}{2025}\natexlab{}.
\newblock \showarticletitle{Confidential VMs Explained: An Empirical Analysis
  of {AMD} {SEV-SNP} and Intel {TDX}}. In
  \bibinfo{booktitle}{\emph{{SIGMETRICS} (Abstracts)}}.
  \bibinfo{publisher}{{ACM}}, \bibinfo{pages}{34--36}.
\newblock


\bibitem[Mitra and Gilbert(2013)]%
        {gossipworkplace}
\bibfield{author}{\bibinfo{person}{Tanushree Mitra} {and} \bibinfo{person}{Eric
  Gilbert}.} \bibinfo{year}{2013}\natexlab{}.
\newblock \showarticletitle{"Analyzing gossip in workplace email" by Tanushree
  Mitra and Eric Gilbert, with Ching-man Au Yeung as coordinator}.
\newblock \bibinfo{journal}{\emph{{SIGWEB} Newsl.}} \bibinfo{volume}{2013},
  \bibinfo{number}{Winter} (\bibinfo{year}{2013}), \bibinfo{pages}{5:1--5:7}.
\newblock


\bibitem[Moghimi et~al\mbox{.}(2017)]%
        {cachezoom}
\bibfield{author}{\bibinfo{person}{Ahmad Moghimi}, \bibinfo{person}{Gorka
  Irazoqui}, {and} \bibinfo{person}{Thomas Eisenbarth}.}
  \bibinfo{year}{2017}\natexlab{}.
\newblock \showarticletitle{CacheZoom: How {SGX} Amplifies the Power of Cache
  Attacks}. In \bibinfo{booktitle}{\emph{{CHES}}}
  \emph{(\bibinfo{series}{Lecture Notes in Computer Science})}.
  \bibinfo{publisher}{Springer}, \bibinfo{pages}{69--90}.
\newblock


\bibitem[Mohoney et~al\mbox{.}(2023)]%
        {hqi}
\bibfield{author}{\bibinfo{person}{Jason Mohoney}, \bibinfo{person}{Anil
  Pacaci}, \bibinfo{person}{Shihabur~Rahman Chowdhury}, \bibinfo{person}{Ali
  Mousavi}, \bibinfo{person}{Ihab~F. Ilyas}, \bibinfo{person}{Umar~Farooq
  Minhas}, \bibinfo{person}{Jeffrey Pound}, {and} \bibinfo{person}{Theodoros
  Rekatsinas}.} \bibinfo{year}{2023}\natexlab{}.
\newblock \showarticletitle{High-Throughput Vector Similarity Search in
  Knowledge Graphs}.
\newblock \bibinfo{journal}{\emph{Proc. {ACM} Manag. Data}}
  \bibinfo{volume}{1}, \bibinfo{number}{2} (\bibinfo{year}{2023}),
  \bibinfo{pages}{197:1--197:25}.
\newblock


\bibitem[Mohoney et~al\mbox{.}(2025)]%
        {quake}
\bibfield{author}{\bibinfo{person}{Jason Mohoney}, \bibinfo{person}{Devesh
  Sarda}, \bibinfo{person}{Mengze Tang}, \bibinfo{person}{Shihabur~Rahman
  Chowdhury}, \bibinfo{person}{Anil Pacaci}, \bibinfo{person}{Ihab~F. Ilyas},
  \bibinfo{person}{Theodoros Rekatsinas}, {and} \bibinfo{person}{Shivaram
  Venkataraman}.} \bibinfo{year}{2025}\natexlab{}.
\newblock \showarticletitle{Quake: Adaptive Indexing for Vector Search}. In
  \bibinfo{booktitle}{\emph{{OSDI}}}. \bibinfo{publisher}{{USENIX}
  Association}, \bibinfo{pages}{153--169}.
\newblock


\bibitem[Morris et~al\mbox{.}(2023)]%
        {textembeddingsrevealalmostastext}
\bibfield{author}{\bibinfo{person}{John~X. Morris}, \bibinfo{person}{Volodymyr
  Kuleshov}, \bibinfo{person}{Vitaly Shmatikov}, {and}
  \bibinfo{person}{Alexander~M. Rush}.} \bibinfo{year}{2023}\natexlab{}.
\newblock \showarticletitle{Text Embeddings Reveal (Almost) As Much As Text}.
  In \bibinfo{booktitle}{\emph{{EMNLP}}}. \bibinfo{publisher}{Association for
  Computational Linguistics}, \bibinfo{pages}{12448--12460}.
\newblock


\bibitem[Murdock et~al\mbox{.}(2020)]%
        {plundervolt}
\bibfield{author}{\bibinfo{person}{Kit Murdock}, \bibinfo{person}{David~F.
  Oswald}, \bibinfo{person}{Flavio~D. Garcia}, \bibinfo{person}{Jo~Van Bulck},
  \bibinfo{person}{Daniel Gruss}, {and} \bibinfo{person}{Frank Piessens}.}
  \bibinfo{year}{2020}\natexlab{}.
\newblock \showarticletitle{Plundervolt: Software-based Fault Injection Attacks
  against Intel {SGX}}. In \bibinfo{booktitle}{\emph{{SP}}}.
  \bibinfo{publisher}{{IEEE}}, \bibinfo{pages}{1466--1482}.
\newblock


\bibitem[Naghibijouybari et~al\mbox{.}(2018)]%
        {renderedinsecure}
\bibfield{author}{\bibinfo{person}{Hoda Naghibijouybari},
  \bibinfo{person}{Ajaya Neupane}, \bibinfo{person}{Zhiyun Qian}, {and}
  \bibinfo{person}{Nael~B. Abu{-}Ghazaleh}.} \bibinfo{year}{2018}\natexlab{}.
\newblock \showarticletitle{Rendered Insecure: {GPU} Side Channel Attacks are
  Practical}. In \bibinfo{booktitle}{\emph{{CCS}}}. \bibinfo{publisher}{{ACM}},
  \bibinfo{pages}{2139--2153}.
\newblock


\bibitem[Nan et~al\mbox{.}(2025)]%
        {nemori}
\bibfield{author}{\bibinfo{person}{Jiayan Nan}, \bibinfo{person}{Wenquan Ma},
  \bibinfo{person}{Wenlong Wu}, {and} \bibinfo{person}{Yize Chen}.}
  \bibinfo{year}{2025}\natexlab{}.
\newblock \showarticletitle{Nemori: Self-Organizing Agent Memory Inspired by
  Cognitive Science}.
\newblock \bibinfo{journal}{\emph{arXiv preprint arXiv:2508.03341}}
  (\bibinfo{year}{2025}).
\newblock


\bibitem[Niu et~al\mbox{.}(2022)]%
        {narrator}
\bibfield{author}{\bibinfo{person}{Jianyu Niu}, \bibinfo{person}{Wei Peng},
  \bibinfo{person}{Xiaokuan Zhang}, {and} \bibinfo{person}{Yinqian Zhang}.}
  \bibinfo{year}{2022}\natexlab{}.
\newblock \showarticletitle{{NARRATOR:} Secure and Practical State Continuity
  for Trusted Execution in the Cloud}. In \bibinfo{booktitle}{\emph{{CCS}}}.
  \bibinfo{publisher}{{ACM}}, \bibinfo{pages}{2385--2399}.
\newblock


\bibitem[{NVIDIA}(2025)]%
        {nvidiasecureai}
\bibfield{author}{\bibinfo{person}{{NVIDIA}}.} \bibinfo{year}{2025}\natexlab{}.
\newblock \bibinfo{title}{{NVIDIA} Secure {AI} with Blackwell and Hopper
  {GPUs}}.
\newblock
  \bibinfo{howpublished}{\url{https://docs.nvidia.com/nvidia-secure-ai-with-blackwell-and-hopper-gpus-whitepaper.pdf}}.
\newblock


\bibitem[OpenAI(2024)]%
        {openaimemory}
\bibfield{author}{\bibinfo{person}{OpenAI}.} \bibinfo{year}{2024}\natexlab{}.
\newblock \bibinfo{title}{Memory and New Controls for {C}hat{GPT}}.
\newblock
  \bibinfo{howpublished}{\url{https://openai.com/index/memory-and-new-controls-for-chatgpt/}}.
\newblock


\bibitem[{OpenAI}(2025)]%
        {openaigov2025}
\bibfield{author}{\bibinfo{person}{{OpenAI}}.} \bibinfo{year}{2025}\natexlab{}.
\newblock \bibinfo{booktitle}{\emph{Report on Government Requests for User Data
  ({H1} 2025)}}.
\newblock \bibinfo{type}{{T}echnical {R}eport}.
\newblock
\urldef\tempurl%
\url{https://cdn.openai.com/trust-and-transparency/report-2025h1-government-requests-for-user-data.pdf}
\showURL{%
\tempurl}


\bibitem[{OpenAI}(2026)]%
        {chatgpt}
\bibfield{author}{\bibinfo{person}{{OpenAI}}.} \bibinfo{year}{2026}\natexlab{}.
\newblock \bibinfo{title}{ChatGPT}.
\newblock \bibinfo{howpublished}{\url{https://chatgpt.com/}}.
\newblock


\bibitem[Ostrovsky(1990)]%
        {oram90}
\bibfield{author}{\bibinfo{person}{Rafail Ostrovsky}.}
  \bibinfo{year}{1990}\natexlab{}.
\newblock \showarticletitle{Efficient Computation on Oblivious RAMs}. In
  \bibinfo{booktitle}{\emph{{STOC}}}. \bibinfo{publisher}{{ACM}},
  \bibinfo{pages}{514--523}.
\newblock


\bibitem[Ostrovsky(1992)]%
        {oram96}
\bibfield{author}{\bibinfo{person}{Rafail Ostrovsky}.}
  \bibinfo{year}{1992}\natexlab{}.
\newblock \emph{\bibinfo{title}{Software protection and simulation on oblivious
  RAMs}}.
\newblock \bibinfo{thesistype}{Ph.\,D. Dissertation}.
  \bibinfo{school}{Massachusetts Institute of Technology, Cambridge, MA,
  {USA}}.
\newblock
\urldef\tempurl%
\url{https://hdl.handle.net/1721.1/103684}
\showURL{%
\tempurl}


\bibitem[Packer et~al\mbox{.}(2024)]%
        {memgpt}
\bibfield{author}{\bibinfo{person}{Charles Packer}, \bibinfo{person}{Sarah
  Wooders}, \bibinfo{person}{Kevin Lin}, \bibinfo{person}{Vivian Fang},
  \bibinfo{person}{Shishir~G. Patil}, \bibinfo{person}{Ion Stoica}, {and}
  \bibinfo{person}{Joseph~E. Gonzalez}.} \bibinfo{year}{2024}\natexlab{}.
\newblock \showarticletitle{MemGPT: Towards LLMs as Operating Systems}.
\newblock \bibinfo{journal}{\emph{arXiv preprint arXiv:2310.08560}}
  (\bibinfo{year}{2024}).
\newblock


\bibitem[Pan et~al\mbox{.}(2025)]%
        {secom}
\bibfield{author}{\bibinfo{person}{Zhuoshi Pan}, \bibinfo{person}{Qianhui Wu},
  \bibinfo{person}{Huiqiang Jiang}, \bibinfo{person}{Xufang Luo},
  \bibinfo{person}{Hao Cheng}, \bibinfo{person}{Dongsheng Li},
  \bibinfo{person}{Yuqing Yang}, \bibinfo{person}{Chin{-}Yew Lin},
  \bibinfo{person}{H.~Vicky Zhao}, \bibinfo{person}{Lili Qiu}, {and}
  \bibinfo{person}{Jianfeng Gao}.} \bibinfo{year}{2025}\natexlab{}.
\newblock \showarticletitle{SeCom: On Memory Construction and Retrieval for
  Personalized Conversational Agents}. In \bibinfo{booktitle}{\emph{{ICLR}}}.
  \bibinfo{publisher}{OpenReview.net}.
\newblock


\bibitem[Park et~al\mbox{.}(2023)]%
        {generativeagents}
\bibfield{author}{\bibinfo{person}{Joon~Sung Park}, \bibinfo{person}{Joseph~C.
  O'Brien}, \bibinfo{person}{Carrie~Jun Cai}, \bibinfo{person}{Meredith~Ringel
  Morris}, \bibinfo{person}{Percy Liang}, {and} \bibinfo{person}{Michael~S.
  Bernstein}.} \bibinfo{year}{2023}\natexlab{}.
\newblock \showarticletitle{Generative Agents: Interactive Simulacra of Human
  Behavior}. In \bibinfo{booktitle}{\emph{{UIST}}}. \bibinfo{publisher}{{ACM}},
  \bibinfo{pages}{2:1--2:22}.
\newblock


\bibitem[Parno et~al\mbox{.}(2011)]%
        {memoir}
\bibfield{author}{\bibinfo{person}{Bryan Parno}, \bibinfo{person}{Jacob~R.
  Lorch}, \bibinfo{person}{John~R. Douceur}, \bibinfo{person}{James~W.
  Mickens}, {and} \bibinfo{person}{Jonathan~M. McCune}.}
  \bibinfo{year}{2011}\natexlab{}.
\newblock \showarticletitle{Memoir: Practical State Continuity for Protected
  Modules}. In \bibinfo{booktitle}{\emph{{IEEE} Symposium on Security and
  Privacy}}. \bibinfo{publisher}{{IEEE} Computer Society},
  \bibinfo{pages}{379--394}.
\newblock


\bibitem[Pass et~al\mbox{.}(2017)]%
        {formalabstractionsforenclaveattestation}
\bibfield{author}{\bibinfo{person}{Rafael Pass}, \bibinfo{person}{Elaine Shi},
  {and} \bibinfo{person}{Florian Tram{\`{e}}r}.}
  \bibinfo{year}{2017}\natexlab{}.
\newblock \showarticletitle{Formal Abstractions for Attested Execution Secure
  Processors}. In \bibinfo{booktitle}{\emph{{EUROCRYPT} {(1)}}}
  \emph{(\bibinfo{series}{Lecture Notes in Computer Science})}.
  \bibinfo{pages}{260--289}.
\newblock


\bibitem[Patel et~al\mbox{.}(2024)]%
        {acorn}
\bibfield{author}{\bibinfo{person}{Liana Patel}, \bibinfo{person}{Peter Kraft},
  \bibinfo{person}{Carlos Guestrin}, {and} \bibinfo{person}{Matei Zaharia}.}
  \bibinfo{year}{2024}\natexlab{}.
\newblock \showarticletitle{{ACORN:} Performant and Predicate-Agnostic Search
  Over Vector Embeddings and Structured Data}.
\newblock \bibinfo{journal}{\emph{Proc. {ACM} Manag. Data}}
  \bibinfo{volume}{2}, \bibinfo{number}{3} (\bibinfo{year}{2024}),
  \bibinfo{pages}{120}.
\newblock


\bibitem[{PCI-SIG}(2020)]%
        {pcieide}
\bibfield{author}{\bibinfo{person}{{PCI-SIG}}.}
  \bibinfo{year}{2020}\natexlab{}.
\newblock \bibinfo{title}{{PCIe} {Base} {Specification} {Revision} 5.0:
  {Integrity} and {Data} {Encryption} ({IDE}) {ECN}}.
\newblock
  \bibinfo{howpublished}{\url{https://pcisig.com/PCI\%20Express/ECN/Base/IntegrityandDataEncryption}}.
\newblock


\bibitem[{PCI-SIG}(2022)]%
        {pcietdisp}
\bibfield{author}{\bibinfo{person}{{PCI-SIG}}.}
  \bibinfo{year}{2022}\natexlab{}.
\newblock \bibinfo{title}{{TEE} {Device} {Interface} {Security} {Protocol}
  ({TDISP}), {Revision} 1.0}.
\newblock
  \bibinfo{howpublished}{\url{https://pcisig.com/PCI\%20Express/ECN/Base/TEEDeviceInterfaceSecurityProtocol}}.
\newblock


\bibitem[{Personal.ai}(2026)]%
        {personalai}
\bibfield{author}{\bibinfo{person}{{Personal.ai}}.}
  \bibinfo{year}{2026}\natexlab{}.
\newblock \bibinfo{title}{Personal.ai}.
\newblock \bibinfo{howpublished}{\url{https://www.personal.ai/}}.
\newblock


\bibitem[{Pinecone Systems}(2026)]%
        {pinecone}
\bibfield{author}{\bibinfo{person}{{Pinecone Systems}}.}
  \bibinfo{year}{2026}\natexlab{}.
\newblock \bibinfo{title}{Pinecone}.
\newblock \bibinfo{howpublished}{\url{https://www.pinecone.io/}}.
\newblock


\bibitem[Poliakov and Shvai(2024)]%
        {multimetarag}
\bibfield{author}{\bibinfo{person}{Mykhailo Poliakov} {and}
  \bibinfo{person}{Nadiya Shvai}.} \bibinfo{year}{2024}\natexlab{}.
\newblock \showarticletitle{Multi-Meta-RAG: Improving RAG for Multi-Hop Queries
  using Database Filtering with LLM-Extracted Metadata}.
\newblock \bibinfo{journal}{\emph{arXiv preprint arXiv:2406.13213}}
  (\bibinfo{year}{2024}).
\newblock


\bibitem[Priebe et~al\mbox{.}(2018)]%
        {enclavedb}
\bibfield{author}{\bibinfo{person}{Christian Priebe}, \bibinfo{person}{Kapil
  Vaswani}, {and} \bibinfo{person}{Manuel Costa}.}
  \bibinfo{year}{2018}\natexlab{}.
\newblock \showarticletitle{EnclaveDB: {A} Secure Database Using {SGX}}. In
  \bibinfo{booktitle}{\emph{{IEEE} Symposium on Security and Privacy}}.
  \bibinfo{publisher}{{IEEE} Computer Society}, \bibinfo{pages}{264--278}.
\newblock


\bibitem[Qiang et~al\mbox{.}(2024)]%
        {genai2024worldbank}
\bibfield{author}{\bibinfo{person}{Christine~Zhenwei Qiang},
  \bibinfo{person}{Yan Liu}, {and} \bibinfo{person}{He Wang}.}
  \bibinfo{year}{2024}\natexlab{}.
\newblock \bibinfo{title}{Who on earth is using generative {AI}?}
\newblock
  \bibinfo{howpublished}{\url{https://blogs.worldbank.org/en/digital-development/who-on-earth-is-using-generative-ai-}}.
\newblock


\bibitem[Rasmussen et~al\mbox{.}(2025)]%
        {zep}
\bibfield{author}{\bibinfo{person}{Preston Rasmussen}, \bibinfo{person}{Pavlo
  Paliychuk}, \bibinfo{person}{Travis Beauvais}, \bibinfo{person}{Jack Ryan},
  {and} \bibinfo{person}{Daniel Chalef}.} \bibinfo{year}{2025}\natexlab{}.
\newblock \showarticletitle{Zep: A Temporal Knowledge Graph Architecture for
  Agent Memory}.
\newblock \bibinfo{journal}{\emph{arXiv preprint arXiv:2501.13956}}
  (\bibinfo{year}{2025}).
\newblock


\bibitem[Rauscher et~al\mbox{.}(2025)]%
        {tdxploit}
\bibfield{author}{\bibinfo{person}{Fabian Rauscher}, \bibinfo{person}{Luca
  Wilke}, \bibinfo{person}{Hannes Weissteiner}, \bibinfo{person}{Thomas
  Eisenbarth}, {and} \bibinfo{person}{Daniel Gruss}.}
  \bibinfo{year}{2025}\natexlab{}.
\newblock \showarticletitle{TDXploit: Novel Techniques for Single-Stepping and
  Cache Attacks on Intel {TDX}}. In \bibinfo{booktitle}{\emph{{USENIX} Security
  Symposium}}. \bibinfo{publisher}{{USENIX} Association},
  \bibinfo{pages}{1207--1222}.
\newblock


\bibitem[Rayner and Clifton(2009)]%
        {raynerwpm}
\bibfield{author}{\bibinfo{person}{Keith Rayner} {and} \bibinfo{person}{Charles
  Clifton}.} \bibinfo{year}{2009}\natexlab{}.
\newblock \showarticletitle{Language processing in reading and speech
  perception is fast and incremental: Implications for event-related potential
  research}.
\newblock \bibinfo{journal}{\emph{Biological Psychology}} \bibinfo{volume}{80},
  \bibinfo{number}{1} (\bibinfo{year}{2009}), \bibinfo{pages}{4--9}.
\newblock
\href{https://doi.org/10.1016/j.biopsycho.2008.05.002}{doi:\nolinkurl{10.1016/j.biopsycho.2008.05.002}}


\bibitem[Rebello et~al\mbox{.}(2021)]%
        {fsyncfailures}
\bibfield{author}{\bibinfo{person}{Anthony Rebello}, \bibinfo{person}{Yuvraj
  Patel}, \bibinfo{person}{Ramnatthan Alagappan}, \bibinfo{person}{Andrea~C.
  Arpaci{-}Dusseau}, {and} \bibinfo{person}{Remzi~H. Arpaci{-}Dusseau}.}
  \bibinfo{year}{2021}\natexlab{}.
\newblock \showarticletitle{Can Applications Recover from fsync Failures?}
\newblock \bibinfo{journal}{\emph{{ACM} Trans. Storage}} \bibinfo{volume}{17},
  \bibinfo{number}{2} (\bibinfo{year}{2021}), \bibinfo{pages}{12:1--12:30}.
\newblock
\href{https://doi.org/10.1145/3450338}{doi:\nolinkurl{10.1145/3450338}}


\bibitem[{Reclaim.ai}(2024)]%
        {reclaim2024}
\bibfield{author}{\bibinfo{person}{{Reclaim.ai}}.}
  \bibinfo{year}{2024}\natexlab{}.
\newblock \bibinfo{title}{Smart Meetings Trends Report (145+ Stats)}.
\newblock
  \bibinfo{howpublished}{\url{https://reclaim.ai/blog/smart-meetings-report}}.
\newblock


\bibitem[Ren et~al\mbox{.}(2015)]%
        {ringoram}
\bibfield{author}{\bibinfo{person}{Ling Ren}, \bibinfo{person}{Christopher~W.
  Fletcher}, \bibinfo{person}{Albert Kwon}, \bibinfo{person}{Emil Stefanov},
  \bibinfo{person}{Elaine Shi}, \bibinfo{person}{Marten van Dijk}, {and}
  \bibinfo{person}{Srinivas Devadas}.} \bibinfo{year}{2015}\natexlab{}.
\newblock \showarticletitle{Constants Count: Practical Improvements to
  Oblivious {RAM}}. In \bibinfo{booktitle}{\emph{{USENIX} Security Symposium}}.
  \bibinfo{publisher}{{USENIX} Association}, \bibinfo{pages}{415--430}.
\newblock


\bibitem[{Rust Random Contributors}(2025)]%
        {randchacha}
\bibfield{author}{\bibinfo{person}{{Rust Random Contributors}}.}
  \bibinfo{year}{2025}\natexlab{}.
\newblock \bibinfo{title}{{rand\_chacha}: {ChaCha} random number generator}.
\newblock \bibinfo{howpublished}{\url{https://github.com/rust-random/rand}}.
\newblock


\bibitem[{RustCrypto Developers}(2026)]%
        {rustcrypto}
\bibfield{author}{\bibinfo{person}{{RustCrypto Developers}}.}
  \bibinfo{year}{2026}\natexlab{}.
\newblock \bibinfo{title}{{RustCrypto}: Pure {Rust} Implementations of
  Cryptographic Algorithms}.
\newblock \bibinfo{howpublished}{\url{https://github.com/RustCrypto}}.
\newblock


\bibitem[Sappelli et~al\mbox{.}(2016)]%
        {emailintent}
\bibfield{author}{\bibinfo{person}{Maya Sappelli}, \bibinfo{person}{Gabriella
  Pasi}, \bibinfo{person}{Suzan Verberne}, \bibinfo{person}{Maaike de Boer},
  {and} \bibinfo{person}{Wessel Kraaij}.} \bibinfo{year}{2016}\natexlab{}.
\newblock \showarticletitle{Assessing e-mail intent and tasks in e-mail
  messages}.
\newblock \bibinfo{journal}{\emph{Inf. Sci.}}  \bibinfo{volume}{358-359}
  (\bibinfo{year}{2016}), \bibinfo{pages}{1--17}.
\newblock


\bibitem[Sarkar et~al\mbox{.}(2021)]%
        {parallelchat}
\bibfield{author}{\bibinfo{person}{Advait Sarkar}, \bibinfo{person}{Sean
  Rintel}, \bibinfo{person}{Damian Borowiec}, \bibinfo{person}{Rachel
  Bergmann}, \bibinfo{person}{Sharon Gillett}, \bibinfo{person}{Danielle
  Bragg}, \bibinfo{person}{Nancy Baym}, {and} \bibinfo{person}{Abigail
  Sellen}.} \bibinfo{year}{2021}\natexlab{}.
\newblock \showarticletitle{The promise and peril of parallel chat in video
  meetings for work}. In \bibinfo{booktitle}{\emph{{CHI} Extended Abstracts}}.
  \bibinfo{publisher}{{ACM}}, \bibinfo{pages}{260:1--260:8}.
\newblock


\bibitem[Sarthi et~al\mbox{.}(2024)]%
        {raptor}
\bibfield{author}{\bibinfo{person}{Parth Sarthi}, \bibinfo{person}{Salman
  Abdullah}, \bibinfo{person}{Aditi Tuli}, \bibinfo{person}{Shubh Khanna},
  \bibinfo{person}{Anna Goldie}, {and} \bibinfo{person}{Christopher~D.
  Manning}.} \bibinfo{year}{2024}\natexlab{}.
\newblock \showarticletitle{{RAPTOR:} Recursive Abstractive Processing for
  Tree-Organized Retrieval}. In \bibinfo{booktitle}{\emph{{ICLR}}}.
  \bibinfo{publisher}{OpenReview.net}.
\newblock


\bibitem[Schl{\"{u}}ter et~al\mbox{.}(2024a)]%
        {wesee}
\bibfield{author}{\bibinfo{person}{Benedict Schl{\"{u}}ter},
  \bibinfo{person}{Supraja Sridhara}, \bibinfo{person}{Andrin Bertschi}, {and}
  \bibinfo{person}{Shweta Shinde}.} \bibinfo{year}{2024}\natexlab{a}.
\newblock \showarticletitle{WeSee: Using Malicious {\#}VC Interrupts to Break
  {AMD} {SEV-SNP}}. In \bibinfo{booktitle}{\emph{{SP}}}.
  \bibinfo{publisher}{{IEEE}}, \bibinfo{pages}{4220--4238}.
\newblock


\bibitem[Schl{\"{u}}ter et~al\mbox{.}(2024b)]%
        {heckler}
\bibfield{author}{\bibinfo{person}{Benedict Schl{\"{u}}ter},
  \bibinfo{person}{Supraja Sridhara}, \bibinfo{person}{Mark Kuhne},
  \bibinfo{person}{Andrin Bertschi}, {and} \bibinfo{person}{Shweta Shinde}.}
  \bibinfo{year}{2024}\natexlab{b}.
\newblock \showarticletitle{{HECKLER:} Breaking Confidential VMs with Malicious
  Interrupts}. In \bibinfo{booktitle}{\emph{{USENIX} Security Symposium}}.
  \bibinfo{publisher}{{USENIX} Association}.
\newblock


\bibitem[Schroeder and Gibson(2007)]%
        {diskfailuresrealworld}
\bibfield{author}{\bibinfo{person}{Bianca Schroeder} {and}
  \bibinfo{person}{Garth~A. Gibson}.} \bibinfo{year}{2007}\natexlab{}.
\newblock \showarticletitle{Disk Failures in the Real World: What Does an
  {MTTF} of 1, 000, 000 Hours Mean to You?}. In
  \bibinfo{booktitle}{\emph{{FAST}}}. \bibinfo{publisher}{{USENIX}},
  \bibinfo{pages}{1--16}.
\newblock


\bibitem[{Seagate Technology}(2020)]%
        {seagate2020}
\bibfield{author}{\bibinfo{person}{{Seagate Technology}}.}
  \bibinfo{year}{2020}\natexlab{}.
\newblock \bibinfo{title}{Rethink Data: Put More of Your Business Data to
  Work--From Edge to Cloud}.
\newblock
  \bibinfo{howpublished}{\url{https://www.seagate.com/files/www-content/our-story/rethink-data/files/Rethink_Data_Report_2020.pdf}}.
\newblock


\bibitem[Sehwag et~al\mbox{.}(2025)]%
        {awsagentcorememory}
\bibfield{author}{\bibinfo{person}{Akarsha Sehwag}, \bibinfo{person}{Dani
  Mitchell}, \bibinfo{person}{Gopikrishnan Anilkumar}, \bibinfo{person}{Mani
  Khanuja}, {and} \bibinfo{person}{Noor Randhawa}.}
  \bibinfo{year}{2025}\natexlab{}.
\newblock \bibinfo{title}{Amazon Bedrock AgentCore Memory: Building
  context-aware agents}.
\newblock
  \bibinfo{howpublished}{\url{https://aws.amazon.com/blogs/machine-learning/amazon-bedrock-agentcore-memory-building-context-aware-agents/}}.
\newblock


\bibitem[Shen et~al\mbox{.}(2020)]%
        {occlum}
\bibfield{author}{\bibinfo{person}{Youren Shen}, \bibinfo{person}{Hongliang
  Tian}, \bibinfo{person}{Yu Chen}, \bibinfo{person}{Kang Chen},
  \bibinfo{person}{Runji Wang}, \bibinfo{person}{Yi Xu}, \bibinfo{person}{Yubin
  Xia}, {and} \bibinfo{person}{Shoumeng Yan}.} \bibinfo{year}{2020}\natexlab{}.
\newblock \showarticletitle{Occlum: Secure and Efficient Multitasking Inside a
  Single Enclave of Intel {SGX}}. In \bibinfo{booktitle}{\emph{{ASPLOS}}}.
  \bibinfo{publisher}{{ACM}}, \bibinfo{pages}{955--970}.
\newblock


\bibitem[Shezan et~al\mbox{.}(2020)]%
        {voicepersonalassistant}
\bibfield{author}{\bibinfo{person}{Faysal~Hossain Shezan},
  \bibinfo{person}{Hang Hu}, \bibinfo{person}{Jiamin Wang},
  \bibinfo{person}{Gang Wang}, {and} \bibinfo{person}{Yuan Tian}.}
  \bibinfo{year}{2020}\natexlab{}.
\newblock \showarticletitle{Read Between the Lines: An Empirical Measurement of
  Sensitive Applications of Voice Personal Assistant Systems}. In
  \bibinfo{booktitle}{\emph{{WWW}}}. \bibinfo{publisher}{{ACM} / {IW3C2}},
  \bibinfo{pages}{1006--1017}.
\newblock


\bibitem[Shi et~al\mbox{.}(2011)]%
        {binarytreeoram}
\bibfield{author}{\bibinfo{person}{Elaine Shi}, \bibinfo{person}{T.{-}H.~Hubert
  Chan}, \bibinfo{person}{Emil Stefanov}, {and} \bibinfo{person}{Mingfei Li}.}
  \bibinfo{year}{2011}\natexlab{}.
\newblock \showarticletitle{Oblivious {RAM} with O((logN)3) Worst-Case Cost}.
  In \bibinfo{booktitle}{\emph{{ASIACRYPT}}} \emph{(\bibinfo{series}{Lecture
  Notes in Computer Science})}. \bibinfo{publisher}{Springer},
  \bibinfo{pages}{197--214}.
\newblock


\bibitem[Shinn et~al\mbox{.}(2023)]%
        {reflexion}
\bibfield{author}{\bibinfo{person}{Noah Shinn}, \bibinfo{person}{Federico
  Cassano}, \bibinfo{person}{Ashwin Gopinath}, \bibinfo{person}{Karthik
  Narasimhan}, {and} \bibinfo{person}{Shunyu Yao}.}
  \bibinfo{year}{2023}\natexlab{}.
\newblock \showarticletitle{Reflexion: language agents with verbal
  reinforcement learning}. In \bibinfo{booktitle}{\emph{NeurIPS}}.
\newblock


\bibitem[{Signal Foundation}(2026)]%
        {signal}
\bibfield{author}{\bibinfo{person}{{Signal Foundation}}.}
  \bibinfo{year}{2026}\natexlab{}.
\newblock \bibinfo{title}{{Signal}}.
\newblock \bibinfo{howpublished}{\url{https://signal.org/}}.
\newblock


\bibitem[Singh et~al\mbox{.}(2021)]%
        {freshdiskann}
\bibfield{author}{\bibinfo{person}{Aditi Singh}, \bibinfo{person}{Suhas~Jayaram
  Subramanya}, \bibinfo{person}{Ravishankar Krishnaswamy}, {and}
  \bibinfo{person}{Harsha~Vardhan Simhadri}.} \bibinfo{year}{2021}\natexlab{}.
\newblock \showarticletitle{FreshDiskANN: A Fast and Accurate Graph-Based ANN
  Index for Streaming Similarity Search}.
\newblock \bibinfo{journal}{\emph{arXiv preprint arXiv:2105.09613}}
  (\bibinfo{year}{2021}).
\newblock


\bibitem[Skjuve et~al\mbox{.}(2023)]%
        {longitudinalstudyofselfdisclosure}
\bibfield{author}{\bibinfo{person}{Marita Skjuve}, \bibinfo{person}{Asbj{\o}rn
  F{\o}lstad}, {and} \bibinfo{person}{Petter~Bae Brandtz{\ae}g}.}
  \bibinfo{year}{2023}\natexlab{}.
\newblock \showarticletitle{A Longitudinal Study of Self-Disclosure in
  Human-Chatbot Relationships}.
\newblock \bibinfo{journal}{\emph{Interact. Comput.}} \bibinfo{volume}{35},
  \bibinfo{number}{1} (\bibinfo{year}{2023}), \bibinfo{pages}{24--39}.
\newblock


\bibitem[Smith(2011)]%
        {pewmessage}
\bibfield{author}{\bibinfo{person}{Aaron Smith}.}
  \bibinfo{year}{2011}\natexlab{}.
\newblock \bibinfo{title}{How Americans Use Text Messaging}.
\newblock
  \bibinfo{howpublished}{\url{https://www.pewresearch.org/internet/2011/09/19/how-americans-use-text-messaging/}}.
\newblock
\newblock
\shownote{Pew Research Center.}.


\bibitem[Smith(2025)]%
        {ringcrate}
\bibfield{author}{\bibinfo{person}{Brian Smith}.}
  \bibinfo{year}{2025}\natexlab{}.
\newblock \bibinfo{title}{{ring}: Safe, fast, small crypto using {Rust}}.
\newblock \bibinfo{howpublished}{\url{https://github.com/briansmith/ring}}.
\newblock


\bibitem[Soleimani et~al\mbox{.}(2025)]%
        {weave}
\bibfield{author}{\bibinfo{person}{Mahdi Soleimani}, \bibinfo{person}{Grace
  Jia}, {and} \bibinfo{person}{Anurag Khandelwal}.}
  \bibinfo{year}{2025}\natexlab{}.
\newblock \showarticletitle{Weave: Efficient and Expressive Oblivious Analytics
  at Scale}. In \bibinfo{booktitle}{\emph{{OSDI}}}.
  \bibinfo{publisher}{{USENIX} Association}, \bibinfo{pages}{939--955}.
\newblock


\bibitem[Soliman et~al\mbox{.}(2022)]%
        {mulch}
\bibfield{author}{\bibinfo{person}{Hadeel Soliman}, \bibinfo{person}{Lingfei
  Zhao}, \bibinfo{person}{Zhipeng Huang}, \bibinfo{person}{Subhadeep Paul},
  {and} \bibinfo{person}{Kevin~S. Xu}.} \bibinfo{year}{2022}\natexlab{}.
\newblock \showarticletitle{The Multivariate Community Hawkes Model for
  Dependent Relational Events in Continuous-time Networks}. In
  \bibinfo{booktitle}{\emph{{ICML}}} \emph{(\bibinfo{series}{Proceedings of
  Machine Learning Research})}. \bibinfo{publisher}{{PMLR}},
  \bibinfo{pages}{20329--20346}.
\newblock


\bibitem[Song and Raghunathan(2020)]%
        {informationleakageinembeddings}
\bibfield{author}{\bibinfo{person}{Congzheng Song} {and}
  \bibinfo{person}{Ananth Raghunathan}.} \bibinfo{year}{2020}\natexlab{}.
\newblock \showarticletitle{Information Leakage in Embedding Models}. In
  \bibinfo{booktitle}{\emph{{CCS}}}. \bibinfo{publisher}{{ACM}},
  \bibinfo{pages}{377--390}.
\newblock


\bibitem[{Stanford Institute for Human-Centered AI}(2025)]%
        {haiindex2025}
\bibfield{author}{\bibinfo{person}{{Stanford Institute for Human-Centered
  AI}}.} \bibinfo{year}{2025}\natexlab{}.
\newblock \bibinfo{title}{The 2025 {AI} Index Report}.
\newblock
  \bibinfo{howpublished}{\url{https://hai.stanford.edu/ai-index/2025-ai-index-report}}.
\newblock


\bibitem[Stefanov et~al\mbox{.}(2018)]%
        {pathoram}
\bibfield{author}{\bibinfo{person}{Emil Stefanov}, \bibinfo{person}{Marten van
  Dijk}, \bibinfo{person}{Elaine Shi}, \bibinfo{person}{T.{-}H.~Hubert Chan},
  \bibinfo{person}{Christopher~W. Fletcher}, \bibinfo{person}{Ling Ren},
  \bibinfo{person}{Xiangyao Yu}, {and} \bibinfo{person}{Srinivas Devadas}.}
  \bibinfo{year}{2018}\natexlab{}.
\newblock \showarticletitle{Path {ORAM:} An Extremely Simple Oblivious {RAM}
  Protocol}.
\newblock \bibinfo{journal}{\emph{J. {ACM}}} \bibinfo{volume}{65},
  \bibinfo{number}{4} (\bibinfo{year}{2018}), \bibinfo{pages}{18:1--18:26}.
\newblock
\href{https://doi.org/10.1145/3177872}{doi:\nolinkurl{10.1145/3177872}}


\bibitem[Su et~al\mbox{.}(2026)]%
        {beyonddialoguetime}
\bibfield{author}{\bibinfo{person}{Miao Su}, \bibinfo{person}{Yucan Guo},
  \bibinfo{person}{Zhongni Hou}, \bibinfo{person}{Long Bai},
  \bibinfo{person}{Zixuan Li}, \bibinfo{person}{Yufei Zhang},
  \bibinfo{person}{Guojun Yin}, \bibinfo{person}{Wei Lin},
  \bibinfo{person}{Xiaolong Jin}, \bibinfo{person}{Jiafeng Guo}, {and}
  \bibinfo{person}{Xueqi Cheng}.} \bibinfo{year}{2026}\natexlab{}.
\newblock \showarticletitle{Beyond Dialogue Time: Temporal Semantic Memory for
  Personalized LLM Agents}.
\newblock \bibinfo{journal}{\emph{arXiv preprint arXiv:2601.07468}}
  (\bibinfo{year}{2026}).
\newblock


\bibitem[Subramanya et~al\mbox{.}(2019)]%
        {diskann}
\bibfield{author}{\bibinfo{person}{Suhas~Jayaram Subramanya},
  \bibinfo{person}{Devvrit}, \bibinfo{person}{Harsha~Vardhan Simhadri},
  \bibinfo{person}{Ravishankar Krishnaswamy}, {and} \bibinfo{person}{Rohan
  Kadekodi}.} \bibinfo{year}{2019}\natexlab{}.
\newblock \showarticletitle{{DiskANN}: Fast Accurate Billion-point Nearest
  Neighbor Search on a Single Node}. In \bibinfo{booktitle}{\emph{{NeurIPS}}}.
  \bibinfo{pages}{13748--13758}.
\newblock


\bibitem[{Sunday Robotics}(2026)]%
        {sunday}
\bibfield{author}{\bibinfo{person}{{Sunday Robotics}}.}
  \bibinfo{year}{2026}\natexlab{}.
\newblock \bibinfo{title}{Sunday Robotics}.
\newblock \bibinfo{howpublished}{\url{https://www.sunday.ai/}}.
\newblock


\bibitem[Tan et~al\mbox{.}(2025b)]%
        {personabench}
\bibfield{author}{\bibinfo{person}{Juntao Tan}, \bibinfo{person}{Liangwei
  Yang}, \bibinfo{person}{Zuxin Liu}, \bibinfo{person}{Zhiwei Liu},
  \bibinfo{person}{Rithesh~R. N.}, \bibinfo{person}{Tulika~Manoj Awalgaonkar},
  \bibinfo{person}{Jianguo Zhang}, \bibinfo{person}{Weiran Yao},
  \bibinfo{person}{Ming Zhu}, \bibinfo{person}{Shirley Kokane},
  \bibinfo{person}{Silvio Savarese}, \bibinfo{person}{Huan Wang},
  \bibinfo{person}{Caiming Xiong}, {and} \bibinfo{person}{Shelby Heinecke}.}
  \bibinfo{year}{2025}\natexlab{b}.
\newblock \showarticletitle{PersonaBench: Evaluating {AI} Models on
  Understanding Personal Information through Accessing (Synthetic) Private User
  Data}. In \bibinfo{booktitle}{\emph{{ACL} (Findings)}}
  \emph{(\bibinfo{series}{Findings of {ACL}})}. \bibinfo{publisher}{Association
  for Computational Linguistics}, \bibinfo{pages}{878--893}.
\newblock


\bibitem[Tan et~al\mbox{.}(2023)]%
        {timelineqa}
\bibfield{author}{\bibinfo{person}{Wang{-}Chiew Tan}, \bibinfo{person}{Jane
  Dwivedi{-}Yu}, \bibinfo{person}{Yuliang Li}, \bibinfo{person}{Lambert
  Mathias}, \bibinfo{person}{Marzieh Saeidi}, \bibinfo{person}{Jing~Nathan
  Yan}, {and} \bibinfo{person}{Alon~Y. Halevy}.}
  \bibinfo{year}{2023}\natexlab{}.
\newblock \showarticletitle{TimelineQA: {A} Benchmark for Question Answering
  over Timelines}. In \bibinfo{booktitle}{\emph{{ACL} (Findings)}}
  \emph{(\bibinfo{series}{Findings of {ACL}})}. \bibinfo{publisher}{Association
  for Computational Linguistics}, \bibinfo{pages}{77--91}.
\newblock


\bibitem[Tan et~al\mbox{.}(2026)]%
        {privgemo}
\bibfield{author}{\bibinfo{person}{Xingyu Tan}, \bibinfo{person}{Xiaoyang
  Wang}, \bibinfo{person}{Qing Liu}, \bibinfo{person}{Xiwei Xu},
  \bibinfo{person}{Xin Yuan}, \bibinfo{person}{Liming Zhu}, {and}
  \bibinfo{person}{Wenjie Zhang}.} \bibinfo{year}{2026}\natexlab{}.
\newblock \showarticletitle{PrivGemo: Privacy-Preserving Dual-Tower Graph
  Retrieval for Empowering LLM Reasoning with Memory Augmentation}.
\newblock \bibinfo{journal}{\emph{arXiv preprint arXiv:2601.08739}}
  (\bibinfo{year}{2026}).
\newblock


\bibitem[Tan et~al\mbox{.}(2025a)]%
        {rmm}
\bibfield{author}{\bibinfo{person}{Zhen Tan}, \bibinfo{person}{Jun Yan},
  \bibinfo{person}{I{-}Hung Hsu}, \bibinfo{person}{Rujun Han},
  \bibinfo{person}{Zifeng Wang}, \bibinfo{person}{Long~T. Le},
  \bibinfo{person}{Yiwen Song}, \bibinfo{person}{Yanfei Chen},
  \bibinfo{person}{Hamid Palangi}, \bibinfo{person}{George Lee},
  \bibinfo{person}{Anand~Rajan Iyer}, \bibinfo{person}{Tianlong Chen},
  \bibinfo{person}{Huan Liu}, \bibinfo{person}{Chen{-}Yu Lee}, {and}
  \bibinfo{person}{Tomas Pfister}.} \bibinfo{year}{2025}\natexlab{a}.
\newblock \showarticletitle{In Prospect and Retrospect: Reflective Memory
  Management for Long-term Personalized Dialogue Agents}. In
  \bibinfo{booktitle}{\emph{{ACL} {(1)}}}. \bibinfo{publisher}{Association for
  Computational Linguistics}, \bibinfo{pages}{8416--8439}.
\newblock


\bibitem[{TechCrunch}(2026)]%
        {metaraybannametag}
\bibfield{author}{\bibinfo{person}{{TechCrunch}}.}
  \bibinfo{year}{2026}\natexlab{}.
\newblock \bibinfo{title}{Meta plans to add facial recognition to its smart
  glasses, report claims}.
\newblock
  \bibinfo{howpublished}{\url{https://techcrunch.com/2026/02/13/meta-plans-to-add-facial-recognition-to-its-smart-glasses-report-claims/}}.
\newblock


\bibitem[{The Linux Kernel Documentation}(2026)]%
        {linuxkerneltdx}
\bibfield{author}{\bibinfo{person}{{The Linux Kernel Documentation}}.}
  \bibinfo{year}{2026}\natexlab{}.
\newblock \bibinfo{title}{Intel Trust Domain Extensions ({TDX})}.
\newblock
  \bibinfo{howpublished}{\url{https://www.kernel.org/doc/html/next/x86/tdx.html}}.
\newblock


\bibitem[{Tinfoil Inc.}(2026a)]%
        {tinfoilcvm}
\bibfield{author}{\bibinfo{person}{{Tinfoil Inc.}}}
  \bibinfo{year}{2026}\natexlab{a}.
\newblock \bibinfo{title}{Confidential {VM} Image}.
\newblock \bibinfo{howpublished}{\url{https://github.com/tinfoilsh/cvmimage}}.
\newblock


\bibitem[{Tinfoil Inc.}(2026b)]%
        {tinfoilconfidentiality}
\bibfield{author}{\bibinfo{person}{{Tinfoil Inc.}}}
  \bibinfo{year}{2026}\natexlab{b}.
\newblock \bibinfo{title}{A Primer on Secure Enclaves}.
\newblock
  \bibinfo{howpublished}{\url{https://docs.tinfoil.sh/verification/secure-enclave-primer}}.
\newblock


\bibitem[{Tinfoil Inc.}(2026c)]%
        {tinfoilmodelidentity}
\bibfield{author}{\bibinfo{person}{{Tinfoil Inc.}}}
  \bibinfo{year}{2026}\natexlab{c}.
\newblock \bibinfo{title}{Proving Model Identity}.
\newblock
  \bibinfo{howpublished}{\url{https://tinfoil.sh/blog/2026-02-03-proving-model-identity}}.
\newblock


\bibitem[{Tinfoil Inc.}(2026d)]%
        {tinfoilshim}
\bibfield{author}{\bibinfo{person}{{Tinfoil Inc.}}}
  \bibinfo{year}{2026}\natexlab{d}.
\newblock \bibinfo{title}{\texttt{tfshim}: Attestation Proxy}.
\newblock \bibinfo{howpublished}{\url{https://github.com/tinfoilsh/tfshim}}.
\newblock


\bibitem[{Tinfoil Inc.}(2026e)]%
        {tinfoilverification}
\bibfield{author}{\bibinfo{person}{{Tinfoil Inc.}}}
  \bibinfo{year}{2026}\natexlab{e}.
\newblock \bibinfo{title}{Tinfoil Attestation Verification}.
\newblock
  \bibinfo{howpublished}{\url{https://docs.tinfoil.sh/verification/how-to-verify}}.
\newblock


\bibitem[{Tinfoil Inc.}(2026f)]%
        {tinfoilcontainers}
\bibfield{author}{\bibinfo{person}{{Tinfoil Inc.}}}
  \bibinfo{year}{2026}\natexlab{f}.
\newblock \bibinfo{title}{Tinfoil Containers}.
\newblock
  \bibinfo{howpublished}{\url{https://docs.tinfoil.sh/containers/overview}}.
\newblock


\bibitem[{Tinfoil Inc.}(2026g)]%
        {tinfoilcontainerspricing}
\bibfield{author}{\bibinfo{person}{{Tinfoil Inc.}}}
  \bibinfo{year}{2026}\natexlab{g}.
\newblock \bibinfo{title}{Tinfoil Containers}.
\newblock \bibinfo{howpublished}{\url{https://tinfoil.sh/containers}}.
\newblock
\newblock
\shownote{Accessed 2026-10-27: \$0.06 per vCPU-hour and \$0.01 per GB-hour}.


\bibitem[{Tinfoil Inc.}(2026h)]%
        {tinfoilgptosspricing}
\bibfield{author}{\bibinfo{person}{{Tinfoil Inc.}}}
  \bibinfo{year}{2026}\natexlab{h}.
\newblock \bibinfo{title}{Tinfoil GPT-OSS 120B Model Page}.
\newblock \bibinfo{howpublished}{\url{https://tinfoil.sh/models/gpt-oss-120b}}.
\newblock
\newblock
\shownote{Accessed 2026-10-27: \$0.15 and \$0.60 per million input and output
  tokens}.


\bibitem[{Tinfoil Inc.}(2026i)]%
        {tinfoilnomicpricing}
\bibfield{author}{\bibinfo{person}{{Tinfoil Inc.}}}
  \bibinfo{year}{2026}\natexlab{i}.
\newblock \bibinfo{title}{Tinfoil Nomic Embed Text Model Page}.
\newblock
  \bibinfo{howpublished}{\url{https://tinfoil.sh/models/nomic-embed-text}}.
\newblock
\newblock
\shownote{Accessed 2026-10-27: \$0.05 per million input tokens}.


\bibitem[{Tinfoil Inc.}(2026j)]%
        {tinfoilsdk}
\bibfield{author}{\bibinfo{person}{{Tinfoil Inc.}}}
  \bibinfo{year}{2026}\natexlab{j}.
\newblock \bibinfo{title}{Tinfoil {SDK}}.
\newblock
  \bibinfo{howpublished}{\url{https://github.com/tinfoilsh/tinfoil-python}}.
\newblock


\bibitem[{Tinfoil Inc.}(2026k)]%
        {tinfoil}
\bibfield{author}{\bibinfo{person}{{Tinfoil Inc.}}}
  \bibinfo{year}{2026}\natexlab{k}.
\newblock \bibinfo{title}{Tinfoil: Verifiably Private {AI} Powered by Secure
  Enclaves}.
\newblock \bibinfo{howpublished}{\url{https://tinfoil.sh/}}.
\newblock


\bibitem[{Tinfoil Inc.}(2026l)]%
        {tinfoilsidechannels}
\bibfield{author}{\bibinfo{person}{{Tinfoil Inc.}}}
  \bibinfo{year}{2026}\natexlab{l}.
\newblock \bibinfo{title}{What About Side-Channels?}
\newblock
  \bibinfo{howpublished}{\url{https://tinfoil.sh/blog/2025-05-15-side-channels}}.
\newblock


\bibitem[Tinoco et~al\mbox{.}(2023)]%
        {enigmap}
\bibfield{author}{\bibinfo{person}{Afonso Tinoco}, \bibinfo{person}{Sixiang
  Gao}, {and} \bibinfo{person}{Elaine Shi}.} \bibinfo{year}{2023}\natexlab{}.
\newblock \showarticletitle{EnigMap: External-Memory Oblivious Map for Secure
  Enclaves}. In \bibinfo{booktitle}{\emph{{USENIX} Security Symposium}}.
  \bibinfo{publisher}{{USENIX} Association}, \bibinfo{pages}{4033--4050}.
\newblock


\bibitem[Tran et~al\mbox{.}(2025)]%
        {openlifelogqa}
\bibfield{author}{\bibinfo{person}{Quang-Linh Tran}, \bibinfo{person}{Binh
  Nguyen}, \bibinfo{person}{Gareth J.~F. Jones}, {and} \bibinfo{person}{Cathal
  Gurrin}.} \bibinfo{year}{2025}\natexlab{}.
\newblock \showarticletitle{OpenLifelogQA: An Open-Ended Multi-Modal Lifelog
  Question-Answering Dataset}.
\newblock \bibinfo{journal}{\emph{arXiv preprint arXiv:2508.03583}}
  (\bibinfo{year}{2025}).
\newblock


\bibitem[Trivedi et~al\mbox{.}(2024)]%
        {appworld}
\bibfield{author}{\bibinfo{person}{Harsh Trivedi}, \bibinfo{person}{Tushar
  Khot}, \bibinfo{person}{Mareike Hartmann}, \bibinfo{person}{Ruskin Manku},
  \bibinfo{person}{Vinty Dong}, \bibinfo{person}{Edward Li},
  \bibinfo{person}{Shashank Gupta}, \bibinfo{person}{Ashish Sabharwal}, {and}
  \bibinfo{person}{Niranjan Balasubramanian}.} \bibinfo{year}{2024}\natexlab{}.
\newblock \showarticletitle{AppWorld: {A} Controllable World of Apps and People
  for Benchmarking Interactive Coding Agents}. In
  \bibinfo{booktitle}{\emph{{ACL} {(1)}}}. \bibinfo{publisher}{Association for
  Computational Linguistics}, \bibinfo{pages}{16022--16076}.
\newblock


\bibitem[van Schaik et~al\mbox{.}(2021)]%
        {cacheout}
\bibfield{author}{\bibinfo{person}{Stephan van Schaik}, \bibinfo{person}{Marina
  Minkin}, \bibinfo{person}{Andrew Kwong}, \bibinfo{person}{Daniel Genkin},
  {and} \bibinfo{person}{Yuval Yarom}.} \bibinfo{year}{2021}\natexlab{}.
\newblock \showarticletitle{CacheOut: Leaking Data on Intel CPUs via Cache
  Evictions}. In \bibinfo{booktitle}{\emph{{SP}}}. \bibinfo{publisher}{{IEEE}},
  \bibinfo{pages}{339--354}.
\newblock


\bibitem[van Schaik et~al\mbox{.}(2024)]%
        {sgxfail}
\bibfield{author}{\bibinfo{person}{Stephan van Schaik},
  \bibinfo{person}{Alexander Seto}, \bibinfo{person}{Thomas Yurek},
  \bibinfo{person}{Adam Batori}, \bibinfo{person}{Bader AlBassam},
  \bibinfo{person}{Daniel Genkin}, \bibinfo{person}{Andrew Miller},
  \bibinfo{person}{Eyal Ronen}, \bibinfo{person}{Yuval Yarom}, {and}
  \bibinfo{person}{Christina Garman}.} \bibinfo{year}{2024}\natexlab{}.
\newblock \showarticletitle{SoK: SGX.Fail: How Stuff Gets eXposed}. In
  \bibinfo{booktitle}{\emph{{SP}}}. \bibinfo{publisher}{{IEEE}},
  \bibinfo{pages}{4143--4162}.
\newblock


\bibitem[Vanoverloop et~al\mbox{.}(2025)]%
        {tlblur}
\bibfield{author}{\bibinfo{person}{Daan Vanoverloop},
  \bibinfo{person}{Andr{\'{e}}s S{\'{a}}nchez}, \bibinfo{person}{Flavio
  Toffalini}, \bibinfo{person}{Frank Piessens}, \bibinfo{person}{Mathias
  Payer}, {and} \bibinfo{person}{Jo~Van Bulck}.}
  \bibinfo{year}{2025}\natexlab{}.
\newblock \showarticletitle{TLBlur: Compiler-Assisted Automated Hardening
  against Controlled Channels on Off-the-Shelf Intel {SGX} Platforms}. In
  \bibinfo{booktitle}{\emph{{USENIX} Security Symposium}}.
  \bibinfo{publisher}{{USENIX} Association}, \bibinfo{pages}{1167--1186}.
\newblock


\bibitem[Wang et~al\mbox{.}(2022)]%
        {groupchat}
\bibfield{author}{\bibinfo{person}{Dakuo Wang}, \bibinfo{person}{Haoyu Wang},
  \bibinfo{person}{Mo Yu}, \bibinfo{person}{Zahra Ashktorab}, {and}
  \bibinfo{person}{Ming Tan}.} \bibinfo{year}{2022}\natexlab{}.
\newblock \showarticletitle{Group Chat Ecology in Enterprise Instant Messaging:
  How Employees Collaborate Through Multi-User Chat Channels on Slack}.
\newblock \bibinfo{journal}{\emph{Proc. {ACM} Hum. Comput. Interact.}}
  \bibinfo{volume}{6}, \bibinfo{number}{{CSCW1}} (\bibinfo{year}{2022}),
  \bibinfo{pages}{94:1--94:14}.
\newblock


\bibitem[Wang et~al\mbox{.}(2023b)]%
        {voyager}
\bibfield{author}{\bibinfo{person}{Guanzhi Wang}, \bibinfo{person}{Yuqi Xie},
  \bibinfo{person}{Yunfan Jiang}, \bibinfo{person}{Ajay Mandlekar},
  \bibinfo{person}{Chaowei Xiao}, \bibinfo{person}{Yuke Zhu},
  \bibinfo{person}{Linxi Fan}, {and} \bibinfo{person}{Anima Anandkumar}.}
  \bibinfo{year}{2023}\natexlab{b}.
\newblock \showarticletitle{Voyager: An Open-Ended Embodied Agent with Large
  Language Models}.
\newblock \bibinfo{journal}{\emph{arXiv preprint arXiv:2305.16291}}
  (\bibinfo{year}{2023}).
\newblock


\bibitem[Wang et~al\mbox{.}(2021)]%
        {milvus}
\bibfield{author}{\bibinfo{person}{Jianguo Wang}, \bibinfo{person}{Xiaomeng
  Yi}, \bibinfo{person}{Rentong Guo}, \bibinfo{person}{Hai Jin},
  \bibinfo{person}{Peng Xu}, \bibinfo{person}{Shengjun Li},
  \bibinfo{person}{Xiangyu Wang}, \bibinfo{person}{Xiangzhou Guo},
  \bibinfo{person}{Chengming Li}, \bibinfo{person}{Xiaohai Xu},
  \bibinfo{person}{Kun Yu}, \bibinfo{person}{Yuxing Yuan},
  \bibinfo{person}{Yinghao Zou}, \bibinfo{person}{Jiquan Long},
  \bibinfo{person}{Yudong Cai}, \bibinfo{person}{Zhenxiang Li},
  \bibinfo{person}{Zhifeng Zhang}, \bibinfo{person}{Yihua Mo},
  \bibinfo{person}{Jun Gu}, \bibinfo{person}{Ruiyi Jiang}, \bibinfo{person}{Yi
  Wei}, {and} \bibinfo{person}{Charles Xie}.} \bibinfo{year}{2021}\natexlab{}.
\newblock \showarticletitle{Milvus: {A} Purpose-Built Vector Data Management
  System}. In \bibinfo{booktitle}{\emph{{SIGMOD} Conference}}.
  \bibinfo{publisher}{{ACM}}, \bibinfo{pages}{2614--2627}.
\newblock


\bibitem[Wang et~al\mbox{.}(2023a)]%
        {nhq}
\bibfield{author}{\bibinfo{person}{Mengzhao Wang}, \bibinfo{person}{Lingwei
  Lv}, \bibinfo{person}{Xiaoliang Xu}, \bibinfo{person}{Yuxiang Wang},
  \bibinfo{person}{Qiang Yue}, {and} \bibinfo{person}{Jiongkang Ni}.}
  \bibinfo{year}{2023}\natexlab{a}.
\newblock \showarticletitle{An Efficient and Robust Framework for Approximate
  Nearest Neighbor Search with Attribute Constraint}. In
  \bibinfo{booktitle}{\emph{NeurIPS}}.
\newblock


\bibitem[Wang et~al\mbox{.}(2026)]%
        {personatrace}
\bibfield{author}{\bibinfo{person}{Minjia Wang}, \bibinfo{person}{Yunfeng
  Wang}, \bibinfo{person}{Xiao Ma}, \bibinfo{person}{Dexin Lv},
  \bibinfo{person}{Qifan Guo}, \bibinfo{person}{Lynn Zheng},
  \bibinfo{person}{Benliang Wang}, \bibinfo{person}{Lei Wang},
  \bibinfo{person}{Jiannan Li}, \bibinfo{person}{Yongwei Xing},
  \bibinfo{person}{David Xu}, {and} \bibinfo{person}{Zheng Sun}.}
  \bibinfo{year}{2026}\natexlab{}.
\newblock \showarticletitle{PersonaTrace: Synthesizing Realistic Digital
  Footprints with LLM Agents}.
\newblock \bibinfo{journal}{\emph{arXiv preprint arXiv:2603.11955}}
  (\bibinfo{year}{2026}).
\newblock


\bibitem[Wang et~al\mbox{.}(2025a)]%
        {odysseybench}
\bibfield{author}{\bibinfo{person}{Weixuan Wang}, \bibinfo{person}{Dongge Han},
  \bibinfo{person}{Daniel~Madrigal Diaz}, \bibinfo{person}{Jin Xu},
  \bibinfo{person}{Victor Rühle}, {and} \bibinfo{person}{Saravan Rajmohan}.}
  \bibinfo{year}{2025}\natexlab{a}.
\newblock \showarticletitle{OdysseyBench: Evaluating LLM Agents on Long-Horizon
  Complex Office Application Workflows}.
\newblock \bibinfo{journal}{\emph{arXiv preprint arXiv:2508.09124}}
  (\bibinfo{year}{2025}).
\newblock


\bibitem[Wang et~al\mbox{.}(2024b)]%
        {ragbestpractices}
\bibfield{author}{\bibinfo{person}{Xiaohua Wang}, \bibinfo{person}{Zhenghua
  Wang}, \bibinfo{person}{Xuan Gao}, \bibinfo{person}{Feiran Zhang},
  \bibinfo{person}{Yixin Wu}, \bibinfo{person}{Zhibo Xu},
  \bibinfo{person}{Tianyuan Shi}, \bibinfo{person}{Zhengyuan Wang},
  \bibinfo{person}{Shizheng Li}, \bibinfo{person}{Qi Qian},
  \bibinfo{person}{Ruicheng Yin}, \bibinfo{person}{Changze Lv},
  \bibinfo{person}{Xiaoqing Zheng}, {and} \bibinfo{person}{Xuanjing Huang}.}
  \bibinfo{year}{2024}\natexlab{b}.
\newblock \showarticletitle{Searching for Best Practices in Retrieval-Augmented
  Generation}. In \bibinfo{booktitle}{\emph{{EMNLP}}}.
  \bibinfo{publisher}{Association for Computational Linguistics},
  \bibinfo{pages}{17716--17736}.
\newblock


\bibitem[Wang et~al\mbox{.}(2024a)]%
        {editablememorygraph}
\bibfield{author}{\bibinfo{person}{Zheng Wang}, \bibinfo{person}{Zhongyang Li},
  \bibinfo{person}{Zeren Jiang}, \bibinfo{person}{Dandan Tu}, {and}
  \bibinfo{person}{Wei Shi}.} \bibinfo{year}{2024}\natexlab{a}.
\newblock \showarticletitle{Crafting Personalized Agents through
  Retrieval-Augmented Generation on Editable Memory Graphs}. In
  \bibinfo{booktitle}{\emph{{EMNLP}}}. \bibinfo{publisher}{Association for
  Computational Linguistics}, \bibinfo{pages}{4891--4906}.
\newblock


\bibitem[Wang et~al\mbox{.}(2025c)]%
        {tetd}
\bibfield{author}{\bibinfo{person}{Zhanbo Wang}, \bibinfo{person}{Jiaxin Zhan},
  \bibinfo{person}{Xuhua Ding}, \bibinfo{person}{Fengwei Zhang}, {and}
  \bibinfo{person}{Ning Hu}.} \bibinfo{year}{2025}\natexlab{c}.
\newblock \showarticletitle{{TETD:} Trusted Execution in Trust Domains}. In
  \bibinfo{booktitle}{\emph{{USENIX} Security Symposium}}.
  \bibinfo{publisher}{{USENIX} Association}, \bibinfo{pages}{1187--1206}.
\newblock


\bibitem[Wang et~al\mbox{.}(2025b)]%
        {agentworkflowmemory}
\bibfield{author}{\bibinfo{person}{Zora~Zhiruo Wang}, \bibinfo{person}{Jiayuan
  Mao}, \bibinfo{person}{Daniel Fried}, {and} \bibinfo{person}{Graham Neubig}.}
  \bibinfo{year}{2025}\natexlab{b}.
\newblock \showarticletitle{Agent Workflow Memory}. In
  \bibinfo{booktitle}{\emph{{ICML}}} \emph{(\bibinfo{series}{Proceedings of
  Machine Learning Research})}. \bibinfo{publisher}{{PMLR} / OpenReview.net}.
\newblock


\bibitem[{Weaviate}(2026)]%
        {weaviate}
\bibfield{author}{\bibinfo{person}{{Weaviate}}.}
  \bibinfo{year}{2026}\natexlab{}.
\newblock \bibinfo{title}{Weaviate}.
\newblock \bibinfo{howpublished}{\url{https://weaviate.io/}}.
\newblock


\bibitem[Wei et~al\mbox{.}(2020)]%
        {analyticdbv}
\bibfield{author}{\bibinfo{person}{Chuangxian Wei}, \bibinfo{person}{Bin Wu},
  \bibinfo{person}{Sheng Wang}, \bibinfo{person}{Renjie Lou},
  \bibinfo{person}{Chaoqun Zhan}, \bibinfo{person}{Feifei Li}, {and}
  \bibinfo{person}{Yuanzhe Cai}.} \bibinfo{year}{2020}\natexlab{}.
\newblock \showarticletitle{AnalyticDB-V: {A} Hybrid Analytical Engine Towards
  Query Fusion for Structured and Unstructured Data}.
\newblock \bibinfo{journal}{\emph{Proc. {VLDB} Endow.}} \bibinfo{volume}{13},
  \bibinfo{number}{12} (\bibinfo{year}{2020}), \bibinfo{pages}{3152--3165}.
\newblock


\bibitem[Weissteiner et~al\mbox{.}(2025)]%
        {teecorrelate}
\bibfield{author}{\bibinfo{person}{Hannes Weissteiner}, \bibinfo{person}{Fabian
  Rauscher}, \bibinfo{person}{Robin~Leander Schr{\"{o}}der},
  \bibinfo{person}{Jonas Juffinger}, \bibinfo{person}{Stefan Gast},
  \bibinfo{person}{Jan Wichelmann}, \bibinfo{person}{Thomas Eisenbarth}, {and}
  \bibinfo{person}{Daniel Gruss}.} \bibinfo{year}{2025}\natexlab{}.
\newblock \showarticletitle{TEEcorrelate: An Information-Preserving Defense
  against Performance-Counter Attacks on TEEs}. In
  \bibinfo{booktitle}{\emph{{USENIX} Security Symposium}}.
  \bibinfo{publisher}{{USENIX} Association}, \bibinfo{pages}{2481--2498}.
\newblock


\bibitem[{WhatsApp}(2026)]%
        {whatsapp}
\bibfield{author}{\bibinfo{person}{{WhatsApp}}.}
  \bibinfo{year}{2026}\natexlab{}.
\newblock \bibinfo{title}{{WhatsApp}}.
\newblock \bibinfo{howpublished}{\url{https://www.whatsapp.com/}}.
\newblock


\bibitem[Wheeler and Nezlek(1977)]%
        {Wheeler1977SexDI}
\bibfield{author}{\bibinfo{person}{Ladd Wheeler} {and} \bibinfo{person}{John
  Nezlek}.} \bibinfo{year}{1977}\natexlab{}.
\newblock \showarticletitle{Sex differences in social participation.}
\newblock \bibinfo{journal}{\emph{Journal of Personality and Social
  Psychology}} \bibinfo{volume}{35}, \bibinfo{number}{10}
  (\bibinfo{year}{1977}), \bibinfo{pages}{742--754}.
\newblock
\href{https://doi.org/10.1037/0022-3514.35.10.742}{doi:\nolinkurl{10.1037/0022-3514.35.10.742}}


\bibitem[Whittaker et~al\mbox{.}(2011)]%
        {emailrefinding}
\bibfield{author}{\bibinfo{person}{Steve Whittaker}, \bibinfo{person}{Tara
  Matthews}, \bibinfo{person}{Julian~A. Cerruti}, \bibinfo{person}{Hernan
  Badenes}, {and} \bibinfo{person}{John~C. Tang}.}
  \bibinfo{year}{2011}\natexlab{}.
\newblock \showarticletitle{Am {I} wasting my time organizing email?: a study
  of email refinding}. In \bibinfo{booktitle}{\emph{{CHI}}}.
  \bibinfo{publisher}{{ACM}}, \bibinfo{pages}{3449--3458}.
\newblock


\bibitem[Wilke et~al\mbox{.}(2024a)]%
        {tdxdown}
\bibfield{author}{\bibinfo{person}{Luca Wilke}, \bibinfo{person}{Florian
  Sieck}, {and} \bibinfo{person}{Thomas Eisenbarth}.}
  \bibinfo{year}{2024}\natexlab{a}.
\newblock \showarticletitle{TDXdown: Single-Stepping and Instruction Counting
  Attacks against Intel {TDX}}. In \bibinfo{booktitle}{\emph{{CCS}}}.
  \bibinfo{publisher}{{ACM}}, \bibinfo{pages}{79--93}.
\newblock


\bibitem[Wilke et~al\mbox{.}(2024b)]%
        {sevstep}
\bibfield{author}{\bibinfo{person}{Luca Wilke}, \bibinfo{person}{Jan
  Wichelmann}, \bibinfo{person}{Anja Rabich}, {and} \bibinfo{person}{Thomas
  Eisenbarth}.} \bibinfo{year}{2024}\natexlab{b}.
\newblock \showarticletitle{{SEV-Step A} Single-Stepping Framework for
  {AMD-SEV}}.
\newblock \bibinfo{journal}{\emph{{IACR} Trans. Cryptogr. Hardw. Embed. Syst.}}
  \bibinfo{volume}{2024}, \bibinfo{number}{1} (\bibinfo{year}{2024}),
  \bibinfo{pages}{180--206}.
\newblock


\bibitem[Wolinsky et~al\mbox{.}(2013)]%
        {hangwithyourbuddies}
\bibfield{author}{\bibinfo{person}{David~Isaac Wolinsky}, \bibinfo{person}{Ewa
  Syta}, {and} \bibinfo{person}{Bryan Ford}.} \bibinfo{year}{2013}\natexlab{}.
\newblock \showarticletitle{Hang with your buddies to resist intersection
  attacks}. In \bibinfo{booktitle}{\emph{{CCS}}}. \bibinfo{publisher}{{ACM}},
  \bibinfo{pages}{1153--1166}.
\newblock


\bibitem[Woodcock(2026)]%
        {celiveo2026}
\bibfield{author}{\bibinfo{person}{Mary Woodcock}.}
  \bibinfo{year}{2026}\natexlab{}.
\newblock \bibinfo{title}{Your Documents Unlock AI Corporate Memory Power}.
\newblock
  \bibinfo{howpublished}{\url{https://www.celiveo.com/blog/ai-document-management-the-important-employee-documents-to-capture-for-competitive-advantage-in-2026/}}.
\newblock


\bibitem[Woodward(2026)]%
        {geminipersonalintelligence}
\bibfield{author}{\bibinfo{person}{Josh Woodward}.}
  \bibinfo{year}{2026}\natexlab{}.
\newblock \bibinfo{title}{Gemini introduces Personal Intelligence}.
\newblock
  \bibinfo{howpublished}{\url{https://blog.google/innovation-and-ai/products/gemini-app/personal-intelligence/}}.
\newblock


\bibitem[Wu et~al\mbox{.}(2025a)]%
        {longmemeval}
\bibfield{author}{\bibinfo{person}{Di Wu}, \bibinfo{person}{Hongwei Wang},
  \bibinfo{person}{Wenhao Yu}, \bibinfo{person}{Yuwei Zhang},
  \bibinfo{person}{Kai{-}Wei Chang}, {and} \bibinfo{person}{Dong Yu}.}
  \bibinfo{year}{2025}\natexlab{a}.
\newblock \showarticletitle{LongMemEval: Benchmarking Chat Assistants on
  Long-Term Interactive Memory}. In \bibinfo{booktitle}{\emph{{ICLR}}}.
  \bibinfo{publisher}{OpenReview.net}.
\newblock


\bibitem[Wu et~al\mbox{.}(2022)]%
        {hqann}
\bibfield{author}{\bibinfo{person}{Wei Wu}, \bibinfo{person}{Junlin He},
  \bibinfo{person}{Yu Qiao}, \bibinfo{person}{Guoheng Fu}, \bibinfo{person}{Li
  Liu}, {and} \bibinfo{person}{Jin Yu}.} \bibinfo{year}{2022}\natexlab{}.
\newblock \showarticletitle{{HQANN:} Efficient and Robust Similarity Search for
  Hybrid Queries with Structured and Unstructured Constraints}. In
  \bibinfo{booktitle}{\emph{{CIKM}}}. \bibinfo{publisher}{{ACM}},
  \bibinfo{pages}{4580--4584}.
\newblock


\bibitem[Wu et~al\mbox{.}(2025b)]%
        {sgmem}
\bibfield{author}{\bibinfo{person}{Yaxiong Wu}, \bibinfo{person}{Yongyue
  Zhang}, \bibinfo{person}{Sheng Liang}, {and} \bibinfo{person}{Yong Liu}.}
  \bibinfo{year}{2025}\natexlab{b}.
\newblock \showarticletitle{SGMem: Sentence Graph Memory for Long-Term
  Conversational Agents}.
\newblock \bibinfo{journal}{\emph{arXiv preprint arXiv:2509.21212}}
  (\bibinfo{year}{2025}).
\newblock


\bibitem[Wunder et~al\mbox{.}(2025)]%
        {datalossrecovery}
\bibfield{author}{\bibinfo{person}{Julia Wunder}, \bibinfo{person}{Rick Wash},
  \bibinfo{person}{Karen Renaud}, \bibinfo{person}{Daniela~A Oliveira}, {and}
  \bibinfo{person}{Zinaida Benenson}.} \bibinfo{year}{2025}\natexlab{}.
\newblock \showarticletitle{Achieving Resilience: Data Loss and Recovery on
  Devices for Personal Use in Three Countries}. In
  \bibinfo{booktitle}{\emph{{CHI}}}. \bibinfo{publisher}{{ACM}},
  \bibinfo{pages}{830:1--830:26}.
\newblock


\bibitem[Xiu et~al\mbox{.}(2026)]%
        {astrabench}
\bibfield{author}{\bibinfo{person}{Zidi Xiu}, \bibinfo{person}{David~Q. Sun},
  \bibinfo{person}{Kevin Cheng}, \bibinfo{person}{Maitrik Patel},
  \bibinfo{person}{Josh Date}, \bibinfo{person}{Yizhe Zhang},
  \bibinfo{person}{Jiarui Lu}, \bibinfo{person}{Omar Attia},
  \bibinfo{person}{Raviteja Vemulapalli}, \bibinfo{person}{Oncel Tuzel},
  \bibinfo{person}{Meng Cao}, {and} \bibinfo{person}{Samy Bengio}.}
  \bibinfo{year}{2026}\natexlab{}.
\newblock \showarticletitle{ASTRA-bench: Evaluating Tool-Use Agent Reasoning
  and Action Planning with Personal User Context}.
\newblock \bibinfo{journal}{\emph{arXiv preprint arXiv:2603.01357}}
  (\bibinfo{year}{2026}).
\newblock


\bibitem[Xu et~al\mbox{.}(2022)]%
        {beyondgoldfish}
\bibfield{author}{\bibinfo{person}{Jing Xu}, \bibinfo{person}{Arthur Szlam},
  {and} \bibinfo{person}{Jason Weston}.} \bibinfo{year}{2022}\natexlab{}.
\newblock \showarticletitle{Beyond Goldfish Memory: Long-Term Open-Domain
  Conversation}. In \bibinfo{booktitle}{\emph{{ACL} {(1)}}}.
  \bibinfo{publisher}{Association for Computational Linguistics},
  \bibinfo{pages}{5180--5197}.
\newblock


\bibitem[Xu et~al\mbox{.}(2025)]%
        {amem}
\bibfield{author}{\bibinfo{person}{Wujiang Xu}, \bibinfo{person}{Zujie Liang},
  \bibinfo{person}{Kai Mei}, \bibinfo{person}{Hang Gao},
  \bibinfo{person}{Juntao Tan}, {and} \bibinfo{person}{Yongfeng Zhang}.}
  \bibinfo{year}{2025}\natexlab{}.
\newblock \showarticletitle{A-MEM: Agentic Memory for LLM Agents}.
\newblock \bibinfo{journal}{\emph{arXiv preprint arXiv:2502.12110}}
  (\bibinfo{year}{2025}).
\newblock


\bibitem[Xu et~al\mbox{.}(2020a)]%
        {xu2020}
\bibfield{author}{\bibinfo{person}{Xuhai Xu}, \bibinfo{person}{Ahmed~Hassan
  Awadallah}, \bibinfo{person}{Susan~T. Dumais}, \bibinfo{person}{Farheen
  Omar}, \bibinfo{person}{Bogdan Popp}, \bibinfo{person}{Robert Rounthwaite},
  {and} \bibinfo{person}{Farnaz Jahanbakhsh}.}
  \bibinfo{year}{2020}\natexlab{a}.
\newblock \showarticletitle{Understanding User Behavior For Document
  Recommendation}. In \bibinfo{booktitle}{\emph{{WWW}}}.
  \bibinfo{publisher}{{ACM} / {IW3C2}}, \bibinfo{pages}{3012--3018}.
\newblock


\bibitem[Xu et~al\mbox{.}(2020b)]%
        {mansw}
\bibfield{author}{\bibinfo{person}{Xiaoliang Xu}, \bibinfo{person}{Chang Li},
  \bibinfo{person}{Yuxiang Wang}, {and} \bibinfo{person}{Yixing Xia}.}
  \bibinfo{year}{2020}\natexlab{b}.
\newblock \showarticletitle{Multiattribute approximate nearest neighbor search
  based on navigable small world graph}.
\newblock \bibinfo{journal}{\emph{Concurr. Comput. Pract. Exp.}}
  \bibinfo{volume}{32}, \bibinfo{number}{24} (\bibinfo{year}{2020}).
\newblock


\bibitem[Xu et~al\mbox{.}(2015)]%
        {controlledchannelattacks}
\bibfield{author}{\bibinfo{person}{Yuanzhong Xu}, \bibinfo{person}{Weidong
  Cui}, {and} \bibinfo{person}{Marcus Peinado}.}
  \bibinfo{year}{2015}\natexlab{}.
\newblock \showarticletitle{Controlled-Channel Attacks: Deterministic Side
  Channels for Untrusted Operating Systems}. In
  \bibinfo{booktitle}{\emph{{IEEE} Symposium on Security and Privacy}}.
  \bibinfo{publisher}{{IEEE} Computer Society}, \bibinfo{pages}{640--656}.
\newblock


\bibitem[Xu et~al\mbox{.}(2024)]%
        {irangegraph}
\bibfield{author}{\bibinfo{person}{Yuexuan Xu}, \bibinfo{person}{Jianyang Gao},
  \bibinfo{person}{Yutong Gou}, \bibinfo{person}{Cheng Long}, {and}
  \bibinfo{person}{Christian~S. Jensen}.} \bibinfo{year}{2024}\natexlab{}.
\newblock \showarticletitle{iRangeGraph: Improvising Range-dedicated Graphs for
  Range-filtering Nearest Neighbor Search}.
\newblock \bibinfo{journal}{\emph{Proc. {ACM} Manag. Data}}
  \bibinfo{volume}{2}, \bibinfo{number}{6} (\bibinfo{year}{2024}),
  \bibinfo{pages}{239:1--239:26}.
\newblock


\bibitem[Xu et~al\mbox{.}(2023)]%
        {spfresh}
\bibfield{author}{\bibinfo{person}{Yuming Xu}, \bibinfo{person}{Hengyu Liang},
  \bibinfo{person}{Jin Li}, \bibinfo{person}{Shuotao Xu}, \bibinfo{person}{Qi
  Chen}, \bibinfo{person}{Qianxi Zhang}, \bibinfo{person}{Cheng Li},
  \bibinfo{person}{Ziyue Yang}, \bibinfo{person}{Fan Yang},
  \bibinfo{person}{Yuqing Yang}, \bibinfo{person}{Peng Cheng}, {and}
  \bibinfo{person}{Mao Yang}.} \bibinfo{year}{2023}\natexlab{}.
\newblock \showarticletitle{SPFresh: Incremental In-Place Update for
  Billion-Scale Vector Search}. In \bibinfo{booktitle}{\emph{{SOSP}}}.
  \bibinfo{publisher}{{ACM}}, \bibinfo{pages}{545--561}.
\newblock


\bibitem[Yagnik(2025)]%
        {googletrustedexecution}
\bibfield{author}{\bibinfo{person}{Jay Yagnik}.}
  \bibinfo{year}{2025}\natexlab{}.
\newblock \bibinfo{title}{Private {AI} {C}ompute: our next step in building
  private and helpful {AI}}.
\newblock
  \bibinfo{howpublished}{\url{https://blog.google/innovation-and-ai/products/google-private-ai-compute/}}.
\newblock


\bibitem[Yan et~al\mbox{.}(2025)]%
        {relocatevote}
\bibfield{author}{\bibinfo{person}{Yuqin Yan}, \bibinfo{person}{Wei Huang},
  \bibinfo{person}{Ilya Grishchenko}, \bibinfo{person}{Gururaj Saileshwar},
  \bibinfo{person}{Aastha Mehta}, {and} \bibinfo{person}{David Lie}.}
  \bibinfo{year}{2025}\natexlab{}.
\newblock \showarticletitle{Relocate-Vote: Using Sparsity Information to
  Exploit Ciphertext Side-Channels}. In \bibinfo{booktitle}{\emph{{USENIX}
  Security Symposium}}. \bibinfo{publisher}{{USENIX} Association},
  \bibinfo{pages}{5699--5717}.
\newblock


\bibitem[Yan et~al\mbox{.}(2026)]%
        {ShareChat}
\bibfield{author}{\bibinfo{person}{Yueru Yan}, \bibinfo{person}{Tuc Nguyen},
  \bibinfo{person}{Bo Su}, \bibinfo{person}{Melissa Lieffers}, {and}
  \bibinfo{person}{Thai Le}.} \bibinfo{year}{2026}\natexlab{}.
\newblock \showarticletitle{ShareChat: A Dataset of Chatbot Conversations in
  the Wild}.
\newblock \bibinfo{journal}{\emph{arXiv preprint arXiv:2512.17843}}
  (\bibinfo{year}{2026}).
\newblock


\bibitem[Yang et~al\mbox{.}(2020)]%
        {pase}
\bibfield{author}{\bibinfo{person}{Wen Yang}, \bibinfo{person}{Tao Li},
  \bibinfo{person}{Gai Fang}, {and} \bibinfo{person}{Hong Wei}.}
  \bibinfo{year}{2020}\natexlab{}.
\newblock \showarticletitle{{PASE:} PostgreSQL Ultra-High-Dimensional
  Approximate Nearest Neighbor Search Extension}. In
  \bibinfo{booktitle}{\emph{{SIGMOD} Conference}}. \bibinfo{publisher}{{ACM}},
  \bibinfo{pages}{2241--2253}.
\newblock


\bibitem[Yuan et~al\mbox{.}(2022)]%
        {multideviceusage}
\bibfield{author}{\bibinfo{person}{Ye Yuan}, \bibinfo{person}{Nathalie Riche},
  \bibinfo{person}{Nicolai Marquardt}, \bibinfo{person}{Molly~Jane Nicholas},
  \bibinfo{person}{Teddy Seyed}, \bibinfo{person}{Hugo Romat},
  \bibinfo{person}{Bongshin Lee}, \bibinfo{person}{Michel Pahud},
  \bibinfo{person}{Jonathan Goldstein}, \bibinfo{person}{Rojin Vishkaie},
  \bibinfo{person}{Christian Holz}, {and} \bibinfo{person}{Ken Hinckley}.}
  \bibinfo{year}{2022}\natexlab{}.
\newblock \showarticletitle{Understanding Multi-Device Usage Patterns: Physical
  Device Configurations and Fragmented Workflows}. In
  \bibinfo{booktitle}{\emph{{CHI}}}. \bibinfo{publisher}{{ACM}},
  \bibinfo{pages}{64:1--64:22}.
\newblock


\bibitem[Zhang et~al\mbox{.}(2023)]%
        {vbase}
\bibfield{author}{\bibinfo{person}{Qianxi Zhang}, \bibinfo{person}{Shuotao Xu},
  \bibinfo{person}{Qi Chen}, \bibinfo{person}{Guoxin Sui},
  \bibinfo{person}{Jiadong Xie}, \bibinfo{person}{Zhizhen Cai},
  \bibinfo{person}{Yaoqi Chen}, \bibinfo{person}{Yinxuan He},
  \bibinfo{person}{Yuqing Yang}, \bibinfo{person}{Fan Yang},
  \bibinfo{person}{Mao Yang}, {and} \bibinfo{person}{Lidong Zhou}.}
  \bibinfo{year}{2023}\natexlab{}.
\newblock \showarticletitle{{VBASE:} Unifying Online Vector Similarity Search
  and Relational Queries via Relaxed Monotonicity}. In
  \bibinfo{booktitle}{\emph{{OSDI}}}. \bibinfo{publisher}{{USENIX}
  Association}, \bibinfo{pages}{377--395}.
\newblock


\bibitem[Zhang et~al\mbox{.}(2019)]%
        {mobileflashwear}
\bibfield{author}{\bibinfo{person}{Tao Zhang}, \bibinfo{person}{Aviad Zuck},
  \bibinfo{person}{Donald~E. Porter}, {and} \bibinfo{person}{Dan Tsafrir}.}
  \bibinfo{year}{2019}\natexlab{}.
\newblock \showarticletitle{Apps Can Quickly Destroy Your Mobile's Flash: Why
  They Don't, and How to Keep It That Way}. In
  \bibinfo{booktitle}{\emph{MobiSys}}. \bibinfo{publisher}{{ACM}},
  \bibinfo{pages}{207--221}.
\newblock


\bibitem[Zhang et~al\mbox{.}(2024a)]%
        {invalidatecompare}
\bibfield{author}{\bibinfo{person}{Zhenkai Zhang}, \bibinfo{person}{Kunbei
  Cai}, \bibinfo{person}{Yanan Guo}, \bibinfo{person}{Fan Yao}, {and}
  \bibinfo{person}{Xing Gao}.} \bibinfo{year}{2024}\natexlab{a}.
\newblock \showarticletitle{Invalidate+Compare: {A} Timer-Free {GPU} Cache
  Attack Primitive}. In \bibinfo{booktitle}{\emph{{USENIX} Security
  Symposium}}. \bibinfo{publisher}{{USENIX} Association}.
\newblock


\bibitem[Zhang et~al\mbox{.}(2024b)]%
        {memsim}
\bibfield{author}{\bibinfo{person}{Zeyu Zhang}, \bibinfo{person}{Quanyu Dai},
  \bibinfo{person}{Luyu Chen}, \bibinfo{person}{Zeren Jiang},
  \bibinfo{person}{Rui Li}, \bibinfo{person}{Jieming Zhu}, \bibinfo{person}{Xu
  Chen}, \bibinfo{person}{Yi Xie}, \bibinfo{person}{Zhenhua Dong}, {and}
  \bibinfo{person}{Ji-Rong Wen}.} \bibinfo{year}{2024}\natexlab{b}.
\newblock \showarticletitle{MemSim: A Bayesian Simulator for Evaluating Memory
  of LLM-based Personal Assistants}.
\newblock \bibinfo{journal}{\emph{arXiv preprint arXiv:2409.20163}}
  (\bibinfo{year}{2024}).
\newblock


\bibitem[Zhao et~al\mbox{.}(2022)]%
        {airship}
\bibfield{author}{\bibinfo{person}{Weijie Zhao}, \bibinfo{person}{Shulong Tan},
  {and} \bibinfo{person}{Ping Li}.} \bibinfo{year}{2022}\natexlab{}.
\newblock \showarticletitle{Constrained approximate similarity search on
  proximity graph}.
\newblock \bibinfo{journal}{\emph{arXiv preprint arXiv:2210.14958}}
  (\bibinfo{year}{2022}).
\newblock


\bibitem[Zhao et~al\mbox{.}(2026)]%
        {amabench}
\bibfield{author}{\bibinfo{person}{Yujie Zhao}, \bibinfo{person}{Boqin Yuan},
  \bibinfo{person}{Junbo Huang}, \bibinfo{person}{Haocheng Yuan},
  \bibinfo{person}{Zhongming Yu}, \bibinfo{person}{Haozhou Xu},
  \bibinfo{person}{Lanxiang Hu}, \bibinfo{person}{Abhilash Shankarampeta},
  \bibinfo{person}{Zimeng Huang}, \bibinfo{person}{Wentao Ni},
  \bibinfo{person}{Yuandong Tian}, {and} \bibinfo{person}{Jishen Zhao}.}
  \bibinfo{year}{2026}\natexlab{}.
\newblock \showarticletitle{AMA-Bench: Evaluating Long-Horizon Memory for
  Agentic Applications}.
\newblock \bibinfo{journal}{\emph{arXiv preprint arXiv:2602.22769}}
  (\bibinfo{year}{2026}).
\newblock


\bibitem[Zhong et~al\mbox{.}(2025)]%
        {lsmvec}
\bibfield{author}{\bibinfo{person}{Shurui Zhong}, \bibinfo{person}{Dingheng
  Mo}, {and} \bibinfo{person}{Siqiang Luo}.} \bibinfo{year}{2025}\natexlab{}.
\newblock \showarticletitle{LSM-VEC: A Large-Scale Disk-Based System for
  Dynamic Vector Search}.
\newblock \bibinfo{journal}{\emph{arXiv preprint arXiv:2505.17152}}
  (\bibinfo{year}{2025}).
\newblock


\bibitem[Zhong et~al\mbox{.}(2024)]%
        {memorybank}
\bibfield{author}{\bibinfo{person}{Wanjun Zhong}, \bibinfo{person}{Lianghong
  Guo}, \bibinfo{person}{Qiqi Gao}, \bibinfo{person}{He Ye}, {and}
  \bibinfo{person}{Yanlin Wang}.} \bibinfo{year}{2024}\natexlab{}.
\newblock \showarticletitle{MemoryBank: Enhancing Large Language Models with
  Long-Term Memory}. In \bibinfo{booktitle}{\emph{{AAAI}}}.
  \bibinfo{publisher}{{AAAI} Press}, \bibinfo{pages}{19724--19731}.
\newblock


\bibitem[Zhou et~al\mbox{.}(2023)]%
        {enterprisesearch}
\bibfield{author}{\bibinfo{person}{Daniel~Xiaodan Zhou}, \bibinfo{person}{Lan
  Liu}, \bibinfo{person}{Anmol Anubhai}, \bibinfo{person}{Maansi Shandilya},
  \bibinfo{person}{Steph Sigalas}, \bibinfo{person}{William~Yang Wang}, {and}
  \bibinfo{person}{Zhiheng Huang}.} \bibinfo{year}{2023}\natexlab{}.
\newblock \showarticletitle{Beyond Accurate Answers: Evaluating Open-Domain
  Question Answering in Enterprise Search}. In
  \bibinfo{booktitle}{\emph{{CHIIR}}}. \bibinfo{publisher}{{ACM}},
  \bibinfo{pages}{308--312}.
\newblock


\bibitem[Zhou et~al\mbox{.}(2025b)]%
        {pacmann}
\bibfield{author}{\bibinfo{person}{Mingxun Zhou}, \bibinfo{person}{Elaine Shi},
  {and} \bibinfo{person}{Giulia Fanti}.} \bibinfo{year}{2025}\natexlab{b}.
\newblock \showarticletitle{Pacmann: Efficient Private Approximate Nearest
  Neighbor Search}. In \bibinfo{booktitle}{\emph{{ICLR}}}.
  \bibinfo{publisher}{OpenReview.net}.
\newblock


\bibitem[Zhou et~al\mbox{.}(2026)]%
        {memtrust}
\bibfield{author}{\bibinfo{person}{Xing Zhou}, \bibinfo{person}{Dmitrii
  Ustiugov}, \bibinfo{person}{Haoxin Shang}, {and} \bibinfo{person}{Kisson
  Lin}.} \bibinfo{year}{2026}\natexlab{}.
\newblock \showarticletitle{MemTrust: A Zero-Trust Architecture for Unified AI
  Memory System}.
\newblock \bibinfo{journal}{\emph{arXiv preprint arXiv:2601.07004}}
  (\bibinfo{year}{2026}).
\newblock


\bibitem[Zhou et~al\mbox{.}(2025a)]%
        {mem1}
\bibfield{author}{\bibinfo{person}{Zijian Zhou}, \bibinfo{person}{Ao Qu},
  \bibinfo{person}{Zhaoxuan Wu}, \bibinfo{person}{Sunghwan Kim},
  \bibinfo{person}{Alok Prakash}, \bibinfo{person}{Daniela Rus},
  \bibinfo{person}{Jinhua Zhao}, \bibinfo{person}{Bryan Kian~Hsiang Low}, {and}
  \bibinfo{person}{Paul~Pu Liang}.} \bibinfo{year}{2025}\natexlab{a}.
\newblock \showarticletitle{MEM1: Learning to Synergize Memory and Reasoning
  for Efficient Long-Horizon Agents}.
\newblock \bibinfo{journal}{\emph{arXiv preprint arXiv:2506.15841}}
  (\bibinfo{year}{2025}).
\newblock


\bibitem[Zhu et~al\mbox{.}(2025)]%
        {compass}
\bibfield{author}{\bibinfo{person}{Jinhao Zhu}, \bibinfo{person}{Liana Patel},
  \bibinfo{person}{Matei Zaharia}, {and} \bibinfo{person}{Raluca~Ada Popa}.}
  \bibinfo{year}{2025}\natexlab{}.
\newblock \showarticletitle{Compass: Encrypted Semantic Search with High
  Accuracy}. In \bibinfo{booktitle}{\emph{{OSDI}}}.
  \bibinfo{publisher}{{USENIX} Association}, \bibinfo{pages}{915--938}.
\newblock


\bibitem[Zuo et~al\mbox{.}(2024)]%
        {serf}
\bibfield{author}{\bibinfo{person}{Chaoji Zuo}, \bibinfo{person}{Miao Qiao},
  \bibinfo{person}{Wenchao Zhou}, \bibinfo{person}{Feifei Li}, {and}
  \bibinfo{person}{Dong Deng}.} \bibinfo{year}{2024}\natexlab{}.
\newblock \showarticletitle{SeRF: Segment Graph for Range-Filtering Approximate
  Nearest Neighbor Search}.
\newblock \bibinfo{journal}{\emph{Proc. {ACM} Manag. Data}}
  \bibinfo{volume}{2}, \bibinfo{number}{1} (\bibinfo{year}{2024}),
  \bibinfo{pages}{69:1--69:26}.
\newblock


\end{thebibliography}

\ifpearl
\else
\appendix
\section{Security Proof}
\label{sec:security-proof}

We prove the security theorem (Theorem~\ref{thm:opal-security}) in the
$\mathcal{G}_{\mathsf{att}}$-hybrid model
(\S\ref{subsec:threat-model-and-security-guarantees} of the main paper), where each enclave is a
stateful black box whose internal state is hidden from the adversary, which sees
only I/O crossing enclave boundaries. It remains to characterize the extra
inter-enclave call trace, show that $\Dream()$ adds no extra I/O, and prove
freshness of the persisted client state. As in prior work~\citep[\S B and \S C]{compass}, we assume no ORAM stash
or capacity overflow. We rely on Compass's batched-access ORAM proof for
two guarantees. First, ORAM executions with the same public trace structure are
computationally indistinguishable~\citep[Thm.~1; Supp.~Lemma~4]{compass}.
Second, tampering with Merkle-authenticated ORAM storage is detected except with
negligible probability~\citep[Thm.~1; Supp.~Lemma~5]{compass}.

For request $q_{i,b}$, let $\mathcal{T}_{i,b}$ denote the \emph{observable
trace}: the ordered sequence of events visible to the adversary during
execution. There are three event types:
\begin{itemize}[leftmargin=*, noitemsep, topsep=2pt]
    \item $\mathsf{t}_{\mathsf{call}}$: one fixed-size inter-enclave call,
    consisting of a padded encrypted request and padded encrypted response on an
    authenticated encrypted channel.
    \item $\mathsf{t}_{\mathsf{oram}}(s,\, \ORAM_X)$: one batched ORAM access of
    public batch size~$s$ to store~$\ORAM_X$.
    \item $\mathsf{t}_{\mathsf{client}}(\ell_q,\ell_r)$: the client's request and the
    response, of lengths $\ell_q,\ell_r$ given by $\mathcal{L}$, on an authenticated
    encrypted channel.
\end{itemize}

\noindent
All enclave-internal computation, including KG traversal, ANN scoring,
reranking, stash manipulation, and $\Dream()$ processing, produces no externally
visible event unless it triggers one of the three event types above.

\begin{definition}[Structural equivalence]
\label{def:structural-equiv}
Two traces $\mathcal{T}$ and $\mathcal{T}'$ are \emph{structurally equivalent},
written $\mathcal{T} \equiv_{\mathsf{s}} \mathcal{T}'$, if:
\begin{enumerate}[leftmargin=*, noitemsep, topsep=2pt]
    \item they have the same number of events;
    \item each position has the same event type; and
    \item whenever the event at a position is $\mathsf{t}_{\mathsf{oram}}$, both
    traces access the same ORAM store with the same public batch size, and
    whenever it is $\mathsf{t}_{\mathsf{client}}$, both have the same lengths.
\end{enumerate}
\end{definition}

\subsection{Condition~1: Indistinguishability}

\begin{lemma}[Dreaming preserves the enclosing trace]
\label{lem:dreaming}
The $\Dream()$ callback preserves the observable trace of the request that runs
it.
\end{lemma}

\begin{proof}
$\Dream()$ runs entirely inside $\OpalEnc$ after the request's ORAM batches. It
acts only on enclave-resident metadata: it expires items whose TTL has passed. The TTL threshold used for expiry is a
deterministic function of public quantities (\S\ref{subsec:adaptive-retention} of the main paper).
Expiry retires the item in enclave-resident metadata, and a later ingestion
overwrites its slot.
None of these steps triggers an extra ORAM access or inter-enclave call, or
changes the store, path, or public batch size of the enclosing ORAM operation.
Therefore the enclosing access contributes the same single event
$\mathsf{t}_{\mathsf{oram}}(s, \ORAM_X)$ as before, so the observable trace is
unchanged.
\end{proof}

\begin{lemma}[Request trace structure]
\label{lem:request-structure}
For any admissible request pair $q_{i,0}$ and $q_{i,1}$ with
$\mathcal{L}(q_{i,0}) = \mathcal{L}(q_{i,1})$, the traces satisfy
$\mathcal{T}_{i,0} \equiv_{\mathsf{s}} \mathcal{T}_{i,1}$.
\end{lemma}

\begin{proof}
Equal leakage fixes the request type, so it suffices to consider the $\Query$
and $\Ingest$ cases separately.

If the pair is $\Query$, only
$\LLMEnc.\allowbreak\Traverse$,
$\EmbEnc.\allowbreak\Embed$, and
$\LLMEnc.\allowbreak\Synthesize$ cross enclave boundaries, yielding three
$\mathsf{t}_{\mathsf{call}}$ events.
The retrieval phases contribute one
$\mathsf{t}_{\mathsf{oram}}(n,\allowbreak \ORAM_{\ANN})$
event and one
$\mathsf{t}_{\mathsf{oram}}(K,\allowbreak \ORAM_{\mathsf{Data}})$
event, where $n, K \in \Pi$, and the client's request and response contribute one
$\mathsf{t}_{\mathsf{client}}$ event whose lengths $\mathcal{L}$ fixes. By Lemma~\ref{lem:dreaming}, $\Dream()$ adds no event. Hence every
$\Query$ trace has the same public structure, determined only by~$\Pi$.

If the pair is $\Ingest$, the ingestion path yields one $\mathsf{t}_{\mathsf{call}}$,
one $\mathsf{t}_{\mathsf{oram}}(1, \ORAM_{\ANN})$, one
$\mathsf{t}_{\mathsf{oram}}(1, \ORAM_{\mathsf{Data}})$, and one $\mathsf{t}_{\mathsf{client}}$ event whose lengths $\mathcal{L}$ fixes. Hence every
$\Ingest$ trace has public structure determined only by public
quantities.

Therefore $\mathcal{T}_{i,0} \equiv_{\mathsf{s}} \mathcal{T}_{i,1}$.
\end{proof}

\begin{proof}[Proof of Condition~1]
Fix a PPT adversary~$\mathcal{A}$ and let $m$ be the polynomially bounded number
of requests. By Lemma~\ref{lem:request-structure}, every admissible request pair
yields structurally equivalent traces.

Condition on any fixed admissible step. By the scope of
Theorem~\ref{thm:opal-security}, no earlier observed ORAM access used this step's
ORAM client state.
For each $\mathsf{t}_{\mathsf{oram}}(s,\allowbreak \ORAM_X)$ event,
the imported Compass batched-access ORAM result
(\citealp[Thm.~1; Supp.~Lemma~4]{compass}) gives computational
indistinguishability because the store and public batch size match. For each
$\mathsf{t}_{\mathsf{call}}$ event, the authenticated encrypted channel
assumption hides the payload and reveals only that one fixed-size inter-enclave
exchange occurred. For the $\mathsf{t}_{\mathsf{client}}$ event, the channel hides the
payload, and admissibility makes the lengths equal. Thus the entire step is computationally indistinguishable
across the two worlds.

Use the standard adaptive sequential hybrid. Let $H_\ell$ ($0 \le \ell \le m$)
be the experiment in which the first $\ell$ admissible requests are answered in
world~1 and the remaining requests in world~0. Then $H_0$ is the game with $b=0$
and $H_m$ is the game with $b=1$.

For adjacent hybrids $H_\ell$ and $H_{\ell+1}$, the common prefix through
step~$\ell$ induces the same distribution over the next admissible request pair.
Applying the fixed-step argument above at step~$\ell{+}1$ gives $|H_\ell -
H_{\ell+1}| \le \mathsf{negl}(\lambda)$. Summing over the polynomially many
steps yields
\[
    \bigl|\Pr[b' = b] - \tfrac{1}{2}\bigr|
    \;\le\;
    m \cdot \mathsf{negl}(\lambda)
    \;=\;
    \mathsf{negl}(\lambda). \qedhere
\]
\end{proof}

\subsection{Condition~2: Integrity and Freshness}

\begin{lemma}[Integrity]
\label{lem:integrity}
Assuming the theorem hypotheses, the probability that the adversary causes the
challenger to accept tampered persisted state is $\mathsf{negl}(\lambda)$.
\end{lemma}

\begin{proof}
As in \S\ref{sssec:freshness-rollback} of the main paper, let $R$ denote the MAC-protected
rollback record authenticating $\mathsf{ctr}$ and both ORAM roots, and let $C$
denote the AEAD-sealed client checkpoint.

Persisted state consists of three authenticated objects: the two ORAM trees, the
rollback record~$R$, and the sealed client checkpoint~$C$.

For $\ORAM_{\ANN}$ and $\ORAM_{\mathsf{Data}}$, the Compass
ORAM integrity result
(\citealp[Thm.~1; Supp.~Lemma~5]{compass})
detects tampering except with negligible probability.
For~$R$, any accepted modification requires forging the MAC on its
canonically encoded tuple under~$k_{\mathsf{mac}}$. For~$C$, any accepted
modification requires breaking the authenticity of the AEAD
under~$k_{\mathsf{enc}}$.

Therefore any accepted tampering of persisted state implies either a failure of
the imported Compass batched-access ORAM integrity result on one of the ORAM
trees, a MAC forgery on~$R$, or an AEAD forgery on~$C$. By the theorem
hypotheses and a union bound over polynomially many operations, this occurs
with probability at most $\mathsf{negl}(\lambda)$.
\end{proof}

\begin{lemma}[Freshness]
\label{lem:freshness}
Assuming the theorem hypotheses, the probability that the adversary causes the
challenger to accept a replayed client request or stale persisted state for
the current request is $\mathsf{negl}(\lambda)$.
\end{lemma}

\begin{proof}
Let $\kappa_i \le i$ denote the latest request index whose client state has
been checkpointed for the logical client by step~$i$ under the provider's
residency and checkpoint policy; in the security game, $\kappa_i$ is an
exogenous public schedule, independent of user communication patterns and fixed
across the two worlds.

Before any ORAM access, the enclave accepts a client request only if its
counter~$\mathsf{ctr}$ is exactly one more than the counter in~$R$.
Therefore any client request whose counter was already committed is rejected.

For the ORAM-backed storage path, the MAC in~$R$ authenticates the current
counter~$\mathsf{ctr}$ together with the expected ORAM roots. Therefore any
accepted stale storage for the current request implies either acceptance of a
stale counter, a failure of the imported Compass batched-access ORAM integrity
result on one of the ORAM trees, or a MAC forgery on~$R$. The first case is
impossible by the monotonic counter check, and the latter two occur with
probability at most
$\mathsf{negl}(\lambda)$. The sealed checkpoint~$C$ binds the same counter, so
accepting a stale~$C$ likewise requires a stale counter or an AEAD forgery on~$C$.
\end{proof}

\begin{proof}[Proof of Condition~2]
In the $\mathcal{G}_{\mathsf{att}}$-hybrid model, any accepted incorrect
execution must come from one of two sources: accepted tampering of persisted
state, or acceptance of a replayed client request or stale persisted state for
the current request. Lemma~\ref{lem:integrity} rules out the first except with
$\mathsf{negl}(\lambda)$, and Lemma~\ref{lem:freshness} rules out the second
except with $\mathsf{negl}(\lambda)$. Therefore
Condition~2 holds except with $\mathsf{negl}(\lambda)$.
\end{proof}

\section{Synthetic Data Pipeline Details}
\label{sec:supp-synthetic}

\subsection{Hawkes Process Parameterization}
\label{subsec:supp-hawkes}

The conditional intensity of modality~$m$ at time~$t$ is:
\begin{equation}
\begin{split}
    \lambda_m(t \mid \mathcal{H}_t)
    &=
    \mathbf{1}\!\bigl[\gamma_m^{s(t)} > 0\bigr]
    \biggl(
        \gamma_m^{s(t)}\,\mu_m \\
    &\qquad+
        \sum_{m' \in \mathcal{M}}
        \sum_{\substack{t_i < t \\ m_i = m'}}
        \phi_{m' \to m}(t - t_i)
    \biggr),
\end{split}
    \label{eq:hawkes}
\end{equation}
where $\gamma_m^{s(t)}$ is the life-state multiplier for modality~$m$ in
state~$s(t)$ (normalized so that the weekday-averaged modulation is approximately neutral), $\mu_m$
is the baseline rate, and
\[
\phi_{m' \to m}(u)
=
\alpha_{m' \to m}\,\beta_m\,e^{-\beta_m u}\,\mathbf{1}[u > 0]
\]
is the excitation kernel with branching ratio $\alpha_{m' \to m}$ and
destination-specific decay rate $\beta_m$. The decay rate is chosen to be
destination-specific rather than source-specific. Offspring
arrive at the natural pace of the destination modality 
regardless of what triggered them (e.g., a meeting-triggered 
email arrives at email pace), better reflecting how quickly 
each modality is acted upon in practice. The exponential kernel is chosen over other alternatives as it 
admits an $O(1)$ recursive update and produces interpretable 
branching ratios, though it captures only a single response 
timescale per modality.

The $6 \times 6$ branching matrix is sparse by design: document, query, and ambient
modalities act as pure sinks ($\alpha_{d \to m} = \alpha_{q \to m} = \alpha_{\mathit{amb} \to m} = 0$ for all $m$). Self-excitation parameters are empirically grounded: $\alpha_{e \to e} =
0.40$ is derived from the IkeNet reply rate~\cite{Fox02042016}, adjusted
downward for the representativeness gap between intra-enterprise email and a
full mixed inbox, and further constrained by the endogenous budget reserved for meeting- 
and message-triggered email. $\alpha_{\mathit{msg} \to \mathit{msg}} = 0.54$ is set below
the value of $0.70$ derived from observed messaging thread
lengths~\cite{groupchat}, reserving endogenous budget for cross-modal paths.

\subsection{Recency-Weighted Question Sampling}
\label{subsec:recency-sup}
For question sampling, artifact selection is weighted by a 
two-phase recency boost $w(t) = C_1 e^{-\nu t} + 
C_2\,(t{+}1)^{-\delta}$, combining a fast exponential 
component for short-term recency and a slower power-law 
tail for long-term access. The functional form is 
motivated by empirical observations: personal 
information access distributions show a steep recency concentration (6.6\% same-day, 21.9\%
within a week, 45.9\% within a month)~\cite{stuffiveseen}, and collective memory of significant events follows the same two-phase decay, transitioning from exponential to power-law decay around 10 days after an event \cite{igarashi2022}. Parameters are chosen to approximate the empirical bucket 
distribution~\cite{stuffiveseen}: $C_1 = 1.45$, $\nu = 0.20$, 
$C_2 = 1.0$, $\delta = 0.68$.

\begin{table*}[!ht]
\centering
\caption{Empirical sources used to motivate the synthetic data generator and how each source informs the workload design.}
\label{tab:calibration}
\fontsize{7.4}{8.5}\selectfont
\setlength{\tabcolsep}{2pt}
\renewcommand{\arraystretch}{1.24}
\begin{tabular}{@{}p{0.50\textwidth} p{0.45\textwidth}@{}}
\toprule
\textbf{Research Finding} & \textbf{How Used} \\
\midrule

\multicolumn{2}{@{}l@{}}{\textit{Email}} \\[2pt]

\raggedright Mean reply rate of 45.6\% in the IkeNet email
network~\cite{Fox02042016} & Email self-excitation set to 0.40, adjusted down
from 0.45 for inbox noise absent in IkeNet and cross-modal budget reserved for
meeting- and message-triggered emails \\

\raggedright Reply time distribution: median 47 min, mean 1157
min~\cite{ageofemail} & Email triggering influence decays with 47-minute
half-life \\

\raggedright $\sim$50\% of corporate emails contain no actionable
task~\cite{emailintent} & 60\% of emails treated as noise \\

\raggedright 8.4\% of institutional emails are organizational bulk
emails~\cite{bulkemails} & 15\% of noise emails are organizational bulk emails
such as newsletters \\

\raggedright $\sim$15\% of Enron emails classified as
gossip~\cite{gossipworkplace} & 25\% of noise emails are gossip \\

\raggedright Enron corpus includes greeting/social email
category~\cite{emailintent} & 15\% of noise emails are short messages and
greetings \\

\midrule
\multicolumn{2}{@{}l@{}}{\textit{Documents}} \\[2pt]

\raggedright 68\% of enterprise data is never actively used after
creation~\cite{seagate2020} & 70\% of documents targeted as noise \\

\raggedright 12-27 important documents per day for operational
roles~\cite{celiveo2026} & 12 signal documents targeted per day \\

\midrule
\multicolumn{2}{@{}l@{}}{\textit{Queries}} \\[2pt]

\raggedright Enterprise intranet: 90K-120K queries from 18K-24K users on a
typical day~\cite{enterprisesearch} &  $\sim$4 queries per person
per workday used as the daily query target \\

\midrule
\multicolumn{2}{@{}l@{}}{\textit{Messages}} \\[2pt]

\raggedright Mean 41.5 personal messages per day (2011)~\cite{pewmessage} &
$\sim$85 personal messages per day after upward adjustment for smartphone era \\

\raggedright Average worker receives 153 Teams messages per
workday~\cite{microsoft2025infiniteworkday} & Used as the baseline for workplace
received messages per day \\

\raggedright Average worker sends 92 Slack messages per day~\cite{lee2025slack}
& Used as the baseline for workplace sent messages; combined with received
messages from Teams data brings total workplace volume to 245/day \\

\raggedright Project channels average 16.4 words per message; IT support 23.0
words~\cite{groupchat} & 20 words targeted per message in generation \\

\raggedright Project threads: $S{=}2.0$; IT support threads: $S{=}4.1$ messages
per thread~\cite{groupchat} & Message self-excitation branching ratio derived
from observed thread sizes, adjusted down to reserve budget for meeting- and
email-triggered messages \\

\midrule
\multicolumn{2}{@{}l@{}}{\textit{Meetings}} \\[2pt]

\raggedright 17.1 meetings/week; 27.3\% unplanned and
spontaneous~\cite{reclaim2024} & 3.4 meetings per day; $\sim$30\% are endogenously
triggered by prior events \\

\raggedright 31\% of meetings involve email multitasking; 24\% involve file
activity~\cite{multitasking} & Meetings trigger follow-up emails and documents
\\

\raggedright 69.7\% of remote meeting participants use parallel
chat~\cite{parallelchat} & Meetings trigger parallel messaging activity; used to
motivate meeting-to-message cross-excitation \\

\raggedright Speaking rate 120-200 words per minute~\cite{raynerwpm} & Meeting
transcripts prompted at 100 words per minute \\

\midrule
\multicolumn{2}{@{}l@{}}{\textit{Ambient conversation}} \\[2pt]

\raggedright 7.4 social interactions of 10+ minutes per
day~\cite{Wheeler1977SexDI} & 8 ambient conversations per day \\

\ifdefined\mlsyssubmission
\bottomrule
\end{tabular}
\end{table*}
\begin{table*}[t]
\centering
\caption{Empirical sources used to motivate the synthetic data generator and how each source informs the workload design (continued).}
\label{tab:calibration-cont}
\fontsize{7.4}{8.5}\selectfont
\setlength{\tabcolsep}{2pt}
\renewcommand{\arraystretch}{1.24}
\begin{tabular}{@{}p{0.50\textwidth} p{0.45\textwidth}@{}}
\toprule
\textbf{Research Finding} & \textbf{How Used} \\
\fi
\midrule
\multicolumn{2}{@{}l@{}}{\textit{Life-state multipliers}} \\[2pt]

\raggedright ChatGPT weekend usage $\sim$2/3 of weekday
levels~\cite{genai2024worldbank} & Search query rate drops to 2/3 on weekends \\

\raggedright 50+ Teams messages sent or received outside work
hours~\cite{microsoft2025infiniteworkday} & Message rate kept near baseline in
the \textsc{FREE\_EVENING} life-state to reflect continued work messaging
alongside personal messaging activity \\

\raggedright No sharp peak in per-hour query distribution~\cite{ShareChat} &
Query intensity is kept roughly uniform across active states \\

\raggedright PowerPoint edits spike 122\% before meetings; peak productivity
9-11 am and 1-3 pm~\cite{microsoft2025infiniteworkday} & Motivates elevated
document and meeting multipliers in the \textsc{FOCUS\_TIME} life-state cycle
block covering 9-11 am and 1-3 pm \\

\raggedright $\sim$60\% of commuters use cell phones for reading or
writing~\cite{drivetext} & Email, search, and message rates at 60\% of baseline
during commuting \\

\midrule
\multicolumn{2}{@{}l@{}}{\textit{Social graph}} \\[2pt]

\raggedright Humans form social groups at preferred sizes of $\sim$5, 15, 50,
150~\cite{dunbarnumber} & Social graph built with 5-person close circle and
15-person regular contact layer \\

\raggedright 58\% of social energy allocated to closest 5
relationships~\cite{plosone} & Contact frequency in sampling weighted toward
inner social circle \\

\raggedright Layered friendship structure confirmed in mobile phone data with
consistent scaling ratio between layers~\cite{dunbarnumber} & Validates the
geometric scaling between social graph layers used in contact sampling \\

\midrule
\multicolumn{2}{@{}l@{}}{\textit{Temporal boost for question sampling}} \\[2pt]

\raggedright Probability that a memory will be needed decays as a power function
of recency, consistent across environmental corpora~\cite{anderson1991} &
Establishes power-law decay as the theoretically grounded functional form for
recency weighting in query sampling \\

\raggedright Personal desktop information search: 6.6\% same-day, 21.9\% within
a week, 45.9\% within a month~\cite{stuffiveseen} & Power-law tail exponent fitted directly from empirical access frequency data \\

\raggedright Office document accesses: 67.4\% same-day, 20.1\% one day old,
8.6\% three days old~\cite{xu2020} & Confirms same-day concentration; validates
heavy weighting of recent data \\

\raggedright Collective memory decays in two phases: fast exponential drop in
the first $\sim$10 days, then a power-law tail~\cite{igarashi2022} & Motivates
two-regime recency curve: a fast exponential short-term component plus a slower
power-law tail \\

\bottomrule
\end{tabular}
\end{table*}

\fi

\end{document}